\newif\ifarxiv
\arxivtrue

\ifarxiv
\documentclass[sigplan,screen,nonacm]{acmart}
\else
\documentclass[sigplan,screen]{acmart}
\fi

\ifarxiv
\setcopyright{none}
\else
\setcopyright{rightsretained}
\acmPrice{}
\acmDOI{10.1145/3519939.3523433}
\acmYear{2022}
\copyrightyear{2022}
\acmSubmissionID{pldi22main-p61-p}
\acmISBN{978-1-4503-9265-5/22/06}
\acmConference[PLDI '22]{Proceedings of the 43rd ACM SIGPLAN International Conference on Programming Language Design and Implementation}{June 13--17, 2022}{San Diego, CA, USA}
\acmBooktitle{Proceedings of the 43rd ACM SIGPLAN International Conference on Programming Language Design and Implementation (PLDI '22), June 13--17, 2022, San Diego, CA, USA}
\fi

\usepackage{booktabs}
\usepackage{subcaption}
\usepackage{tikz}
\usepackage{amsmath}
\usepackage{physics}
\usepackage{xspace}
\usepackage{appendix}
\usepackage{cleveref}
\usepackage[noend]{algpseudocode}
\usepackage{algorithm}
\usepackage{color}
\usepackage{graphicx}
\usepackage{multirow}
\usepackage{soul}
\usepackage{booktabs}
\usepackage{threeparttable}
\usepackage{balance}
\usepackage{siunitx}
\usepackage{xfrac}
\usepackage{rotating,tabularx}
\usepackage{rotating}
\usepackage{diagbox}

\newcommand{\KwTo}{\textbf{to}\xspace}

\newcommand{\sys}{Quartz\xspace}
\newcommand{\regen}{\textproc{RepGen}\xspace}
\newcommand{\qiskit}{Qiskit\xspace}
\newcommand{\quilc}{Quilc\xspace}
\newcommand{\nam}{Nam et al.\xspace}
\newcommand{\tket}{t$|$ket$\rangle$\xspace}
\newcommand{\voqc}{\textproc{voqc}\xspace}
\newcommand{\ecc}{ECC\xspace}
\newcommand{\eccs}{ECCs\xspace}
\newcommand{\eccfy}{\Call{Eccify}}
\newcommand{\ccs}{ECC set\xspace}
\newcommand{\ccss}{ECC sets\xspace}

\newcommand{\ngates}{n\xspace}
\newcommand{\er}[1]{\mbox{\rm\em #1}}
\newcommand{\m}[1]{\mathcal{#1}}
\newcommand{\sem}[1]{\llbracket {#1} \rrbracket}
\newcommand{\precede}{\prec}

\newcommand{\nq}{$(n,q)$}

\newcommand{\jq}{$(j,q)$}
\newcommand{\jminusoneq}{$(j{-}1,q)$}

\newcommand{\zeroq}{$(0,q)$}
\newcommand{\oneq}{$(1,q)$}

\newcommand{\cnqm}{\m{C}^{(n,q,m)}}
\newcommand{\cnq}{\m{C}^{(n,q)}}

\newcommand{\coneq}{\m{C}^{(1,q)}}

\newcommand{\emptyseq}{()}

\newcommand{\dropfirst}[1]{\Call{DropFirst}{#1}}
\newcommand{\droplast}[1]{\Call{DropLast}{#1}}

\newcommand{\tablehl}[1]{\hl{\textbf{#1}}}
\newcommand{\uscore}{\mbox{\tt\char`\_}}

\newcommand{\ZJ}[1]{{\textcolor{red}{ZJ: #1}}}
\newcommand{\TODO}[1]{{\textcolor{red}{TODO: #1}}}
\newcommand{\MX}[1]{{\textcolor{cyan}{MX: #1}}}
\newcommand{\ur}[1]{{\textcolor{purple}{UA: #1}}}
\newcommand{\Alex}[1]{{\textcolor{blue}{AA: #1}}}
\newcommand{\oded}[1]{{\textcolor{olive}{OP: #1}}}

\newcommand{\captionvspace}{0em}
\newcommand{\pptscale}{0.4} %
\newcommand{\evalfigfrac}{0.99} %
\newcommand{\introfigfrac}{1} %

\newcommand{\commentout}[1]{\ifdim\lastskip>0pt\ignorespaces\fi}
\iftrue
  \renewcommand{\ZJ}[1]{\ifdim\lastskip>0pt\ignorespaces\fi}
 \renewcommand{\TODO}[1]{\ifdim\lastskip>0pt\ignorespaces\fi}
 \renewcommand{\MX}[1]{\ifdim\lastskip>0pt\ignorespaces\fi}
 \renewcommand{\ur}[1]{\ifdim\lastskip>0pt\ignorespaces\fi}
 \renewcommand{\Alex}[1]{\ifdim\lastskip>0pt\ignorespaces\fi}
 \renewcommand{\oded}[1]{\ifdim\lastskip>0pt\ignorespaces\fi}
\fi

\newtheorem{definition}{Definition}
\newtheorem{theorem}{Theorem}
\newcommand{\uone}{\texttt{$U_1$}\xspace}
\newcommand{\utwo}{\texttt{$U_2$}\xspace}
\newcommand{\uthree}{\texttt{$U_3$}\xspace}
\newcommand{\cnot}{\texttt{$CNOT$}\xspace}

\newcommand{\rz}{\texttt{$R_z(\lambda)$}\xspace}
\newcommand{\cz}{\texttt{$CZ$}\xspace}

\begin{document}

\ifarxiv
\title{\sys: Superoptimization of Quantum Circuits (Extended Version)} %
\titlenote{This is the extended version of a paper presented in PLDI 2022 \cite{pldi2022-quartz}.
This version includes an additional appendix with detailed results.} %
\fi

\author{Mingkuan Xu}
\affiliation{\institution{Carnegie Mellon University}
\city{Pittsburgh}
\state{PA}
\country{USA}}
\email{mingkuan@cmu.edu}

\author{Zikun Li}
\affiliation{\institution{University of California, Los Angeles}
\city{Los Angeles}
\state{CA}
\country{USA}}
\email{lizikun_@outlook.com}

\author{Oded Padon}
\affiliation{\institution{VMware Research}
\city{Palo Alto}
\state{CA}
\country{USA}}
\email{oded.padon@gmail.com}

\author{Sina Lin}
\affiliation{\institution{Microsoft}
\city{Mountain View}
\state{CA}
\country{USA}}
\email{silin@microsoft.com}

\author{Jessica Pointing}
\affiliation{\institution{University of Oxford}
\city{Oxford}
\country{United Kingdom}}
\email{jessica.pointing@physics.ox.ac.uk}

\author{Auguste Hirth}
\affiliation{\institution{University of California, Los Angeles}
\city{Los Angeles}
\state{CA}
\country{USA}}
\email{ahirth@g.ucla.edu}

\author{Henry Ma}
\affiliation{\institution{University of California, Los Angeles}
\city{Los Angeles}
\state{CA}
\country{USA}}
\email{hetm@g.ucla.edu}

\author{Jens Palsberg}
\affiliation{\institution{University of California, Los Angeles}
\city{Los Angeles}
\state{CA}
\country{USA}}
\email{palsberg@cs.ucla.edu}

\author{Alex Aiken}
\affiliation{\institution{Stanford University}
\city{Stanford}
\state{CA}
\country{USA}}
\email{aiken@cs.stanford.edu}

\author{Umut A. Acar}
\affiliation{\institution{Carnegie Mellon University}
\city{Pittsburgh}
\state{PA}
\country{USA}}
\email{umut@cs.cmu.edu}

\author{Zhihao Jia}
\affiliation{\institution{Carnegie Mellon University}
\city{Pittsburgh}
\state{PA}
\country{USA}}
\email{zhihao@cmu.edu}

\renewcommand{\shortauthors}{M. Xu et al.}
\begin{abstract}
Existing quantum compilers optimize quantum circuits by applying circuit transformations designed by experts.
This approach requires significant manual effort to design and implement circuit transformations for different quantum devices, which use different gate sets, and can miss  optimizations that are hard to find manually.
We propose \sys, a quantum circuit superoptimizer that {\em automatically} generates and verifies circuit transformations for arbitrary quantum gate sets.
For a given gate set, \sys generates candidate circuit transformations by systematically exploring small circuits and verifies the discovered transformations using an automated theorem prover.
To optimize a quantum circuit, \sys uses a cost-based backtracking search that applies the verified transformations to the circuit. %
Our evaluation on three popular gate sets shows that \sys can effectively generate and verify transformations for different gate sets. The generated transformations cover manually designed transformations used by existing optimizers and also include new transformations.  
\sys is therefore able to optimize a broad range of circuits for diverse gate sets, outperforming or matching the performance of hand-tuned circuit optimizers.

\end{abstract}



\maketitle

\ifarxiv
\newpage %
\fi

\section{Introduction}
\label{sec:intro}

Quantum computing comes in many shapes and forms.
There are over a dozen proposals for realizing quantum computing in practice, and nearly all these proposals support different kinds of quantum operations, i.e., instruction set architectures (ISAs).
The increasing diversity in quantum processors makes it challenging to design optimizing compilers for quantum programs, since the compilers must consider a variety of ISAs and carry optimizations specific to different ISAs.

To reduce the execution cost of a quantum circuit, the most common form of optimization is {\em circuit transformations} that substitute a subcircuit matching a specific pattern with a functionally equivalent new subcircuit with improved performance (e.g., using fewer quantum gates).
Existing quantum compilers generally rely on circuit transformations manually designed by experts and applied greedily.
For example, Qiskit~\cite{qiskit} and \tket~\cite{tket} use greedy rule-based strategies to optimize a quantum circuit and perform circuit transformations whenever applicable.
\voqc~\cite{VOQC} formally verifies circuit transformations but still requires users manually specify them.
Although rule-based transformations can reduce the cost of a quantum
circuit, they have two key limitations.

First, because existing optimizers rely on domain experts to design
transformations, they require significant human effort and may also miss subtle optimizations that are
hard to discover manually, resulting in sub-optimal performance.

Second, circuit transformations designed for one quantum device do not directly apply to other devices with different ISAs,
which is problematic in the emerging diverse quantum computing landscape.
For example, IBMQX5~\cite{dumitrescu2018cloud} supports the \uone, \utwo, \uthree and \cnot gates, while Rigetti Agave~\cite{rigetti_agave} supports the
$R_x(\pm\frac{\pi}{2})$, $R_x(\pi)$, \rz, and \cz gates.
As a result, circuit transformations tailored for IBMQX5 cannot optimize circuits on Rigetti Agave, and vice versa.

Recently, Quanto~\cite{quanto} proposed to automatically discover transformations by computing concrete matrix representations of circuits.
Its main restriction is that it does not discover symbolic transformations, which are needed to deal with common \emph{parametric} quantum gates in a general way.

This paper presents \sys, a quantum circuit superoptimizer that automatically generates and verifies symbolic circuit transformations for arbitrary gate sets, including parametric gates.
\sys provides two key advantages over existing quantum circuit optimizers. 
First, for a given set of gates, \sys generates symbolic
circuit transformations and formally verifies their correctness in a
fully {\em automated} way, without any manual effort to design or implement transformations.
Second, \sys explores a more comprehensive set
of circuit transformations by discovering {\em all} possible
transformations up to a certain size, outperforming existing
optimizers with manually designed transformations.

\paragraph{ECC sets.}
We introduce {\em equivalent circuit classes} (\eccs) as a compact way to represent circuit transformations. Each \ecc is a set of functionally equivalent circuits, and two circuits from an \ecc form a valid transformation.
We say that a transformation is \emph{subsumed} by an {\em \ccs} (a set of \eccs) if the transformation can be decomposed into a sequence of transformations, each of which is 
a pair of circuits from the same \ecc in the \ccs.
We use {\em \nq-completeness} to assess the comprehensiveness of an \ccs---an \ccs is \nq-complete if it subsumes {\em all} valid transformations between circuits with at most $n$ gates and $q$ qubits.

\begin{figure}
    \centering
    \includegraphics[width=\introfigfrac\linewidth]{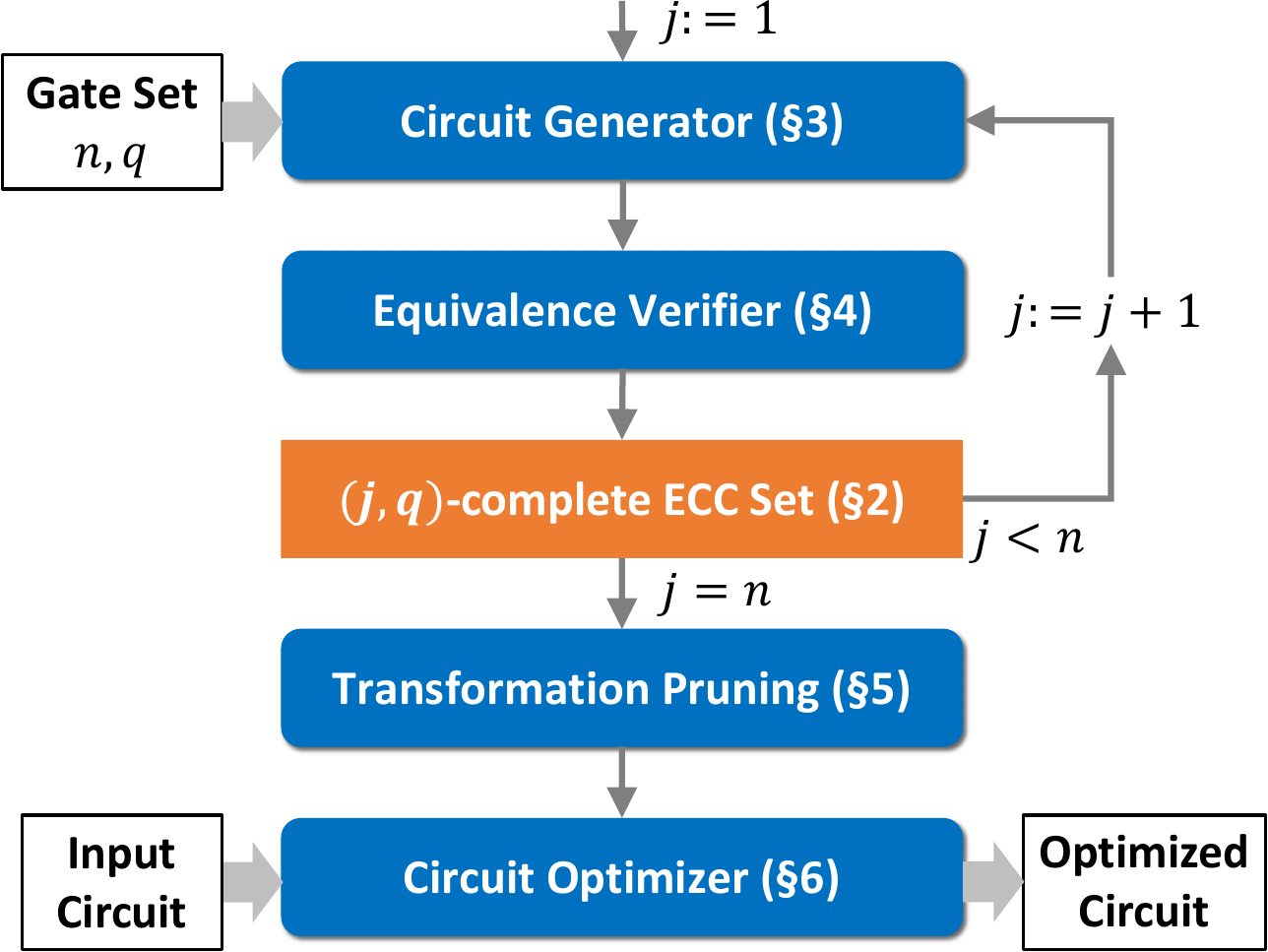}
    \vspace{\captionvspace}
    \caption{\sys overview.
    }
    \label{fig:overview}
\end{figure}

\begin{sloppypar}
\paragraph{Overview.}
\Cref{fig:overview} shows an overview of \sys, which uses an {\em interleaving} approach: it iteratively generates candidate circuits, eliminates redundancy, and verifies equivalences.
In the $j$-th iteration, \sys generates a \jq-complete \ccs based on the \jminusoneq-complete \ccs  from the previous iteration.
The generated \ccs may contain redundant transformations. %
We introduce \regen{}, a representative-based circuit generation algorithm 
that uses a \jminusoneq-complete \ccs
to generate circuits for a \jq-complete \ccs with fewer redundancies.
The circuits are sent to the {\em circuit equivalence verifier},
which formally verifies equivalence between circuits
and produces a \jq-complete \ccs.
After generating an \nq-complete \ccs,
\sys employs several pruning techniques to further eliminate redundancies. %
Finally, \sys's {\em circuit optimizer} applies the discovered transformations to optimize an input circuit.
\end{sloppypar}

\paragraph{Circuit Generator}
Given a gate set and a circuit size $n$,
\sys's {\em circuit generator} generates candidate circuits of size at most $n$ using the \regen{} algorithm, which
avoids generating all possible circuits (of which there are exponentially many)
while ensuring \nq-completeness.
To this end,
\regen{} iteratively constructs \ccss, from
smaller to larger. %
For each \ecc, \regen{} selects a \emph{representative circuit} and
constructs larger circuits by extending these representatives.

To discover equivalences between circuits, \regen{} uses random inputs to assign a {\em fingerprint} (i.e., a hash) to each circuit and checks only the circuits with the same
fingerprint.
We prove an upper bound on the running time of \regen{}
in terms of the number of representatives generated.
For the gate sets considered in our evaluation, \regen{} reduces the number of circuits in an \ccs by one to three orders of magnitudes while maintaining \nq-completeness.

\begin{sloppypar}
\paragraph{Circuit Equivalence Verifier}
\sys's {\em circuit equivalence verifier} checks if two potentially equivalent circuits are indeed functionally equivalent.
A major challenge is dealing with gates that take one or multiple parameters (e.g., \uone, \utwo, and \uthree in IBMQX5, and $R_z$ in Rigetti Agave).
For candidate equivalent circuits, \sys checks whether they are functionally equivalent for %
arbitrary combinations of parameter assignments and quantum states.
To this end, \sys computes symbolic matrix representations of the circuits.
The resulting verification problem involves trigonometric functions and, in the general case, a quantifier alternation; 
\sys soundly eliminates both and reduces circuit equivalence checking to SMT solving
for quantifier-free formulas over the theory of nonlinear real arithmetic.
The resulting SMT queries are efficiently solved by the Z3~\cite{z3} SMT solver.
\end{sloppypar}

\paragraph{Circuit Pruning}
Having generated an $(n, q)$-complete \ccs, \sys optimizes
circuits by applying the transformations specified by the \ccs.
To improve the efficiency of this optimization step, described next,
\sys applies several pruning techniques to eliminate redundant transformations.

\paragraph{Circuit Optimizer}
\sys's {\em circuit optimizer} uses a cost-based backtracking search algorithm adapted from TASO~\cite{jia2019taso} to apply the verified transformations. 
The search is guided by a cost model that compares the performance of different candidate circuits (in our experiments the cost is given by number of gates).
\sys targets the \emph{logical optimization} stage in quantum circuit compilation. That is, \sys operates before {\em qubit mapping} where logical qubits are mapped to physical qubits while respecting hardware constraints~\cite{dhj+magic-2018,wdd+tilt-2021}.

\paragraph{Evaluation}
Our evaluation on three gate sets derived from existing quantum processors shows that \sys can generate and verify circuit transformations for different gate sets in under 30 minutes (using 128 cores).
For logical circuit optimization,
\sys matches and often outperforms existing optimizers.
On a benchmark of 26 circuits, \sys obtains average gate count reductions of 29\%,
30\%, and 49\% for the Nam, IBM, and Rigetti gate sets; 
the corresponding reductions by existing optimizers are 27\%, 23\%, and 39\%.

\section{Symbolic Quantum Circuits}
\label{sec:ecc}

To support parametric gates, \sys introduces \emph{symbolic} quantum circuits and circuit transformations. The latter are represented compactly using \emph{equivalent circuit classes} (\eccs).
This section introduces these concepts.

\paragraph{Quantum circuits.}

Quantum programs are represented as {\em quantum circuits}~\cite{Nielsen00}, as shown in \Cref{fig:conventional_circuit},
where each horizontal wire represents a {\em qubit},
and boxes on these wires represent {\em quantum gates}.
The semantics of a quantum circuit over $q$ qubits is given by a $2^q \times 2^q$ unitary complex matrix.
This matrix can be computed from matrices of individual gates
in a compositional manner, using matrix multiplications (for sequential composition of subcircuits that operate on the same qubits)
and tensor products (for parallel composition of subcircuits that operate on different qubits).
For example, the matrix for the circuit of \Cref{fig:conventional_circuit}
is
$(CNOT \otimes I) \cdot (U_2(\frac{\pi}{2},\pi) \otimes CNOT) \cdot (U_1(-\pi) \otimes H \otimes H)$.

A circuit $C'$ is a \emph{subcircuit} of $C$ if, for some qubit permutation,
the matrix computation for $C$ can be structured as
$\uscore \cdot (M_{C'} \otimes I \otimes \cdots \otimes I) \cdot \uscore$,
where $M_{C'}$ is the matrix for $C'$.
For example, the green box in \Cref{fig:conventional_circuit} highlights a subcircuit,
while the red dashed area is not a subcircuit.
The subcircuit notion is invariant under qubit permutation; e.g.,
the $X$ and $U_1$ gates in \Cref{fig:conventional_circuit} also form a subcircuit.
A circuit's matrix is invariant under 
replacing one subcircuit with another that has the same matrix (but possibly different gates),
which underpins peephole optimization for quantum circuits.

Many gates supported by modern quantum devices take real-valued parameters.
For example, the IBM quantum device supports the $U_1$ gate which takes one parameter and \nolinebreak rotates a qubit about the $x$-axis (on the Bloch sphere), and the $U_2$ gate which takes two parameters for rotating about the $x$- and $z$-axes.
The matrix representations of $U_1$ and $U_2$ are:
\begin{equation}
\label{eqn:u1}
U_1(\theta) = 
\begin{pmatrix}
1 & 0\\
0 & e^{i\theta}
\end{pmatrix}
\quad
U_2(\phi,\lambda) = \frac{1}{\sqrt{2}}
\begin{pmatrix}
1 & -e^{i\lambda}\\
e^{i\phi} & e^{i(\phi+\lambda)}
\end{pmatrix}
\end{equation}

\begin{figure}
    \centering
    \begin{subfigure}{0.48\linewidth} 
    \centering
    \includegraphics[scale=\pptscale]{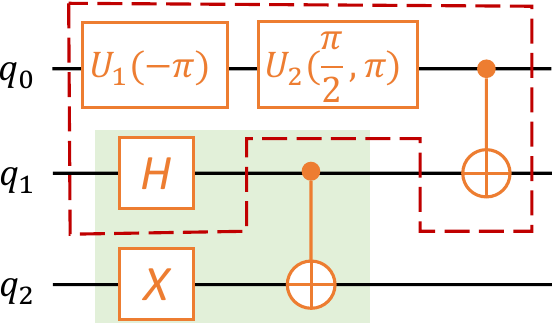}
    \caption{Quantum circuit.}
    \label{fig:conventional_circuit}
    \end{subfigure}
    \hfill
    \begin{subfigure}{0.48\linewidth}
    \centering
    \vspace{4pt} %
    \includegraphics[scale=\pptscale]{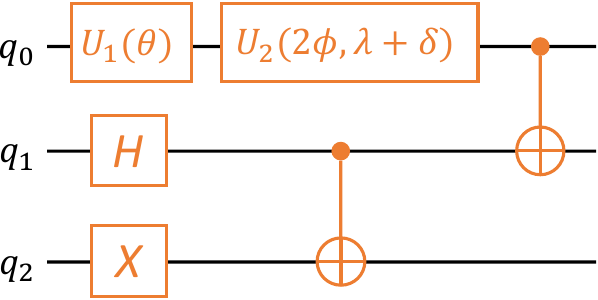}
    \caption{Symbolic Quantum circuit.}
    \label{fig:our_circuit}
    \end{subfigure}
    \\
    \vspace{1em}
    \begin{subfigure}{\linewidth}
    \centering
    \includegraphics[scale=\pptscale]{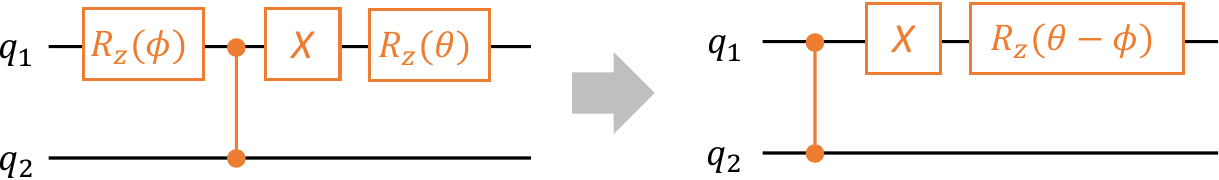}
    \caption{Transformation fusing two $R_z$ gates across $CZ$ and $X$ gates.} 
    \label{fig:transformation_example}
    \end{subfigure}
    \caption{Quantum circuits and transformations.
    The green highlighted box in (a) is a subcircuit, while the red dashed area is not;
    the $U_1(-\pi)$ and $X$ gates also form a subcircuit.
    }
    \label{fig:circuit}
\end{figure}

\paragraph{Symbolic circuits.}
To support superoptimization of circuits with parametric gates,
\sys discovers transformations between \emph{symbolic quantum circuits}, as shown in 
\Cref{fig:our_circuit}, which include (symbolic) parameters ($\theta$, $\phi$, $\lambda$, $\delta$, etc.) and arithmetic operations on these parameters, and are formalized below.
Using such circuits,
\sys can represent transformations such as the one illustrated in \Cref{fig:transformation_example}.

The semantics of a \emph{symbolic quantum circuit}, 
denoted $\sem{\cdot}$, has type 
$\sem{C} \colon \mathbb{R}^{m} \rightarrow \mathbb{C}^{2^q\times 2^q}$
where $C$ is a circuit over $m$ (symbolic) parameters and $q$ qubits.
For a vector of parameter values  $\vec{p}\in\mathbb{R}^m$, 
$\sem{C}(\vec{p})$ is a
$2^q \times 2^q$ unitary complex matrix representing a (concrete) quantum circuit over $q$ qubits.
For example, \cref{eqn:u1} can be seen as defining the semantics of $U_1$ and $U_2$ as single-gate symbolic quantum circuits.
The semantics of a multi-gate symbolic circuit (e.g., \Cref{fig:our_circuit}) 
is derived from that of single-gate circuits using matrix multiplications and tensor products
exactly as for concrete circuits.
Henceforth, we use \emph{circuits} to mean symbolic quantum circuits.

\paragraph{Circuit equivalence and transformations.}
In quantum computing, 
the states $\ket{\psi}$ and $e^{i\beta}\ket{\psi}$
($\beta\in\mathbb{R}$)
are \emph{equivalent up to a global phase},
and from an observational point of view they are identical~\cite{Nielsen00}.
This leads to the following circuit-equivalence definition that underlies \sys's optimization.

\begin{definition}[Circuit Equivalence]
Two symbolic quantum circuits $C_1$ and $C_2$ are \emph{equivalent} if:
\begin{equation}
\label{eqn:equiv}
\forall \vec{p} \in \mathbb{R}^m. \; \exists \beta \in \mathbb{R}. \; 
\sem{C_1}(\vec{p}) = e^{i\beta} \sem{C_2}(\vec{p}).
\end{equation}
\end{definition}
That is, two circuits are equivalent if for every valuation of the parameters they differ only by a phase factor. The phase factor may in some cases be constant, but generally it may be different for different parameter values.
For example, the equivalence between $U_1$ and $R_z$ gates ($U_1(\theta)=e^{i \theta / 2} R_z(\theta)$) requires a parameter-dependent phase factor. %
Crucially for peephole optimization, circuit equivalence is invariant under replacing a subcircuit with an equivalent subcircuit.

A {\em circuit transformation} $(C_T, C_R)$ is a pair of distinct equivalent circuits, where $C_T$ is a {\em target circuit} to be matched with a subcircuit of the circuit being optimized, and $C_R$ is a {\em rewrite circuit} that can replace the target circuit while maintaining equivalence of the optimized circuit and the input circuit.
\Cref{fig:transformation_example} illustrates a circuit transformation. 

\begin{figure}
\centering
\begin{subfigure}{\linewidth}
\centering
\includegraphics[scale=\pptscale]{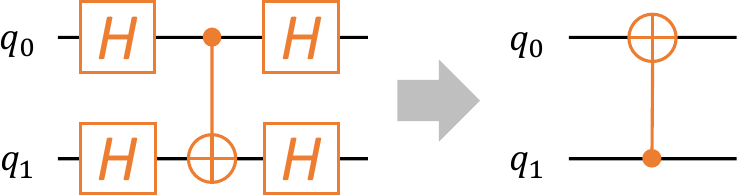}
\caption{Common optimization removing four Hadamard gates.}
\label{fig:hadamard_removing}
\end{subfigure}
\\
\vspace{1em}
\begin{subfigure}{\linewidth}
\centering
\includegraphics[scale=\pptscale]{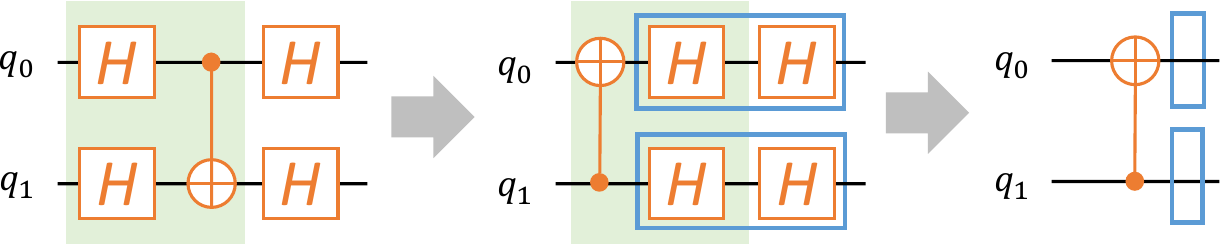}
\caption{Above optimization as a circuit rewriting using transformations.}
\label{fig:hadamard_removing_sequence}
\end{subfigure}
\\
\vspace{1em}
\begin{subfigure}{\linewidth}
\centering
\includegraphics[scale=\pptscale]{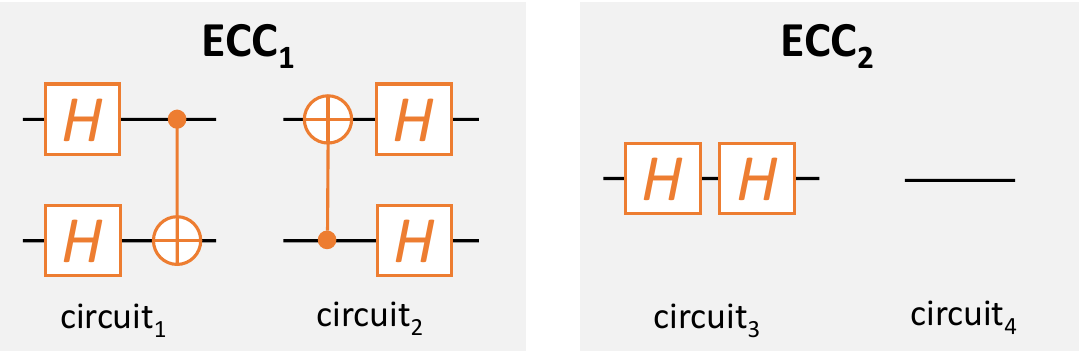}
\caption{An \ccs that includes the transformations used above.}
\label{fig:hadamard_removing_eccs}
\end{subfigure}
\caption{\label{fig:hadamard}%
Illustrating optimization via transformations.
}
\end{figure}

\paragraph{Equivalent Circuit Classes}
\sys uses {\em equivalent circuit classes} (ECCs), to represent many circuit transformations compactly.
An \ecc is a set of equivalent circuits.
A transformation is \emph{included} in an \ecc if both its target and rewrite circuits are in the \ecc. 
\eccs provide a compact representation of circuit transformations: an \ecc with $x$ circuits includes $x(x-1)$ transformations.

A \emph{circuit rewriting} is a sequence of applications of circuit
transformations, which \sys uses for optimization as illustrated in \Cref{fig:hadamard}.
\Cref{fig:hadamard_removing} shows a
common optimization that removes four Hadamard gates (i.e.,
$H$) by flipping a \cnot gate.
\Cref{fig:hadamard_removing_sequence} shows how to perform this
optimization as a circuit rewriting consisting of three (more basic) circuit transformations,
which are instances of the two transformations
specified by the \ccs in \Cref{fig:hadamard_removing_eccs}.

\paragraph{Completeness for \ccss.}
For any number of gates $n$, qubits $q$, and parameters $m$,
we assume a \emph{finite} set $\cnqm$ of %
circuits that can be constructed with at most $n$ gates over $q$ qubits and $m$ parameters. The collection $\{\cnqm\}_{n,q,m\in\mathbb{N}}$ is determined
by the \emph{gate set} $\m{G}$ (finite set of possibly parametric gates) as well as a specification $\Sigma$ of how parameter expressions may be constructed;
e.g., a finite set of arithmetic operations and some bounds on the depth of expressions or the number of times each parameter can be used,
ensuring finiteness of every $\cnqm$.
Henceforth, we fix the gate set $\m{G}$,
the parameter-expression specification $\Sigma$,
and the number of parameters $m$;
and elide $m$ by writing $\cnq$.

A transformation $(C_T, C_R)$ is \emph{subsumed} by an \ccs
if there is a circuit rewriting from $C_T$ to $C_R$ that only uses transformations included in the \ccs.

\begin{definition}[\nq-Completeness]
\label{def:nq}
Given $\m{G}$, $\Sigma$, and $m$ as above,
an \ccs is \emph{\nq-complete} (for $\m{G}$, $\Sigma$, and $m$)
if it subsumes all circuit transformations over $\cnq$.
\end{definition}
An \nq-complete \ccs can be used to rewrite
any two equivalent circuits with at most $n$
gates over $q$ qubits to each other. 
Any \ccs is by default \zeroq-complete because any transformation
involves at least one gate in the target or rewrite circuits.  
An \ccs is \oneq-complete if it subsumes all possible transformations
between single gates.  
\Cref{sec:generator,sec:verifier,sec:pruning} describe our approach for constructing a verified \nq-complete \ccs.

\section{Circuit Generator}
\label{sec:generator}

\sys builds an \nq-complete \ccs using the \regen{} algorithm, developed in this section,
which interleaves circuit generation and equivalence verification (see \Cref{fig:overview}).

\subsection{The \regen{} Algorithm}
\label{sec:repgen}

A straightforward way to generate an $(n, q)$-complete \ccs is to examine
all circuits in $\mathcal{C}^{(n,q)}$, %
but there are exponentially many such circuits.
To tackle this challenge, \regen{} uses {\em representative-based
circuit generation},
which significantly reduces the number of circuits considered.
The key idea is to %
extend a $(j,q)$-complete \ccs to a $(j+1,q)$-complete one by selecting a
representative circuit for each \ecc and using these representatives to build larger circuits.

\begin{algorithm}[t]
\caption{\regen.
}%
\label{alg1}
{
\small
\begin{algorithmic}[1]
\State {\bf Input:} Number of gates $n$ and qubits $q$.
\State {\bf Output:} \nq-complete \ccs. %
\State $\m{D} = \emptyset$ {\em // Hash table of circuit sets indexed by fingerprint} 
\State $\m{D}[\Call{FingerPrint}{\emptyseq}] = \{\emptyseq\}$ {\em// Initialize with empty circuit} \label{alg1:fingerprint}
\State $\m{R}_0 = \{\emptyseq\}$ {\em // Representatives with 0 gates} \label{alg1:initr} %
\State $\m{S}_0 = \emptyset$ {\em // \zeroq-complete \ccs}
\For{$j = 1$ \KwTo $\ngates$}
\State {\em // Step 1: construct circuits (sequences) with $j$ gates}
\For{each circuit $L\in\m{R}_{j-1}$ such that $|L|=j-1$} \label{alg1:enum_circuit}
\For{each gate $g \in \m{G}$} \label{alg1:enum_gate}
\For{each $\iota$ such that $L' = L.(g\ \iota)$ satisfies $\Sigma$}\label{alg1:enum_input}
\If{$\dropfirst{L'} \in \m{R}_{j-1}$} \label{alg1:slice}
\State add $L'$ to $\m{D}[\Call{FingerPrint}{L'}]$ \label{alg1:add_circuit}
\EndIf
\EndFor
\EndFor
\EndFor
\State {\em // Step 2: examine circuits with equal fingerprints} %
\State $\m{S}_j = \bigcup\{\eccfy{\gamma} : \gamma = \m{D}[f] \text{ for some }f\}$ \label{alg1:eccfy}
\State $\m{R}_j = \{\er{min}(\m{E}): \m{E} \in \m{S}_j\}$ // Under $\precede$ (\Cref{def:order}) \label{alg1:rn}
\EndFor
\State \Return{$\{\m{E} \in \m{S}_n : |\m{E}| \ge 2\}$} {\em // Ignore singleton \eccs} \label{alg1:return}
\State
\State {\em // Partition circuits into verified \eccs}
\Function{\eccfy{}}{$\gamma$}
\State $\widehat{\m{S}} = \emptyset$ {\em // Set of \eccs for circuits in $\gamma$}
\For{each circuit $C \in \gamma$}
\If{$C$ is equivalent to a circuit in \ecc $\m{E} \in \widehat{\m{S}}$}
\State $\widehat{\m{S}} = \widehat{\m{S}} \setminus \{\m{E}\} \cup \{\m{E} \cup \{C\}\}$ {\em // Add $C$ to $\m{E}$}
\Else
\State $\widehat{\m{S}} = \widehat{\m{S}} \cup \{\{C\}\}$ {\em // Create a new \ecc for $C$}
\EndIf
\EndFor
\State \Return $\widehat{\m{S}}$
\EndFunction

\end{algorithmic}}
\end{algorithm}

\begin{sloppypar}
\paragraph{Sequence representation for circuits.}
\regen{} represents a circuit as a sequence of instructions that
reflects a topological ordering of its gates (i.e., respecting dependencies).
E.g., a possible sequence for the circuit in \Cref{fig:our_circuit} is:
{\tt%
U1 \nolinebreak $\theta$ \nolinebreak 0;  \nolinebreak
U2 \nolinebreak $2\phi$ \nolinebreak $\lambda{+}\delta$ \nolinebreak 0;  \nolinebreak
H \nolinebreak 1;  \nolinebreak
X \nolinebreak 2;  \nolinebreak
CNOT \nolinebreak 1 \nolinebreak 2;  \nolinebreak
CNOT \nolinebreak 0 \nolinebreak 1%
}.\linebreak
We write $\emptyseq$ for the empty sequence and $L.(g\ \iota)$ for appending gate  $g\in\m{G}$
with arguments $\iota$ (parameter expressions and qubit indices) to sequence $L$; e.g., for the second instruction above $g=\texttt{U2}$ and $\iota = (2\phi, \lambda{+}\delta, 0)$.
Different sequences may represent the same circuit;
e.g. the \texttt{U2} and \texttt{X} instructions above can be swapped.
\regen eliminates some of this representation redundancy by the same mechanism it uses for avoiding redundancy due to circuit equivalence.
\end{sloppypar}

We use $|L|$, $\dropfirst{L}$, and $\droplast{L}$ to denote
the number of gates in a circuit $L$,
its suffix with $|L|-1$ gates, and
its prefix with $|L|-1$ gates.
Note that each of the latter two represents a subcircuit.

We fix an arbitrary total order over single-gate circuits (i.e., $\coneq$),
and lift it to a total order of circuits (i.e., sequences).

\begin{definition}[Circuit Precedence]
We say $L_1$ \emph{precedes} $L_2$, written $L_1 \precede L_2$, if
$|L_1| < |L_2|$, or if $|L_1| = |L_2|$ and $L_1$ is lexicographically smaller than $L_2$.
\label{def:order}
\end{definition}

\Cref{alg1} lists the \regen{} algorithm, which proceeds in rounds
and maintains a database $\m{D}$ of circuits grouped by fingerprints (defined below),
a \jq-complete \ccs $\m{S}_j$,
and a \emph{representative set} (set of representatives) $\m{R}_j$ 
for ECCs in $\m{S}_j$.
The $j$-th round produces a $(j, q)$-complete \ccs from the $(j-1, q)$-complete
\ccs generated in the previous round.
Each round proceeds in two steps.

\paragraph{Step 1: Constructing circuits.}
Before the first round, the initial \ccs is $\m{S}_0 =
\emptyset$ and the representative set is $\m{R}_0 = \{\emptyseq\}$, i.e., a singleton set
consisting of the empty circuit (over $q$ qubits). 
In the $j$-th round, \regen{} uses the \jminusoneq-complete \ccs $\m{S}_{j-1}$
and its representative set $\m{R}_{j-1}$
computed previously, and constructs possible size-$j$ circuits by appending a single
gate to each circuit in $\m{R}_{j-1}$ with size $j-1$.
\regen{} enumerates all possible gates $g$ and arguments $\iota$ according to $\m{G}$ and $\Sigma$. %
For each generated circuit $L'$, 
\regen checks if $\dropfirst{L'}$ is a representative from the previous round.
If so, \regen concludes that $L'$ extends existing
representatives and must be considered in generating $\m{S}_{j}$.
Otherwise, circuit $L'$ is considered redundant and ignored. %
We prove the correctness of \regen{} in \Cref{thm_representative}.

To identify potentially equivalent circuits, 
\regen computes circuit \emph{fingerprints}, and uses them as keys for storing circuits in the hash table $\m{D}$.
The fingerprint is computed using fixed, randomly selected parameter values and quantum states.
Recall that for a circuit $C$ over $q$ qubits and $m$ parameters,
and parameter values $\vec{p}\in\mathbb{R}^m$,
$\sem{C}(\vec{p})$ is a (concrete) $2^q \times 2^q$ complex matrix.
The fingerprint of a circuit $C$ is
\begin{equation}
\label{eqn:fingerprint}
\Call{FingerPrint}{C} = \left| \bra{\psi_0} \sem{C}(\vec{p_0}) \ket{\psi_1} \right|,
\end{equation}
where the parameter values $\vec{p_0}$ and quantum states $\ket{\psi_0}$ and $\ket{\psi_1}$ are fixed and randomly selected,
and $|\cdot|$ denotes modulus of a complex number.
With infinite precision, \cref{eqn:equiv,eqn:fingerprint} ensure that  equivalent circuits
have identical fingerprints.
This section presents and analyzes \regen{} assuming infinite precision, while \Cref{sec:impldetails} presents an adaptation for finite-precision floating-point arithmetic.

\paragraph{Step 2: Examining circuits with equal fingerprints.}
In Step 1, \regen generates circuits and stores them in the hash table $\m{D}$
grouped by their fingerprints.
In Step 2, \regen{} partitions each set
$\gamma = \m{D}[f]$ of potentially equivalent circuits into a verified \ccs 
using the function \eccfy{}. %
\eccfy{} considers each circuit in $\gamma$ and checks if it is equivalent to some existing \ecc in $\gamma$ by querying the verifier (Section~\ref{sec:verifier});
the circuit is then either added to the matching \ecc, or becomes a new singleton \ecc. \regen{} then combines the \ccss for each $\gamma$ to get the \ccs $\m{S}_j$.

Having constructed an \ccs $\m{S}_j$, \regen{} computes
$\m{S}_j$'s representative set $\m{R}_j$, which is the set of
representatives of the \eccs in $\m{S}_j$ (\Cref{alg1} line~\ref{alg1:rn}).
The representative of an \ecc is its
$\precede$-minimum circuit (\Cref{def:order}).
During the operation of \regen{}, singleton \eccs are important: their representatives must be considered when generating circuits in the next round,
and they may grow to non-singleton \eccs as more circuits are generated.
However, singleton \eccs in $\m{S}_n$ ultimately yield no transformations, so we remove them from the result of the \regen{} algorithm (line~\ref{alg1:return}).

\subsection{Correctness of \regen{}}
\label{subsec:representative_pruning}

When constructing circuits of size $j$ (in round $j$),
\regen only considers circuits $L'$ that extend previously constructed representatives,
i.e., $\droplast{L'}\in\m{R}_{j-1}$,
and only when the extension leads to $\dropfirst{L'}\in\m{R}_{j-1}$.
In this section we prove that in spite of that, \regen{} always generates an \nq-complete \ccs.
Below we use $\m{D}_j$ to denote the value of $\m{D}$ after Step 1 in \regen{}'s $j$-th round (i.e., at \Cref{alg1} line~\ref{alg1:eccfy}) and $\m{D}_0$ for the initial value of $\m{D}$ (line~\ref{alg1:fingerprint}).

\begin{lemma}
\label{lem:alg1inv}
\Cref{alg1} maintains the following invariants (for any $1 \leq j \leq n$, and writing $\sqsubseteq$ for Hoare ordering, i.e., $\m{X \sqsubseteq Y} \equiv \forall X \in \m{X}.\, \exists Y\in\m{Y}.\, X \subseteq Y$):
\begin{enumerate}
\item\label{lem:alg1inv:monotone}
$\m{D}_{j-1} \sqsubseteq \m{D}_{j}$,
$\m{S}_{j-1} \sqsubseteq \m{S}_{j}$, and
$\m{R}_{j-1} \subseteq \m{R}_{j}$;
\item\label{lem:alg1inv:d}  for any $L\in\cnq$,
$L\in\bigcup\m{D}_j$ iff $|L| \leq j$ and 
either $L=\emptyseq$ or $\dropfirst{L},\droplast{L}\in\m{R}_{j-1}$; and
\item\label{lem:alg1inv:r} for any $L\in\cnq$, $L\in\m{R}_{j}$ iff $|L| \leq j$ and $L$ does not have a $\precede$-smaller equivalent circuit in $\cnq$.
\end{enumerate}
\end{lemma}
\begin{proof}
For \cref{lem:alg1inv:monotone}, 
$\m{D}_{j-1} \sqsubseteq \m{D}_{j}$ and
$\m{S}_{j-1} \sqsubseteq \m{S}_{j}$ follow from the monotonic updating of $\m{D}$ (i.e., circuits are only added), and
from monotonicity (w.r.t. Hoare ordering) of \eccfy{}.
To see that $\m{R}_{j-1} \subseteq \m{R}_{j}$, observe that in the $j$-th round,
all circuits constructed are of size $j$, so any circuit added to an existing \ecc has more gates than its representative in $\m{R}_{j-1}$, which will therefore remain its representative in $\m{R}_{j}$ (recall that if $|L|<|L'|$ then $L \precede L'$).

We prove \cref{lem:alg1inv:d} by induction on $j$.
Both the base case ($j=1$) and the induction step follow from \Cref{alg1} lines~\ref{alg1:enum_circuit}--\ref{alg1:add_circuit}, combined with 
\cref{lem:alg1inv:monotone} and either the definition of $\m{D}_0$ or the induction hypothesis.

We prove \cref{lem:alg1inv:r} by induction on $j$.
Slightly generalizing from the statement above, we take the base case to be $j=0$, which follows from line~\ref{alg1:initr}.
In the induction step, for sufficiency, consider a circuit $L$, $1 \leq |L| \leq j$, with no $\precede$-smaller equivalent circuit.
(The $|L|=0$ case follows from line~\ref{alg1:initr} and \cref{lem:alg1inv:monotone}.)
Both $\dropfirst{L}$ and $\droplast{L}$ are of size $\leq j-1$
with no $\precede$-smaller equivalent circuits (if either had a $\precede$-smaller equivalent circuit, we could use it to construct a $\precede$-smaller equivalent circuit for $L$). By the induction hypothesis, $\dropfirst{L},\droplast{L}\in\m{R}_{j-1}$,
so by \cref{lem:alg1inv:d}, $L \in \bigcup\m{D}_j$. By lines~\ref{alg1:eccfy}--\ref{alg1:rn},
$\m{R}_j$ includes the $\precede$-minimal element of each class of equivalent circuits
in $\bigcup\m{D}$, so it must include $L$.
Necessity follows from sufficiency, combined with the fact that two circuits in $\m{R}_j$ cannot be equivalent and that $\m{R}_j$ does not contain circuits of size greater than $j$.
\end{proof}

\begin{theorem}[\regen]
\label{thm_representative}
In \Cref{alg1}, every $\m{S}_j$ ($0 \leq j \leq n$) is a \jq-complete \ccs, and the algorithm returns an \nq-complete \ccs.
\end{theorem}
\begin{proof}
We proceed using proof by contradiction.
Let $j$ be the smallest such that $\m{S}_j$ is not \jq-complete.
We must have $j>0$, with $\m{S}_{j-1}$ \jminusoneq-complete,
and by \Cref{lem:alg1inv} \cref{lem:alg1inv:monotone} $\m{S}_{j}$ is also \jminusoneq-complete.
(As $\m{S}_{j-1} \sqsubseteq \m{S}_{j}$, $\m{S}_{j}$ only includes more transformations.)
Let $(L,L')$ be the minimal (under the pairwise lexicographic lifting of $\precede$)
pair of equivalent circuits of size $\leq j$
that cannot be rewritten to each other using transformations included in $\m{S}_j$.
We must have $|L'|=j$, since otherwise $|L|,|L'|\leq j-1$, but $\m{S}_{j}$ is \jminusoneq-complete.

If $\dropfirst{L'}$ has a $\precede$-smaller equivalent circuit,
then $\m{S}_{j}$ can rewrite $L'$ to a $\precede$-smaller equivalent circuit,
which it must also not be able to rewrite to $L$, contradicting the minimality of $L'$.
Therefore, $\dropfirst{L'}$ does not have a $\precede$-smaller equivalent circuit;
The same argument works for $\droplast{L'}$.
Therefore, by using \Cref{lem:alg1inv} \cref{lem:alg1inv:r}
we get $\dropfirst{L'},\droplast{L'}\in\m{R}_{j-1}$,
and by \Cref{lem:alg1inv} \cref{lem:alg1inv:d}, $L'\in\m{D}_j$.
But if $L'\in\m{D}_j$ then either $\m{S}_j$ includes a transformation that rewrites
$L'$ to a smaller equivalent circuits, that it cannot rewrite to $L$, contradicting the minimality of $L'$; or $L'$ does not have a $\precede$-smaller equivalent circuit,
contradicting the definition of the pair $(L,L')$.
\end{proof}

\subsection{Complexity of \regen}
\label{subsec:complexity}

We analyze the time complexity of \regen (its space complexity is the same).
First, observe that the number of single-gate circuits $|\coneq|-1$, which is 
determined by the gate set $\m{G}$, parameter-expression specification $\Sigma$, number of qubits $q$, and parameters $m$,
provides an upper bound for the number of single-gate extensions of any existing circuit.
(The $-1$ is due to $\emptyseq\in\coneq.)$
This \emph{characteristic} of $\m{G}$, $\Sigma$, $q$, and $m$, denoted  $\mathrm{ch}(\m{G}, \Sigma, q, m) = |\coneq|-1$,
bounds the number of iterations of the loops in \Cref{alg1} lines~\ref{alg1:enum_gate} and~\ref{alg1:enum_input} in each round.
(The bound may not be tight as $\Sigma$ may impose more restrictions, e.g., single use of parameters.)

While $\sum_{j=0}^{n}\mathrm{ch}(\m{G}, \Sigma, q, m)^j$ provides a trivial upper bound on the complexity of \regen{}, the following theorem shows that \regen{}'s running time can be bounded using the number of \emph{resulting representatives} $|\m{R}_n|$.
In practice, this number is significantly smaller than $\mathrm{ch}(\m{G}, \Sigma, q, m)^n$ (see \Cref{tab:complexity}).

\begin{theorem}[Complexity of \regen{}]
\label{thm:complexity}
The time complexity of \Cref{alg1}, excluding the verification part (line~\ref{alg1:eccfy}), is \begin{equation*}
O\big(|\m{R}_{\ngates}| \cdot \mathrm{ch}(\m{G}, \Sigma, q, m) \cdot \ngates\big).
\end{equation*}
\end{theorem}
\begin{proof}
The $j$-th round of \Cref{alg1} considers circuits from $\m{R}_{j-1}$ with size $j-1$,
and for each one it considers at most $\mathrm{ch}(\m{G}, \Sigma, q, m)$ possible extensions.
It takes $O(\ngates)$ to construct a new circuit and add it to $\m{D}$.
(We assume $O(1)$ complexity for hash table insert and lookup, i.e., we use average and amortized complexity.)
Summing over all rounds of \Cref{alg1}, and recalling that $\m{R}_{j}\subseteq\m{R}_n$ (\Cref{lem:alg1inv} \cref{lem:alg1inv:monotone}): %
\begin{equation*}
\begin{aligned}
&\sum_{j=1}^{\ngates} |\{L \in \m{R}_{j-1} : |L| = j-1\}| \cdot \mathrm{ch}(\m{G}, \Sigma, q, m) \cdot \ngates 
\\ \leq\ & |\m{R}_{\ngates}| \cdot  \mathrm{ch}(\m{G}, \Sigma, q, m) \cdot  \ngates.
\end{aligned}
\end{equation*}
\end{proof}

Note that if $\m{G}$, $\Sigma$, $q$, and $m$ are considered as constant then the time complexity of \Cref{alg1} is $O(|\m{R}_{\ngates}| \cdot \ngates)$.

\Cref{tab:complexity} lists some empirical $\mathrm{ch}(\m{G}, \Sigma, q, m)$ and $|\m{R}_n|$ values.

\section{Circuit Equivalence Verifier}
\label{sec:verifier}

Given two circuits $C_1$ and $C_2$ over $q$ qubits and $m$ parameters,
the verifier checks if they are equivalent (i.e., up to a global phase).
Recalling \cref{eqn:equiv}, that means checking if 
$
\forall \vec{p} \in \mathbb{R}^m. \; \exists \beta \in \mathbb{R}. \; \sem{C_1}(\vec{p}) = e^{i\beta} \sem{C_2}(\vec{p})$.
Note that the equality here is between two $2^q \times 2^q$ complex matrices.

There are two challenges in automatically checking \cref{eqn:equiv}. One is the quantifier alternation,
which may be needed to account for global phase; the other is the use of trigonometric function, which is common in quantum gates' matrix representations.
For example, the $U_3$ gate supported by the IBM quantum processors has the following matrix representation:
\begin{equation}
\label{eqn:u3}
\sem{U_3}(\theta, \phi,\lambda) = 
\begin{pmatrix}
\cos(\frac{\theta}{2}) & -e^{i\lambda} \sin(\frac{\theta}{2})\\
e^{i\phi}\sin(\frac{\theta}{2}) & e^{i(\phi+\lambda)} \cos(\frac{\theta}{2})
\end{pmatrix}.
\end{equation}

While some SMT solvers support quantifiers and trigonometric functions~\cite{cvc4,DBLP:conf/cade/CimattiGIRS17},
our preliminary attempts showed they cannot directly prove \cref{eqn:equiv} for the circuit transformations generated by \sys.
Our verification approach is therefore to reduce \cref{eqn:equiv} to a quantifier-free formula over nonlinear real arithmetic by eliminating both the quantification over $\beta$ and the trigonometric functions.
The resulting verification conditions are then checked using the Z3~\cite{z3} SMT solver. 
This approach can efficiently verify all circuit transformations generated in our experiments (\Cref{sec:eval:generator}).

\paragraph{Phase factors}

To eliminate the existential quantification over the phase $\beta$,
we search over a finite space of linear combinations of the parameters $\vec{p}$ for a value
that can be used for $\beta$.
We consider $\beta(\vec{p}) = \vec{a} \cdot \vec{p} + b$,
where $\vec{a} \in A$ and $b \in B$
for some \emph{finite} sets $A \subseteq \mathbb{R}^m$ and $B \subseteq \mathbb{R}$.
(Our experimentation with various combinations of quantum gates suggested that $\vec{a}\neq\vec{0}$ is sometimes needed, so we develop the mechanism with this generality; however, in the experiments reported in \Cref{sec:eval}, constant phase factors, i.e. $\vec{a}=\vec{0}$, turned out to be sufficient for the three gate sets and the parameter specifications used.)
Given circuits $C_1$ and $C_2$, we find candidates for the coefficients $\vec{a}$ and $b$
using an approach similar to the one we use for generating candidate transformations.
We select random parameter values $\vec{p_0}$ and quantum states $\ket{\psi_0}$ and $\ket{\psi_1}$
and find all combinations of $\vec{a}$ and $b$ as above that satisfy
the following equation up to a small floating-point error (note that unlike \cref{eqn:fingerprint}, $|\cdot|$ is not used):
\begin{equation}
\label{eqn:phaseab}
\bra{\psi_0} \sem{C_1}(\vec{p_0}) \ket{\psi_1} = e^{i(\vec{a} \cdot \vec{p} + b)} \cdot \bra{\psi_0} \sem{C_2}(\vec{p_0}) \ket{\psi_1},
\end{equation}
For every such candidate coefficients $\vec{a}$ and $b$, we then attempt to verify following equation,
\begin{equation}
\label{eqn:verifier-qf}
\forall \vec{p} \in \mathbb{R}^m. \; \sem{C_1}(\vec{p}) = e^{i(\vec{a} \cdot \vec{p} + b)} \sem{C_2}(\vec{p}),
\end{equation}
which unlike \cref{eqn:equiv}, does not existentially quantify over $\beta$.
If \cref{eqn:verifier-qf} holds for some candidate coefficients, then $C_1$ and $C_2$ are verified to be equivalent. Otherwise, we consider the transformation given by $C_1$ and $C_2$
to fail verification, but that case did not occur in our experiments.

\paragraph{Trigonometric functions}

Matrices of parametric quantum gates we encountered only use their parameters inside arguments to $\sin$ or $\cos$ (after applying Euler's formula). Under this assumption, we reduce \cref{eqn:verifier-qf} to nonlinear real arithmetic in three steps.
First, we eliminate expressions such as $\frac{\theta}{2}$ 
that occur in some quantum gates (e.g., \cref{eqn:u3}) by introducing a fresh variable $\theta'=\frac{\theta}{2}$ and substituting $\theta'+\theta'$ for $\theta$.
After this step, all arguments to $\sin$ and $\cos$ are linear combinations of variables and constants (e.g., from phase factors) with integer coefficients.
Second, we exhaustively apply Euler's formula $e^{i\phi} = \cos\phi + i\sin\phi$, %
and trigonometric identities for parity and sum of angles:
$\sin(-x)=-\sin(x)$,
$\cos(-x)=\cos(x)$,
$\sin(x + y) = \sin(x) \cos(y) + \cos(x) \sin(y)$, and
$\cos(x + y) = \cos(x) \cos(y) - \sin(x) \sin(y)$.
After these steps, $\sin$ and $\cos$ are only applied to atomic terms (variables and constants).
For each constant $c$, we require precise symbolic expressions for $\sin(c)$ and $\cos(c)$ (e.g., $\sin(\frac{\pi}{4}) = \frac{\sqrt{2}}{2}$), and eliminate $\sin$ and $\cos$ over constants using these expressions.
Third, for every variable $t$ such that $\sin(t)$ or $\cos(t)$ is used we substitute $s_t$ for $\sin(t)$ and $c_t$ for $\cos(t)$, where 
$s_t$ and $c_t$ are fresh variables with a constraint $s_t^2 + c_t^2 = 1$,
which fully eliminates trigonometric functions.

Ultimately, Z3 can check the transformed version of \cref{eqn:verifier-qf} using the theory of quantifier-free nonlinear real arithmetic.

During the development of \sys we occasionally 
encountered verification failures,
but these were due to implementation bugs, and 
the counterexamples obtained from Z3 were useful in the debugging process.
Thus, verification is useful not only to ensure the ultimate correctness of the generated transformations,
but also in the development process. %

\section{Pruning Redundant Transformations}

\label{sec:pruning}
\begin{figure}
\centering
\begin{subfigure}{\linewidth}
    \centering
    \includegraphics[scale=\pptscale]{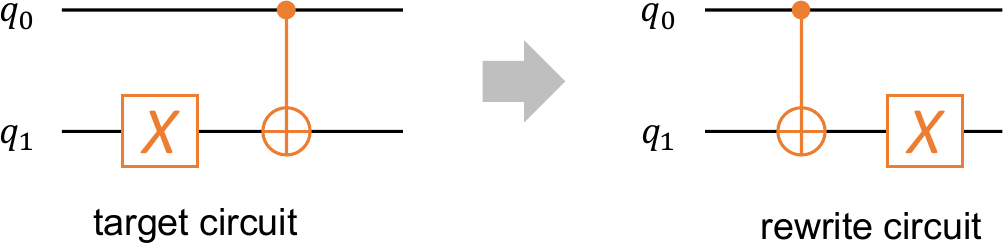}
    \caption{A circuit transformation.}
    \label{fig:subst_example}
\end{subfigure}
\\
\vspace{1em}
\begin{subfigure}{\linewidth}
    \centering
    \includegraphics[scale=\pptscale]{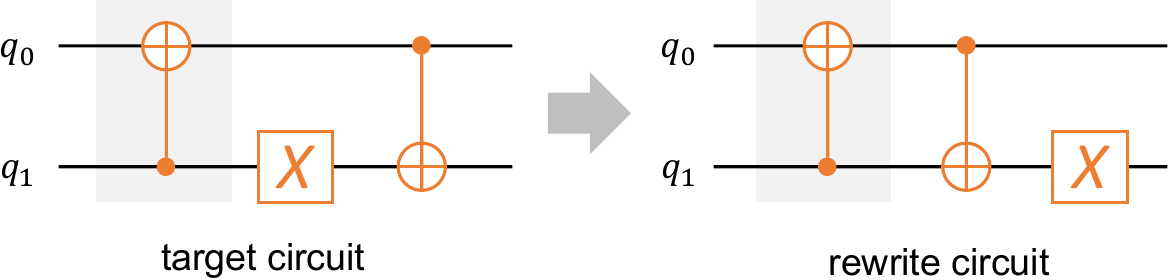}
    \caption{Redundant circuit transformation (common subcircuit).}
    \label{fig:redundant_subst-1}
\end{subfigure}
\caption{Illustrating a redundant circuit transformation.}
\label{fig:common_subcircuit}
\end{figure}
\sys applies two pruning steps after \regen{} generates an \nq-complete \ccs to further eliminate redundancy.
These steps maintain \nq-completeness while reducing the number of transformations the optimizer needs to consider.

\subsection{\ecc Simplification}
\label{subsec:circuit_simplification}
All \eccs generated by \regen{} have circuits with exactly $q$ qubits.
For each \ecc, a qubit (or a parameter) is {\em unused} if all circuits in the \ecc do not operate on the qubit (or the parameter). An {\em \ecc simplification} pass removes all unused qubits and parameters from each \ecc.
After this pass, some \eccs may become identical, among which only one is kept.

Because there is no specific order on parameters in a circuit, \sys also finds identical \eccs under a permutation of the parameters and maintains only one of them.

\subsection{Common Subcircuit Pruning}
\label{subsec:common_subcircuit_pruning}

\sys eliminates transformations whose target and rewrite circuits include a common subcircuit at the beginning or the end.
\Cref{fig:common_subcircuit}
illustrates this \emph{common subcircuit pruning}; the common subcircuit is highlighted in grey. %
\Cref{thm_common} explains why such transformations are always redundant.

\begin{definition}
A subset of gates $C'$ in a circuit $C$ is a \emph{subcircuit at the beginning} of $C$ if all gates in $C'$ are topologically before all gates in $C\setminus C'$.
Similarly, a subset of gates $C'$ in a circuit $C$ is a \emph{subcircuit at the end} of $C$ if all gates in $C'$ are topologically after all gates in $C\setminus C'$.
\end{definition}

\begin{theorem}
\label{thm_common}
For any two quantum circuits $C_1$ and $C_2$ with a common subcircuit at the beginning or the end, if $C_1$ and $C_2$ are equivalent, then eliminating the common subcircuit from $C_1$ and $C_2$ generates two equivalent circuits.
\end{theorem}
\begin{proof}
Recall that $\sem{C} (\vec{p})$ (for all $\vec{p}$---we elide $\vec{p}$ in this proof) denotes the matrix representation of circuit $C$.
Let $\sem{C}^\dagger$ be the conjugate transpose of $\sem{C}$,
and recall that as $\sem{C}$ is unitary, we have $\sem{C}^\dagger \sem{C} = \sem{C} \sem{C}^\dagger = I$.
Let $C_{s}$ denote the common subcircuit shared by $C_1$ and $C_2$.
Let $C_1'$ and $C_2'$ represent the new circuits obtained by removing $C_{s}$ from $C_1$ and $C_2$.
When $C_{s}$ is a common subcircuit at the beginning of $C_1$ and $C_2$, the matrix representations for the new circuits are
$\sem{C_i'} = \sem{C_{s}}^\dagger \sem{C_i}$, where $i=1,2$.
Equivalence between $C_1$ and $C_2$ implies the existence of $\beta$ such that $\sem{C_1} = e^{i\beta} \sem{C_2}$, therefore $\sem{C_1'} = e^{i\beta} \sem{C_2'}$.
The case where $C_s$ is a common subcircuit at the end is similar.
\end{proof}

\Cref{thm_common} shows that every transformation pruned in common subcircuit pruning must be subsumed by other transformations (assuming initial \nq-completeness).

Observe that if two circuits have a common subcircuit at the beginning (resp. the end), then they must have a common gate at the beginning (resp. the end). Therefore, to implement common subcircuit pruning, \sys only checks for a single common gate at the beginning or the end.

\section{Circuit Optimizer}
\label{sec:optimizer}
\begin{figure}
    \centering
    \includegraphics[scale=\pptscale]{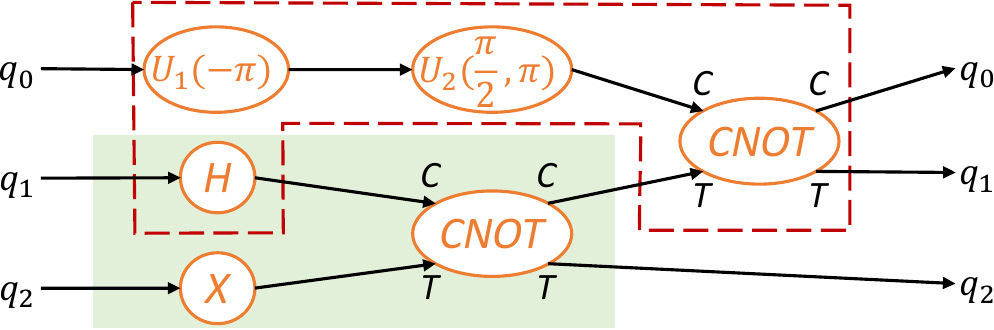}
    \caption{Graph representation for \Cref{fig:conventional_circuit}'s circuit. The green box (subcircuit, also convex subgraph) and red dashed area (not a subcircuit, also non-convex subgraph) match those of \Cref{fig:conventional_circuit}.}
    \label{fig:graph_representation}
\end{figure}

\begin{figure*}
\centering
\includegraphics[width=\linewidth]{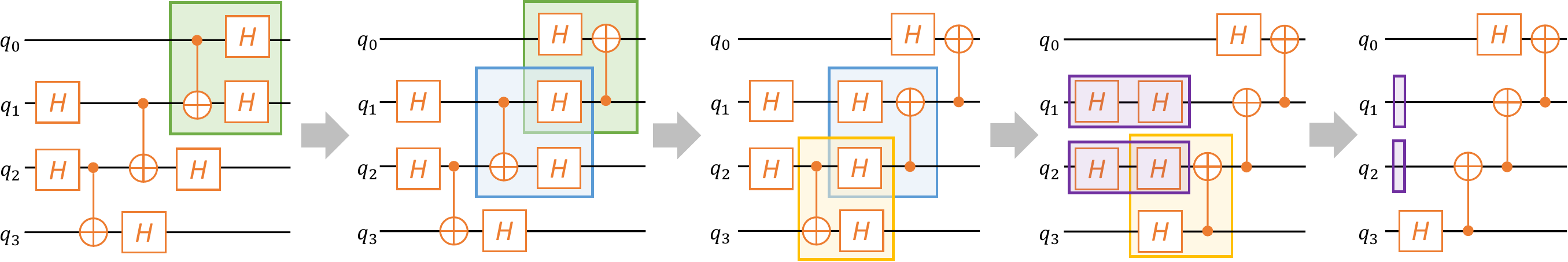}
\vspace{\captionvspace}
\caption{A transformation sequence applied by \sys that reduces the total gate count in the {\tt gf2\^{}4\_mult} circuit by swapping the control and target qubits of three $CNOT$ gates. Note that the first three transformations do not reduce gate count.
}
\label{fig:cnot_transformations}
\end{figure*}

\sys's {\em circuit optimizer} applies the verified transformations generated by the generator to find an optimized equivalent circuit for a given input circuit (see \Cref{fig:overview}).

A key step  is computing $\textproc{Apply}(C, T)$,
the set of circuits that can be obtained by applying  transformation %
$T=(C_T,C_R)$ to circuit $C$.
This involves finding all possible ways to match $C_T$ with a subcircuit of $C$.
\sys's optimizer uses a \emph{graph representation for circuits}, explained below, 
to implement this operation.
In the graph representation, subcircuits correspond to convex subgraphs,\footnote{For a graph $G$, $G'$ is a convex subgraph of $G$ if for any two vertices $u$ and $v$ in $G'$, every path in $G$ from $u$ to $v$ is also contained in $G'$.} 
and \sys adapts the graph-matching procedure from TASO~\cite{jia2019taso} 
to find all matches between $C_T$ and a convex subgraph of $C$.

In the graph representation, a circuit $C$ over $q$ qubits is represented as a directed graph $G$, where each gate over $d$ qubits is a vertex with in- and out-degree $d$.
Edges are labeled to distinguish between qubits in multi-qubit gates
(e.g., the control and target qubits of a $CNOT$ gate).
$G$ also includes $q$ sources and sinks, one for each qubit. %
\Cref{fig:graph_representation} illustrates the graph representation of \Cref{fig:conventional_circuit}'s circuit. As the figure also illustrates, subcircuits correspond to convex subgraphs.

The optimizer first converts an \nq-complete \ccs into a set of transformations (in the graph representation). For each \ecc with $x$ equivalent circuits $C_1,\ldots,C_x$, the optimizer considers a pair of transformations between the representative and each other circuit.
For example, if $C_1$ is the representative circuit in the \ecc, then the optimizer considers transformations $C_1\rightarrow C_i$ and $C_i\rightarrow C_1$ for $1<i\leq x$.
These $2(x - 1)$ transformations guarantee that any two circuits from the same \ecc are reachable from each other.

To optimize an input circuit using the above transformations, the optimizer uses a {\em cost-based backtracking search} algorithm adapted from TASO~\cite{jia2019metaflow, jia2019taso}.
The search is guided by a cost function $\Call{Cost}{\cdot}$ that maps circuits to real numbers.
In our evaluation, the cost is given by the number of gates in a circuit, but other cost functions are possible.

\begin{algorithm}[t]
\caption{Cost-Based Backtracking Search Algorithm.
}
\label{alg:search}
{
\small
\begin{algorithmic}[1]
\State {\bf Inputs:} Verified transformations $\m{T}$, a cost model $\Call{Cost}{\cdot}$, a hyper-parameter $\gamma$, and an input circuit $C_{in}$.
\State {\bf Output:} an optimized circuit $C_{best}$
\State {\em // $\m{Q}$ is a priority queue of circuits sorted by their $\Call{Cost}{\cdot}$.}
\State $\m{Q} = \{C_{in}\}$
\State $C_{best} = C_{in}$
\State $\m{D}_{seen} = \{C_{in}\}$
\While{$\m{Q} \neq \emptyset$ and the search has not timed out}
\State $C=\m{Q}$.dequeue()
\If{$\Call{Cost}{C} < \Call{Cost}{C_{best}}$}
\State $C_{best} = C$
\EndIf
\For{each transformation $T \in \m{T}$}
\For{each $C_{new} \in \Call{Apply}{C, T}\setminus \m{D}_{seen}$}
\If{$\Call{Cost}{C_{new}} < \gamma \cdot \Call{Cost}{C_{best}}$}
\State $\m{Q}$.enqueue($C_{new}$)
\State $\m{D}_{seen} = \m{D}_{seen} \cup \{C_{new}\}$
\EndIf
\EndFor
\EndFor
\EndWhile
\State \Return $C_{best}$
\end{algorithmic}}
\end{algorithm}

\Cref{alg:search} shows the pseudocode of our search algorithm. To find an optimized circuit, candidate circuits are maintained in a priority queue $\m{Q}$.
At each iteration, the lowest-cost circuit $C$ is dequeued, and \sys applies all transformations to get equivalent new circuits $C_{new}$, which are enqueued into $\m{Q}$ for further exploration. Circuits considered in the past are ignored using $\m{D}_{seen}$.

The search is controlled by a hyper-parameter $\gamma$. \sys  ignores candidate circuits whose cost
is greater than $\gamma$ times the cost of the current best circuit $C_{best}$.
The parameter $\gamma$ trades off between search time and the search's ability to avoid local minima.
For $\gamma=1$, \Cref{alg:search} becomes a greedy search that only accepts transformations that strictly improve cost.
On the other hand, a higher value for $\gamma$ 
enables application of
transformations that do not immediately improve the cost, which
may later lead to otherwise inaccessible optimization opportunities.
For example, \Cref{fig:cnot_transformations} depicts a sequence of five transformations that reduce the total gate count in the {\tt gf2\^{}4\_mult} (see \Cref{sec:eval:results}) circuit by four, via flipping three $CNOT$ gates;
note that the first three transformations do not reduce the gate count at all.

\sys's circuit optimizer is designed to optimize circuits before \emph{mapping}.
{\em Circuit mapping} converts a quantum circuit to an equivalent circuit that satisfies hardware constraints of a given quantum processor.
These constraints include connectivity restrictions between qubits and the directions to perform multi-qubit operations.
While transformations discovered by \sys are also applicable to circuits after mapping, applying them naively may break hardware constraints.
Therefore, we leave it as future work to build an optimizer for after-mapping circuits using \sys's transformations.

\section{Implementation and Evaluation}
\label{sec:eval}
We describe our implementation of \sys and evaluate the performance of the generator, the verifier, and the optimizer.
\sys is publicly available as an open-source project~\cite{quartz_github}
and also in the artifact supporting this paper~\cite{quartz_zenodo}.

\subsection{Implementation}
\label{sec:impldetails}

\paragraph{Floating-point arithmetic.}

The \regen{} algorithm as presented in \Cref{sec:generator} uses real-valued fingerprints,
where two equivalent circuits always have the same fingerprint.
Our implementation of \regen{} uses floating-point arithmetic, which introduces some imprecision that can potentially lead to different fingerprints for equivalent circuits.
To account for this imprecision, the implementation assumes there exists an {\em absolute error threshold} $E_{max}$, such that fingerprints of equivalent circuits differ by at most $E_{max}$ when computed with floating-point arithmetic.
The implementation therefore computes, using floating-point arithmetic,
the integer 
$\left\lfloor\frac{\Call{FingerPrint}{C}}{2 E_{max}}\right\rfloor$,
and uses it as the key for storing circuit $C$ in $\m{D}$.
Under our assumption, equivalent circuits may still have different integer hash keys, but they may differ only by $1$. 
Therefore, the implementation introduces an additional step after line~\ref{alg1:eccfy} of \Cref{alg1}, in which \eccs that correspond to circuits with hash keys $h$ and $h+1$ are checked for equivalence and merged if found equivalent.
In our experiments we set $E_{max}=10^{-15}$ based on preliminary exploration of the floating-point errors that occur in practice.

\begin{table}[t]
    \footnotesize
    \centering
    \caption{Gate sets used in our evaluation.}
    \label{tab:gate_set}
    \vspace{\captionvspace}
    \begin{tabular}{l|l}
        \hline
        {\bf Name} & {\bf Supported Gates} \\
        \hline
        Nam~\cite{nam2018automated,VOQC} & $H, \;\; X, \;\; R_z(\lambda),\;\; CNOT$ \\
        IBM~\cite{dumitrescu2018cloud} & $U_1(\theta), \;\;U_2(\phi, \lambda),\;\; U_3(\theta, \phi,\lambda), \;\; CNOT$  \\
        Rigetti~\cite{rigetti_agave} & $R_x(\frac{\pi}{2}), \;\; R_x(-\frac{\pi}{2}), \;\; R_x(\pi)=X, \;\; R_z(\lambda), \;\; CZ$ \\
        \hline
    \end{tabular}
\end{table}

\paragraph{Supported gate sets.}
\sys is a generic quantum circuit optimizer supporting arbitrary gate sets, and it accepts a gate set $\m{G}$ as part of its input.
In our experiments, input circuits are given over the ``Clifford + T'' gate set:
$H$, $T$, $T^\dagger$, $S$, $S^\dagger$, and $CNOT$;
and output (optimized) circuits are in one of the three gate sets listed in \Cref{tab:gate_set}:
Nam, IBM, and Rigetti. 
Nam is a gate set commonly used in prior work~\cite{nam2018automated,VOQC},
IBM is derived from the IBMQX5 quantum processor~\cite{dumitrescu2018cloud},
and Rigetti is derived from the Rigetti Agave quantum processor~\cite{rigetti_gates,rigetti_agave}.

To generate and verify circuit transformations for a new gate set, \sys only requires, for each gate,
a specification of its matrix representation as a function of its parameters, 
such as \cref{eqn:u1}. %
To optimize circuits, a translation procedure of input circuits to the new gate set is also required unless input circuits are provided in the new gate set.

\paragraph{Rotation merging and Toffoli decomposition.}

Before invoking \sys's optimizer, \sys preprocesses circuits
by applying two optimizations: \emph{rotation merging} and \emph{Toffoli decomposition}~\cite{nam2018automated}.
Our preliminary experiments showed that an approach solely based on local transformations and a cost-based search cannot reproduce these optimization passes for large circuits.
Rotation merging combines multiple $R_z$/$U_1$ gates that may be arbitrarily far apart (separated by $X$ or $CNOT$ gates), and appears difficult to be represented as local circuit transformations.
Toffoli decomposition decomposes a Toffoli gate into the Nam gate set, which involves simultaneous transformation of 15 quantum gates and interacts with rotation merging~\cite[p.~11]{nam2018automated}.
We therefore implement these two optimization passes from prior work~\cite{nam2018automated} as a preprocessing step.
Toffoli decomposition requires selecting a polarity for each Toffoli gate, which is computed heuristically by prior work~\cite{nam2018automated}.
Instead, we use a greedy approach: we process the Toffoli gates sequentially and for each gate we consider both polarities and greedily pick the one that results in fewer gates after rotation merging.

For the Nam and IBM gate sets, \sys directly applies rotation merging and Toffoli decomposition as a preprocessing step before the optimizer.
For the Rigetti gate set, which includes $CZ$ rather than $CNOT$, the algorithm from the prior work~\cite{nam2018automated} is not directly applicable;
therefore, \sys uses several additional preprocessing steps, as follows.
Rather than directly transpiling an input circuit to Rigetti,
\sys first transpiles it to Nam and applies Toffoli decomposition and rotation merging.
Next, \sys rewrites each $CNOT$ gate to a sequence of $H$, $CZ$, $H$ gates,
cancels out adjacent $H$ or $CZ$ pairs,
and then fully converts the circuit to Rigetti by transforming $X$ to 
$R_x(\pi)$ and $H$ to $R_x(\pi) \cdot R_z(\frac{\pi}{2}) \cdot R_x(\frac{\pi}{2}) \cdot R_z(-\frac{\pi}{2})$.
Ultimately, \sys invokes the optimizer, using a suitable $(n,q)$-complete ECC set for Rigetti.
Note that elimination of adjacent $H$ or $CZ$ pairs
during the translation from Nam to Rigetti leads to more optimized circuits:
a pair of adjacent $H$ gates (that are canceled out) would otherwise be transformed into a sequence of eight $R_x$ and $R_z$ gates that cannot be canceled out by \sys
since the cancellation is only correct for specific parameter values,
while \sys considers symbolic transformations valid for arbitrary parameter values.

\paragraph{Symbolic parameter expressions.}

As explained in \Cref{sec:ecc}, \sys assumes a fixed number of parameters $m$ and a specification $\Sigma$ for parameter expressions used in circuits.
\sys takes $m$ as input, and supports a flexible form for $\Sigma$ defined by a finite set of parameter expressions and either allowing or disallowing parameters to be used more than once in a circuit.

Our experiments use $m=2$ for the Nam and Rigetti gate sets,
and $m=4$ for the IBM gate set because it contains gates with up to three parameters.
For $\Sigma$, we consider the expressions
$p_{i}$, $2p_i$ and $p_{i}+p_{j}$ where $0 \leq i < m$ and $i < j < m$
(recall that $\vec{p}\in\mathbb{R}^m$ is the vector of parameters),
and restrict each parameter to be used at most once in a circuit. 
This restriction significantly reduces the number of circuits \regen{} considers, especially for the IBM gate set because the $U_3$ gate requires three parameter expressions.

As explained in \Cref{sec:verifier}, the verifier searches for phase factors of the form 
$\vec{a} \cdot \vec{p} + b$.
In our experiments we used
$a \in \{-2, -1, 0, 1, 2\}^m$
and
$b \in \{0, \frac{\pi}{4}, \frac{2\pi}{4}, \dots,\frac{7\pi}{4}\}$,
which proved to be useful and sufficient in our preliminary experimentation with various gates.
We later found that for the gate sets of \Cref{tab:gate_set}, $\vec{a}=\vec{0}$ is actually sufficient.
That is, these gate sets do not admit any transformations with parameter-dependent phase factors for the circuits we considered; they do however need various constant phase factors.

\subsection{Experiment Setup}

We compare \sys with existing quantum circuit optimizers on a
benchmark suite of 26 circuits developed by prior work \cite{nam2018automated, amy2014}. The benchmarks include arithmetic circuits (e.g., adding
integers), multiple controlled $X$ and $Z$ gates (e.g., $CCX$ and $CCZ$), the Galois field
multiplier circuits, and quantum Fourier transformations.
We use \sys to optimize the benchmark circuits to the three gate sets of \Cref{tab:gate_set}.

As in prior work~\cite{nam2018automated, VOQC}, we measure cost of a circuit in terms of the total gate count.
We therefore define the \textproc{Cost} function in \Cref{alg:search} as the number of gates in a circuit.\footnote{\sys can in principle be used to optimize for other metrics, e.g. number of $CNOT$ or $T$ gates, but here we focus on total gate count.}

Setting $n$ and $q$ for generating an \nq-complete \ccs determines the resulting transformations.
Our experiments use: $n=6,q=3$ for the Nam gate set;
$n=4,q=3$ for the IBM gate set; and
$n=3,q=3$ for the Rigetti gate set, which provided good results for our benchmarks.
\Cref{sec:eval:generator,subsec:scalability} discuss the impact of $n$ and $q$ on \sys's performance.

\sys's backtracking search (\Cref{alg:search}) is controlled by the hyper-parameter $\gamma$ and the timeout threshold.
Our experiments use $\gamma = 1.0001$, which yields good results for our benchmarks.
This value for $\gamma$ essentially means we consider cost-preserving transformations but not cost-increasing ones.
For the search timeout, we use 24 hours.
\Cref{subsec:scalability} discusses the timeout threshold and how it interacts with the settings for $n$ and $q$. 
To stop the search from consuming too much memory,
whenever the priority queue of \Cref{alg:search} contains more than 2,000 circuits we prune it and keep only the top 1,000 circuits.
Our preliminary experimentation with this pruning suggested that it does not affect \sys's results.

All experiments were performed on an \texttt{m6i.32xlarge} AWS EC2 instance
with a 128-core CPU and \SI{512}{GB} RAM.%

\subsection{Circuit Optimization Results}
\label{sec:eval:results}
\begin{table}
{
\centering
\caption{%
Gate count results for the Nam gate set.
The best result for each circuit is in bold.
``\sys Preprocess'' lists gate count after \sys's preprocessor (\Cref{sec:impldetails}).
}
\label{tab:nam}
\vspace{\captionvspace}
\footnotesize
\resizebox{\columnwidth}{!}
{%
\begin{threeparttable}[t]
\begin{tabular}{l|rrrr|rr}
\hline
{\bf Circuit} & {\bf Orig.} & {\bf Qiskit} & {\bf Nam} & {\bf \voqc} & {\bf \rotatebox[origin=c]{90}{\begin{tabular}{@{}c@{}}\sys\\Preprocess\end{tabular}}} & {\bf \rotatebox[origin=c]{90}{\begin{tabular}{@{}c@{}}\sys\\End-to-end\end{tabular}}} \\
\hline
{\tt adder\_8} & 900 & 869  & {\bf 606} & 682 & 732 & 724 \\
{\tt barenco\_tof\_3} & 58  & 56 & 40 & 50 & 46 & {\bf 38}\\
{\tt barenco\_tof\_4} & 114 & 109 & 72 & 95 & 86 & {\bf 68}\\
{\tt barenco\_tof\_5} & 170 & 162 & 104 & 140 & 126 & {\bf 98}\\
{\tt barenco\_tof\_10} & 450 & 427 & 264 & 365 & 326 & {\bf 262}\\
{\tt csla\_mux\_3} & 170 & 168 & 155 & 158 & 164 & \textbf{154}\\
{\tt csum\_mux\_9} & 420 & 420 & \textbf{266} & 308 & 308 & 272\\
{\tt gf2\^{}4\_mult} & 225 & 213 & 187 & 192 & 186 & {\bf 177}\\
{\tt gf2\^{}5\_mult} & 347 & 327 & 296 & 291 & 287 & {\bf 277}\\
{\tt gf2\^{}6\_mult} & 495 & 465 & 403 & 410 & 401 & {\bf 391}\\
{\tt gf2\^{}7\_mult} & 669 & 627 & 555 & 549 & 543 & {\bf 531}\\
{\tt gf2\^{}8\_mult} & 883 & 819 & 712 & 705 & {\bf 703} & {\bf 703}\\
{\tt gf2\^{}9\_mult} & 1095 & 1023 & 891 & 885 & 879 & {\bf 873}\\
{\tt gf2\^{}10\_mult} & 1347 & 1257 & 1070 & 1084 & 1062 & {\bf 1060} \\
{\tt mod5\_4} & 63 & 62 & 51 & 56 & 51 & {\bf 26}$^{\dag}$ \\
{\tt mod\_mult\_55} & 119 & 117 & 91 & {\bf 90} & 105 &  93 \\
{\tt mod\_red\_21} & 278 & 261 & {\bf 180} & 214 & 236 & 202 \\ 
{\tt qcla\_adder\_10} & 521 & 512 & {\bf 399} & 438 & 450 & 422 \\
{\tt qcla\_com\_7} & 443 & 428 & {\bf 284} & 314 & 349 & 292 \\
{\tt qcla\_mod\_7} & 884 & 853 & -$^{\dag\dag}$ & 723 & 727 & {\bf 719} \\
{\tt rc\_adder\_6} & 200 & 195 & {\bf 140} & 157 & 174 & 154 \\
{\tt tof\_3} & 45 & 44 & {\bf 35} & 40 & 39 & {\bf 35} \\
{\tt tof\_4} & 75 & 73 & {\bf 55} & 65 & 63 & {\bf 55} \\
{\tt tof\_5} & 105 & 102 & {\bf 75} & 90 & 87 & {\bf 75} \\
{\tt tof\_10} & 255 & 247 & {\bf 175} & 215 & 207 & {\bf 175} \\
{\tt vbe\_adder\_3} & 150 & 146 & 89 & 101 & 115 & {\bf 85} \\
\hline
{\bf \begin{tabular}{@{}c@{}}Geo. Mean\\Reduction\end{tabular}}
& - & 3.9\% & 27.3\% & 18.7\% & 18.6\% & 28.7\% \\
\hline
\end{tabular}%
\begin{tablenotes}
\footnotesize
\item[$\dag$] Computed as the median of seven runs:
25, 26, 26, 26, 32, 32, 32. %
\item[$\dag\dag$] Nam generates an incorrect circuit for \texttt{qcla\_mod\_7}~\cite[Table~1]{kissinger2019reducing}.
\end{tablenotes}
\end{threeparttable}
}
}
\end{table}
\paragraph{Nam gate set.}
\Cref{tab:nam} compares \sys to \qiskit~\cite{qiskit}, Nam~\cite{nam2018automated}, and \voqc~\cite{VOQC} for the Nam gate set.
(The performance of \tket~\cite{tket} for this gate set is similar to \qiskit, see ~\cite{VOQC}.)
The table also shows the gate count following \sys's preprocessing steps 
(rotation merging and Toffoli decomposition, see \Cref{sec:impldetails}).
\sys outperforms \qiskit and \voqc on almost all circuits,
indicating that it discovers most transformations used in these optimizers and also explores new optimization opportunities arising from new transformations and from the use of a cost-guided backtracking search (rather than a greedy approach, e.g., see \Cref{fig:cnot_transformations}).

\sys achieves on-par performance with Nam~\cite{nam2018automated},
a circuit optimizer highly tuned for this gate set.
Nam applies a set of carefully chosen heuristics
such as floating $R_z$ gates and canceling one- and two-qubit gates (see \cite{nam2018automated} for more detail).
While \sys's preprocessor implements two of Nam's optimization passes,
the results of the preprocessor alone are not close to Nam.\footnote{We observe that for the \texttt{gf2\^{}n\_mult} circuits, \sys' preprocessor outperforms Nam. We attribute this difference to our greedy Toffoli decomposition, discussed in~\Cref{sec:impldetails}, which happens to work well for these circuits.}
By using the automatically generated transformations, \sys is 
able to perform optimizations similar to some of Nam's other hand-tuned optimizations, and even outperform Nam on roughly half of the circuits.

For \texttt{mod5\_4}, we observed significant variability between runs, caused by randomness in ordering circuits with the same cost in the priority queue ($\m{Q}$ in \Cref{alg:search}).
Therefore, \Cref{tab:nam} reports the median result from seven runs as well as individual results.
This variability also suggests that \sys's performance can be improved by running the optimizer multiple times and taking the best discovered circuit, or by applying more advanced stochastic search techniques~\cite{DBLP:conf/pldi/KoenigPA21}.

\begin{table}[t]
\centering
\caption{%
Gate count results for the IBM gate set.
The best result for each circuit is in bold.
``\sys Preprocess'' lists gate count after \sys's preprocessor (\Cref{sec:impldetails}).
}
\label{tab:ibmq}
\vspace{\captionvspace}
\small
\resizebox{\columnwidth}{!}
{%
\begin{tabular}{l|rrrr|rr}
\hline
{\bf Circuit} & {\bf Orig.} & {\bf Qiskit} & {\bf t|ket$\rangle$} & {\bf \voqc} & {\bf \rotatebox[origin=c]{90}{\begin{tabular}{@{}c@{}}\sys\\Preprocess\end{tabular}}} & {\bf \rotatebox[origin=c]{90}{\begin{tabular}{@{}c@{}}\sys\\End-to-end\end{tabular}}} \\
\hline
{\tt adder\_8} & 900 & 805 & 775 & 643 & 736 & {\bf 583} \\
{\tt barenco\_tof\_3} & 58 & 51 & 51 & 46 & 46 & {\bf 36}\\
{\tt barenco\_tof\_4} & 114 & 100 & 100 & 89 & 86 & {\bf 67}\\
{\tt barenco\_tof\_5} & 170 & 149 & 149 & 135 & 126 & {\bf 98}\\
{\tt barenco\_tof\_10} & 450 & 394 & 394 & 347 & 326 & {\bf 253}\\
{\tt csla\_mux\_3} & 170 & 153 & 155 & 148 & 164 & {\bf 139}\\
{\tt csum\_mux\_9} & 420 & 382 & 361 & {\bf 308} & 364 & 340\\
{\tt gf2\^{}4\_mult} & 225 & 206 & 206 & 190 & 186 & {\bf 178}\\
{\tt gf2\^{}5\_mult} & 347 & 318 & 319 & 289 & 287 & {\bf 275}\\
{\tt gf2\^{}6\_mult} & 495 & 454 & 454 & 408 & 401 & {\bf 388}\\
{\tt gf2\^{}7\_mult} & 669 & 614 & 614 & 547 & 543 & {\bf 530}\\
{\tt gf2\^{}8\_mult} & 883 & 804 & 806 & 703 & 703 & {\bf 692}\\
{\tt gf2\^{}9\_mult} & 1095 & 1006 & 1009 & 882 & 879 & {\bf 866}\\
{\tt gf2\^{}10\_mult} & 1347 & 1238 & 1240 & 1080 & 1062 & {\bf 1050} \\
{\tt mod5\_4} & 63 & 58 & 58 & 53 & 55 & {\bf 51} \\
{\tt mod\_mult\_55} & 119 & 106 & 102 & {\bf 83} & 109 & 91 \\
{\tt mod\_red\_21} & 278 & 227 & 224 & {\bf 191} & 246 & 205 \\ 
{\tt qcla\_adder\_10} & 521 & 460 & 460 & 409 & 450 & {\bf 372} \\
{\tt qcla\_com\_7} & 443 & 392 & 392 & 292 & 349 & {\bf 267} \\
{\tt qcla\_mod\_7} & 884 & 778 & 780 & 666 & 726 & {\bf 594} \\
{\tt rc\_adder\_6} & 200 & 170 & 172 & {\bf 141} & 186 & 151 \\
{\tt tof\_3} & 45 & 40 & 40 & 36 & 39 & {\bf 31} \\
{\tt tof\_4} & 75 & 66 & 66 & 58 & 63 & {\bf 49} \\
{\tt tof\_5} & 105 & 92 & 92 & 80 & 87 & {\bf 67} \\
{\tt tof\_10} & 255 & 222 & 222 & 190 & 207 & {\bf 157} \\
{\tt vbe\_adder\_3} & 150 & 133 & 139 & 100 & 115 & {\bf 82} \\
\hline
{\bf \begin{tabular}{@{}c@{}}Geo. Mean\\Reduction\end{tabular}}
& - & 11.0\% & 11.2\% & 23.1\% & 17.4\% & 30.1\% \\
\hline
\end{tabular}%
}
\end{table}

\paragraph{IBM gate set.}
\Cref{tab:ibmq} compares \sys with \qiskit~\cite{qiskit}, \tket~\cite{tket}, and \voqc~\cite{VOQC} on the IBM gate set.
\qiskit and \tket include a number of optimizations specific to this gate set, such as merging any sequence of $U_1$, $U_2$, and $U_3$ gates into a single gate~\cite{qiskit_optimize1qgates} and replacing any block of consecutive 1-qubit gates by a single $U_3$ gate~\cite{qiskit_consolidateblocks}.
\sys is able to automatically discover some of these gate-specific optimizations by representing them each as a sequence of transformations.
Overall, \sys outperforms these existing compilers.

\begin{table}[t]
\centering
\caption{%
Gate count results for the Rigetti gate set.
The best result for each circuit is in bold.
``\sys Preprocess'' lists gate count after \sys's preprocessor (\Cref{sec:impldetails}).
}
\label{tab:rigetti}
\vspace{\captionvspace}
\footnotesize
\begin{threeparttable}
\begin{tabular}{l|rrr|rr}
\hline
{\bf Circuit} & {\bf Orig.} & {\bf \quilc} & {\bf t|ket$\rangle$} & {\bf \rotatebox[origin=c]{90}{\begin{tabular}{@{}c@{}}\sys\\Preprocess\end{tabular}}} & {\bf \rotatebox[origin=c]{90}{\begin{tabular}{@{}c@{}}\sys\\End-to-end\end{tabular}}} \\
\hline
{\tt adder\_8} & 5324 & 3345 & 3726 & 4244 & {\bf 2553} \\
{\tt barenco\_tof\_3} & 332 & 203 & 207 & 256 & {\bf 148} \\
{\tt barenco\_tof\_4} & 656 & 390 & 408 & 500 & {\bf 272} \\
{\tt barenco\_tof\_5} & 980 & 607 & 609 & 744 & {\bf 386} \\
{\tt barenco\_tof\_10} & 2600 & 1552 & 1614 & 1964 & {\bf 960} \\
{\tt csla\_mux\_3} & 1030 & {\bf 614} & 641 & 864 & 654 \\
{\tt csum\_mux\_9} & 2296 & 1540 & 1542 & 1736 & {\bf 1100} \\
{\tt gf2\^{}4\_mult} & 1315 & 809 & 827 & 1020 & {\bf 796} \\
{\tt gf2\^{}5\_mult} & 2033 & 1301 & 1277 & 1573 & {\bf 1231} \\
{\tt gf2\^{}6\_mult} & 2905 & 1797 & 1823 & 2235 & {\bf 1751} \\
{\tt gf2\^{}7\_mult} & 3931 & 2427 & 2465 & 3021 & {\bf 2371} \\
{\tt gf2\^{}8\_mult} & 5237 & 3208 & 3276 & 4033 & {\bf 3081} \\
{\tt gf2\^{}9\_mult} & 6445 & 4070 & 4037 & 4933 & {\bf 3986} \\
{\tt gf2\^{}10\_mult} & 7933 & 4977 & {\bf 4967} & 6048 & 5039 \\
{\tt mod5\_4} & 369 & 211 & 238 & 293 & {\bf 197} \\
{\tt mod\_mult\_55} & 657 & 420 & 452 & 531 & {\bf 361} \\
{\tt mod\_red\_21} & 1480 & 880 & 1020 & 1166 & {\bf 738} \\
{\tt qcla\_adder\_10} & 3079 & -$^\dagger$ & 1884 & 2464 & {\bf 1615} \\
{\tt qcla\_com\_7} & 2512 & 1540 & 1606 & 1954 & {\bf 1095} \\
{\tt qcla\_mod\_7} & 5130 & 3164 & 3202 & 4029 & {\bf 2525} \\
{\tt rc\_adder\_6} & 1186 & 706 & 747 & 984 & {\bf 606} \\
{\tt tof\_3} & 255 & 150 & 160 & 201 & {\bf 135} \\
{\tt tof\_4} & 425 & 271 & 270 & 333 & {\bf 199} \\
{\tt tof\_5} & 595 & 354 & 380 & 465 & {\bf 271} \\
{\tt tof\_10} & 1445 & 878 & 930 & 1125 & {\bf 631} \\
{\tt vbe\_adder\_3} & 900 & 534 & 557 & 705 & {\bf 366} \\
\hline
{\bf \begin{tabular}{@{}c@{}}Geo. Mean\\Reduction\end{tabular}}
& - & 38.6\% & 36.3\% & 21.9\% & 49.4\% \\
\hline
\end{tabular}%
\begin{tablenotes}
\footnotesize
\item[$\dagger$] \quilc supports up to 32 qubits while \texttt{qcla\_adder\_10} has 36.
\end{tablenotes}
\end{threeparttable}
\end{table}

\paragraph{Rigetti gate set.}
\Cref{tab:rigetti} compares \sys with \quilc~\cite{quilc} and \tket~\cite{tket} on the Rigetti gate set.
\sys significantly outperforms \tket and \quilc on most circuits, even though \quilc is highly optimized for this gate set.
We also note that while we employ some simplifications in the preprocessing phase for the Rigetti gate set (see \Cref{sec:impldetails}), most of the reduction in gate count comes from the optimization phase.

\subsection{Analyzing \sys's Generator and Verifier}
\label{sec:eval:generator}

\begin{table}[t]
\caption{%
Metrics for \sys's generator, when generating \nq-complete ECC sets
for $q{=}3$ and varying values of $n$ for the three gate sets.
$|\m{T}|$ is the resulting number of transformations,
$|\m{R}_n|$ is the size of the resulting representative set,
and $\mathrm{ch}$ is the characteristic (see \Cref{alg1} and \Cref{thm:complexity}).
}
\label{tab:complexity}
\vspace{\captionvspace}
\footnotesize
\centering
\begin{tabular}{ll|rrrr}
\hline
                                                                                        & $n$ & \multicolumn{1}{c}{\bf $|\m{T}|$} & \multicolumn{1}{c}{\bf $|\m{R}_n|$} & \multicolumn{1}{c}{\begin{tabular}[c]{@{}c@{}}\bf Verification\\ \bf Time (s)\end{tabular}} & \multicolumn{1}{c}{\begin{tabular}[c]{@{}c@{}}\bf Total\\ \bf Time (s)\end{tabular}} \\ \hline
\multirow{6}{*}{\begin{tabular}[c]{@{}l@{}}\bf Nam\\ $\mathrm{ch}=27$\end{tabular}}     & $2$ & 62                                & 397                                 & 1.2                                                                                         & 1.3                                                                                  \\
                                                                                        & $3$ & 196                               & 4,179                               & 2.6                                                                                         & 3.7                                                                                  \\
                                                                                        & $4$ & 1,304                             & 36,177                              & 8.5                                                                                         & 21.4                                                                                 \\
                                                                                        & $5$ & 8,002                             & 269,846                             & 49.5                                                                                        & 174.7                                                                                \\
                                                                                        & $6$ & 56,152                            & 1,777,219                           & 370.3                                                                                       & 1,400.4                                                                              \\
                                                                                        & $7$ & 379,864                           & 10,432,127                          & 2,673.6                                                                                     & 10,461.2                                                                             \\ \hline
\multirow{4}{*}{\begin{tabular}[c]{@{}l@{}}\bf IBM\\ $\mathrm{ch}=1,362$\end{tabular}}  & $2$ & 1,912                             & 22,918                              & 22.9                                                                                        & 38.6                                                                                 \\
                                                                                        & $3$ & 5,086                             & 224,281                             & 100.4                                                                                       & 225.9                                                                                \\
                                                                                        & $4$ & 16,748                            & 1,552,185                           & 356.9                                                                                       & 1,290.0                                                                              \\
                                                                                        & $5$ & 225,068                           & 7,847,203                           & 1,844.8                                                                                     & 8,363.1                                                                              \\ \hline
\multirow{5}{*}{\begin{tabular}[c]{@{}l@{}}\bf Rigetti\\ $\mathrm{ch}=30$\end{tabular}} & $2$ & 66                                & 361                                 & 1.3                                                                                         & 1.5                                                                                  \\
                                                                                        & $3$ & 66                                & 3,143                               & 2.6                                                                                         & 3.7                                                                                  \\
                                                                                        & $4$ & 224                               & 22,043                              & 5.8                                                                                         & 15.4                                                                                 \\
                                                                                        & $5$ & 2,396                             & 134,423                             & 22.7                                                                                        & 100.2                                                                                \\
                                                                                        & $6$ & 15,464                            & 729,842                             & 132.0                                                                                       & 675.3                                                                                \\ \hline
\end{tabular}
\end{table}

\begin{table}[t]
    \caption{%
    Number of circuits considered when using \regen{} with or without the pruning techniques of \Cref{sec:pruning} to generate \nq-complete ECC sets for $q{=}3$ and varying values of $n$ for the three gate sets.    
    Circuits are counted by their sequence representation, as \regen{} considers multiple sequences for each actual circuit (\Cref{sec:generator}).    
    Parenthesis shows reduction relative to the number of all
    possible circuits for $n$ and $q$. 
    ``\regen{}'' corresponds to \regen{} without additional pruning.
    ``+ ECC Simplification'' corresponds to \regen{} combined with ECC simplification.
    ``+ Common Subcircuit'' corresponds to \regen{} combined with all pruning techniques
    and ultimately represents \sys's generator.
    }
    \label{tab:eval_generator}
    \vspace{\captionvspace}
    \small 
    \centering
\resizebox{\linewidth}{!}{%
\begin{threeparttable}
\begin{tabular}{ll|r|rrr}
\hline
                                                                                                                      & $n$ & \multicolumn{1}{c|}{\begin{tabular}[c]{@{}c@{}}\bf Possible\\ \bf Circuits\end{tabular}} & \multicolumn{1}{c}{\bf \regen{}} & \multicolumn{1}{c}{\begin{tabular}[c]{@{}c@{}}\bf + ECC Sim-\\ \bf plification\end{tabular}} & \multicolumn{1}{c}{\begin{tabular}[c]{@{}c@{}}\bf + Common\\ \bf Subcircuit\end{tabular}} \\ \hline
\multirow{6}{*}{\bf \hspace{-2mm}\rotatebox[origin=c]{90}{\begin{tabular}{@{}c@{}}Nam\end{tabular}}\hspace{-3mm}}     & $2$ & 604                                                                                      & 400 (2$\times$)                  & 50 (12$\times$)                                                                              & 50 (12$\times$)                                                                           \\
                                                                                                                      & $3$ & 11,404                                                                                   & 1,180 (10$\times$)               & 231 (49$\times$)                                                                             & 164 (70$\times$)                                                                          \\
                                                                                                                      & $4$ & 198,028                                                                                  & 5,178 (38$\times$)               & 2,170 (91$\times$)                                                                           & 1,199 (165$\times$)                                                                       \\
                                                                                                                      & $5$ & 3,246,220                                                                                & 31,517 (103$\times$)             & 18,244 (178$\times$)                                                                         & 7,661 (424$\times$)                                                                       \\
                                                                                                                      & $6$ & 51,021,964                                                                               & 195,466 (261$\times$)            & 131,554 (388$\times$)                                                                        & 54,538 (936$\times$)                                                                      \\
                                                                                                                      & $7$ & 776,616,076                                                                              & 1,196,163 (649$\times$)          & 875,080 (887$\times$)                                                                        & 369,973 (2,099$\times$)                                                                   \\ \hline
\multirow{4}{*}{\bf \hspace{-2mm}\rotatebox[origin=c]{90}{\begin{tabular}{@{}c@{}}IBM\end{tabular}}\hspace{-3mm}}     & $2$ & 35,005                                                                                   & 23,413 (1$\times$)               & 1,708 (20$\times$)                                                                           & 1,708 (20$\times$)                                                                        \\
                                                                                                                      & $3$ & 533,857                                                                                  & 62,594 (9$\times$)               & 10,287 (52$\times$)                                                                          & 4,563 (117$\times$)                                                                       \\
                                                                                                                      & $4$ & 6,446,209                                                                                & 185,315 (35$\times$)             & 65,343 (99$\times$)                                                                          & 15,746 (409$\times$)                                                                      \\
                                                                                                                      & $5$ & 68,078,785                                                                               & 921,611 (74$\times$)             & 512,975 (133$\times$)                                                                        & 219,551 (310$\times$)                                                                     \\ \hline
\multirow{5}{*}{\bf \hspace{-2mm}\rotatebox[origin=c]{90}{\begin{tabular}{@{}c@{}}Rigetti\end{tabular}}\hspace{-3mm}} & $2$ & 778                                                                                      & 469 (2$\times$)                  & 51 (15$\times$)                                                                              & 51$^\dagger$ (15$\times$)                                                                           \\
                                                                                                                      & $3$ & 17,518                                                                                   & 965 (18$\times$)                 & 117 (150$\times$)                                                                            & 51$^\dagger$ (343$\times$)                                                                          \\
                                                                                                                      & $4$ & 367,843                                                                                  & 2,293 (160$\times$)              & 548 (671$\times$)                                                                            & 203 (1,812$\times$)                                                                       \\
                                                                                                                      & $5$ & 7,354,093                                                                                & 10,568 (696$\times$)             & 4,949 (1,486$\times$)                                                                        & 2,337 (3,147$\times$)                                                                     \\
                                                                                                                      & $6$ & 141,763,468                                                                              & 58,193 (2,436$\times$)           & 35,690 (3,972$\times$)                                                                       & 15,240 (9,302$\times$)                                                                    \\ \hline
\end{tabular}
\begin{tablenotes}
\footnotesize
\item[$\dagger$] For Rigetti, $n=2$ and $n=3$ result in identical transformations---each 3-gate transformation is subsumed by 2-gate transformations in a way identified by \sys.
\end{tablenotes}
\end{threeparttable}
}
\end{table}

We now examine \sys's circuit generator and circuit equivalence verifier.
\Cref{tab:complexity} shows the run times of the entire generation procedure,
and also the time out of that spent in verification, for each of the three gate sets and
for varying values of $n$, while fixing $q=3$.
The table also lists the number of resulting circuit transformations $|\m{T}|$,
the size of the resulting representative set $|\m{R}_n|$,
and the characteristic (see \Cref{alg1} and \Cref{thm:complexity}).
For all gate sets, $|\m{T}|$ and $|\m{R}_n|$ grow exponentially with $n$.
In spite of this exponential growth, 
the generator and verifier can generate, in a reasonable run time of a few hours, an \nq-complete \ccs for values of $n$ and $q$ that are sufficiently large to be useful for circuit optimization.
The growth in the number of transformations significantly affects the optimizer.
For Nam and IBM, our selected values of $n=6$ and $n=4$ result in a similar order of magnitude for $|\m{T}|$. For Rigetti, we use $n=2$, resulting in much smaller $\m{T}$. %
This choice is related to the fact that circuits in the Rigetti gate set are larger by roughly an order of magnitude compared to Nam and IBM (compare ``Orig.'' in \Cref{tab:rigetti} with \Cref{tab:nam,tab:ibmq}; see discussion in  \Cref{subsec:scalability}).

We now evaluate the effectiveness of \regen{} and the pruning techniques described in~\Cref{sec:pruning} for reducing the number of circuits \sys must consider (which is closely correlated with the number of resulting transformations).
To evaluate the relative contribution of each technique,
\Cref{tab:eval_generator} reports the number of circuits considered when applying:
(i)~\regen{} without additional pruning,
(ii)~\regen{} combined with \ecc simplification, and
(iii)~\regen{} combined with both \ecc simplification and common subcircuit pruning;
and compares each of these to a brute force approach of generating all possible circuits with up to $q$ qubits and $n$ gates.
Both \regen{} and the pruning techniques play an important role in eliminating redundant circuits while preserving \nq-completeness.
Ultimately, \regen{} and the pruning techniques reduce the number of transformations the optimizer must consider by one to three orders of magnitude.%

\subsection{Analyzing \sys's Circuit Optimizer}
\label{subsec:scalability}

We now examine \sys's circuit optimizer when using an \nq-complete \ecc set for varying values of $n$ and $q$.
For this study we focus on the Nam gate set, and
compare different values for $n$ and $q$ by the \emph{optimization effectiveness} they yield, defined as the reduction in geometric mean gate count over all circuits (as in the bottom line of \Cref{tab:nam}).
For \texttt{mod5\_4}, when $q=3$ and $3 \leq n \leq 7$,
we use the median of 7 runs %
due to the variability discussed in \Cref{sec:eval:results}.

As we increase $n$ and $q$ we expect \sys's optimizer to:
(i)~be able to reach more optimized circuits, and
(ii)~require more time per search iteration.
Both of these follow from the fact that increasing $n$ and $q$ yields more transformations.
Under a fixed search time budget, we expect the increased cost of search iterations to reduce the positive impact of larger $n$'s and $q$'s.
Because each iteration (\Cref{alg:search})
considers a candidate circuit $C$ and computes $\Call{Apply}{C,T}$
for each transformation $T\in\m{T}$, the cost per iteration scales linearly with
the number of transformations $|\m{T}|$.
Since $|\m{T}|$ varies dramatically as $n$ and $q$ change,\footnote{
\ifarxiv
For example: with $q=3$, $|\m{T}|=196$ for $n=3$ and $|\m{T}|=56,152$ for $n=6$; with $q=4$, $|\m{T}|=208$ for $n=3$ and $|\m{T}|=273,532$ for $n=6$ (see \Cref{tab:complexity:full}). We were unable to generate a $(7,4)$-complete ECC set using \SI{512}{GB} of RAM.
\fi
}
we expect the second effect (slowing down the search) to be significant,
especially for large circuits which typically require more search iterations (and additionally increase \textproc{Apply}'s running time).

\begin{figure}
    \centering
    \includegraphics[width=\evalfigfrac\linewidth]{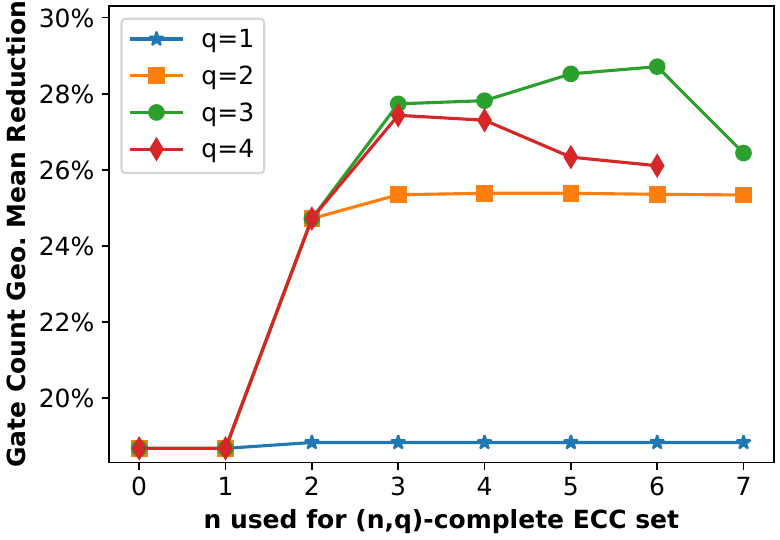}
    \vspace{\captionvspace}
    \caption{
    Optimization effectiveness with \nq-complete \ccss for varying $n$ and $q$ after a 24-hour search timeout.
    For $n=0$ there are no transformations and the results match the ``\sys Preprocess'' column of \Cref{tab:nam}.
    }
    \label{fig:scalability}
\end{figure}

\Cref{fig:scalability} shows optimization effectiveness (reduction in geometric mean gate count)
for varying values of $n$ and $q$,
under a search timeout of 24 hours.
The figure supports the tradeoff discussed above. %
Using too small values for $n$ and $q$ results in low effectiveness,
and as we increase $n$ or $q$ effectiveness increases but then starts decreasing,
as the negative impact of the large number of transformations starts outweighing their benefit.
\ifarxiv
(See \Cref{tab:complexity:full} for $|\m{T}|$ in each configuration.)
\fi
As expected, the optimal setting for $n$ and $q$ generally varies
across circuits---smaller circuits tend to be better optimized with larger values of 
\ifarxiv
$n$ (\Cref{tab:nam:full}).
\fi
Still, \Cref{fig:scalability} shows that there are several settings that yield good overall results:
$3 \le n \le 6$ for $q = 3$, and 
$3 \le n \le 4$ for $q = 4$.%
\footnote{
Interestingly, $q=3;3 \le n \le 6$ cover the best optimization results for all circuits
obtained among all configurations considered in
\ifarxiv
\Cref{fig:scalability} (\Cref{tab:nam:full}).
\fi
}

\begin{figure}
    \centering
    \includegraphics[width=\evalfigfrac\linewidth]{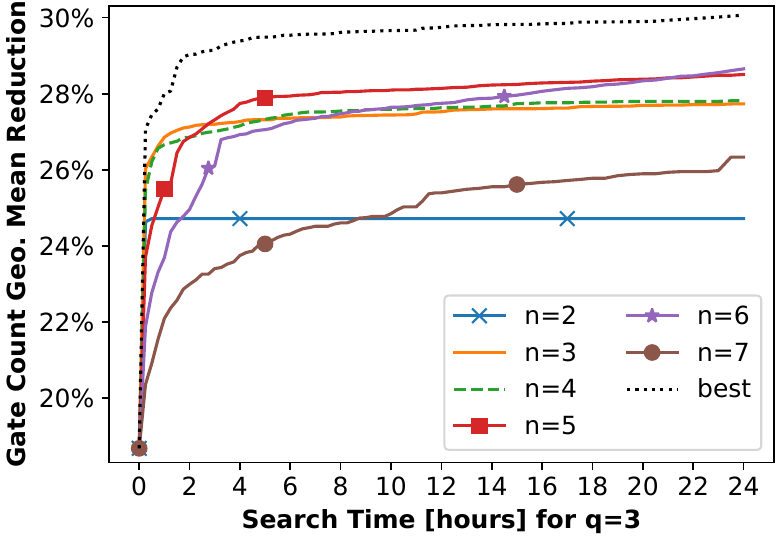}
    \vspace{\captionvspace}
    \caption{
    Optimization effectiveness over time ($q{=}3$; 
    $2{\le} n {\le} 7$).
    For each time point, ``best'' computes the reduction in geometric mean gate count obtained by selecting the most effective value for $n$ at that time point for each circuit (i.e., different circuits may use different $n$'s, and the same circuit may use different $n$'s at different time points).
    }
    \label{fig:time}
\end{figure}

\Cref{fig:time} shows how the search time impacts
optimization for different choices for $n$ (focusing on $q=3$).
For each value of $n$, we observe a quick initial burst, followed by a gentle increase.
At the end of the initial burst, effectiveness monotonically decreases as $n$ increases, for all $3 \le n \le 6$.
As time progresses the gaps diminish and eventually the order is reversed:
at around 21 hours $n=6$ surpasses $n=3$.
The settings $n=2$ and $n=7$ yield poor effectiveness:
$n=2$ does not contain an adequate number of transformations and quickly saturates the search time, while $n=7$ contains too many transformations and progresses too slowly.

\Cref{fig:time} also shows the effectiveness of a hypothetical run constructed by taking the best setting \emph{for each circuit at each time}.
This ``best'' curve considerably outperforms the others, because the best setting for $n$ varies across circuits with different sizes.

\ifarxiv
See \Cref{sec:extended} for more details, including plots akin to \Cref{fig:scalability} and \Cref{fig:time} for each circuit.
\fi

\section{Related Work}
\label{sec:related}

\paragraph{Quantum circuit compilation.}
Several optimizing compilers for quantum circuits have been recently introduced
and are being actively developed:
\qiskit~\cite{qiskit} and \tket~\cite{tket} support generic gate sets;
\quilc~\cite{quilc} is tailored to Rigetti Agave quantum processors;
\voqc~\cite{VOQC} is formally verified in Coq.
CertiQ~\cite{shi2020certiq} is a framework for writing and verifying \qiskit compiler passes.
\nam~\cite{nam2018automated} develop heuristics tailored to the $\{H, X, R_z, CNOT\}$ gate set.
Unlike \sys, these systems rely on quantum-computing experts to design, implement, and verify transformations.

Quanto~\cite{quanto} automatically discovers transformations by computing concrete
matrix representations of circuits. It supports parameters only by considering concrete values, and unlike \sys, it does not discover or verify symbolic transformations, which are the source of many of the challenges \sys deals with.
 Quanto uses  floating-point matrix equality to identify equivalence between circuits, while Quartz uses a combination of fingerprinting, SMT-based verification, the \regen{} algorithm, and other pruning techniques, which are needed since symbolic parameters greatly increase the number of possible circuits in the generation procedure.

Different from the aforementioned quantum optimizers that consider circuit transformations, PyZX~\cite{pyzx} employs ZX-diagrams as an intermediate representation for quantum circuits and uses a small set of complete rewrite rules in ZX-calculus~\cite{hadzihasanovic18, jeandel18} to simplify ZX-diagrams, which are finally converted back into quantum circuits.

While our approach builds on some of the techniques developed in prior work,
\sys is the first quantum circuit optimizer that can automatically generate and verify symbolic circuit transformations for arbitrary gate sets.

\paragraph{Superoptimization.}
Superoptimization is a compiler optimization technique originally designed to search for an optimal sequence of instructions for an input program~\cite{massalin1987}.
Our approach to generating quantum circuit transformations by tracking equivalent classes of circuits is inspired by prior work in automatically generating peephole optimizations for the X86 instruction set~\cite{heule2016, bansal2006} and generating graph substitutions for tensor algebra~\cite{jia2019taso, tensat, wang2021pet}.

TASO~\cite{jia2019taso} is a tensor algebra superoptimizer that optimizes computation graphs of deep neural networks using automatically generated graph substitutions.
TENSAT~\cite{tensat} reuses the graph substitutions discovered by TASO and employs equality saturation for tensor graph superoptimization.
While \sys draws inspiration from TASO and uses a similar search procedure,
it is significantly different from prior superoptimization works because it targets quantum computing, which leads to a different semantics (i.e., using complex matrices) as well as a different notion of program equivalence (i.e., up to a global phase). Verifying quantum circuit transformations therefore uses different techniques compared to other superoptimization contexts.
Applying equality saturation as in TENSAT~\cite{tensat} for optimizing quantum circuits is an interesting avenue for future work.

\section{Conclusion and Future Work}

We have presented \sys, a quantum circuit superoptimizer that automatically generates and verifies circuit transformations for arbitrary gate sets with symbolic parameters.
While \sys shows that a superoptimization-based approach to optimizing quantum circuits is practical, we believe there are many opportunities for further improvement.
As discussed in \Cref{subsec:scalability}, \sys's current search algorithm limits the number of transformations that can be effectively utilized. Improving the search algorithm may therefore lead to better optimization using \nq-complete ECC sets for larger values of $n$ and $q$, which may  also require improving the generator.
Another limitation of \sys that suggests an opportunity for future work is that it only targets the logical circuit optimization stage and does not consider qubit mapping.
Applying superoptimization to jointly optimize circuit logic and qubit mapping is both challenging and promising.

\section*{Acknowledgments}
We thank the anonymous PLDI reviewers and our shepherd, Xiaodi Wu, 
for their feedback.
This work was partially supported by the
\grantsponsor{NSF}{National Science Foundation}{nsf.gov}
under grant numbers~%
\grantnum{NSF}{CCF-2115104},
\grantnum{NSF}{CCF-2119352},
and
\grantnum{NSF}{CCF-2107241}.

\ifarxiv
\clearpage %
\fi

\bibliography{bib}

\ifarxiv
\clearpage
\appendix
\section{Detailed Results}
\label{sec:extended}

\begin{table*}[hbt!]
\setlength{\tabcolsep}{2pt}
\small
\caption{Gate count results for the Nam gate set when using $(n,q)$-complete \ccss with varying values of $n$ and $q$, and a search timeout of 24 hours. ``Pr.'' shows the result when $n=0$ or $n=1$: the \ccs is empty, so the results match the ``Quartz Preprocess'' column of \Cref{tab:nam}. When $q=1$, the results are the same when $2 \le n \le 7$.}
\begin{tabular}{l|r|r|r|rrrrrr|rrrrrr|rrrrr}
\hline
\multirow{2}{*}{\textbf{Circuit}}      & \multicolumn{1}{l|}{\multirow{2}{*}{\textbf{Orig.}}} & \multicolumn{1}{c|}{\multirow{2}{*}{\textbf{Pr.}}} & \multicolumn{1}{c|}{\multirow{2}{*}{$q=1$}} & \multicolumn{6}{c|}{$q=2,n=\uscore$}                                                                                                           & \multicolumn{6}{c|}{$q=3,n=\uscore$}                                                                                                           & \multicolumn{5}{c}{$q=4,n=\uscore$}                                                                                   \\ \cline{5-21} 
                                       & \multicolumn{1}{l|}{}                                & \multicolumn{1}{c|}{}                              & \multicolumn{1}{c|}{}                       & \multicolumn{1}{c}{2} & \multicolumn{1}{c}{3} & \multicolumn{1}{c}{4} & \multicolumn{1}{c}{5} & \multicolumn{1}{c}{6} & \multicolumn{1}{c|}{7} & \multicolumn{1}{c}{2} & \multicolumn{1}{c}{3} & \multicolumn{1}{c}{4} & \multicolumn{1}{c}{5} & \multicolumn{1}{c}{6} & \multicolumn{1}{c|}{7} & \multicolumn{1}{c}{2} & \multicolumn{1}{c}{3} & \multicolumn{1}{c}{4} & \multicolumn{1}{c}{5} & \multicolumn{1}{c}{6} \\ \hline
\texttt{tof\_3}                        & 45                                                   & 39                                                 & 39                                          & \tablehl{35}           & \tablehl{35}           & \tablehl{35}           & \tablehl{35}           & \tablehl{35}           & \tablehl{35}            & \tablehl{35}           & \tablehl{35}           & \tablehl{35}           & \tablehl{35}           & \tablehl{35}           & \tablehl{35}            & \tablehl{35}           & \tablehl{35}           & \tablehl{35}           & \tablehl{35}           & \tablehl{35}           \\
\texttt{barenco\_tof\_3}               & 58                                                   & 46                                                 & 46                                          & 42                    & 42                    & 42                    & 42                    & 42                    & 42                     & 42                    & 40                    & 40                    & 40                    & \tablehl{38}           & \tablehl{38}            & 42                    & 40                    & 40                    & 40                    & \tablehl{38}           \\
\texttt{mod5\_4}                       & 63                                                   & 51                                                 & 51                                          & 51                    & 51                    & 51                    & 51                    & 51                    & 51                     & 51                    & 45                    & 41                    & 27                    & \tablehl{26}           & 31                     & 51                    & 45                    & 37                    & 32                    & \tablehl{26}           \\
\texttt{tof\_4}                        & 75                                                   & 63                                                 & 63                                          & \tablehl{55}           & \tablehl{55}           & \tablehl{55}           & \tablehl{55}           & \tablehl{55}           & \tablehl{55}            & \tablehl{55}           & \tablehl{55}           & \tablehl{55}           & \tablehl{55}           & \tablehl{55}           & \tablehl{55}            & \tablehl{55}           & \tablehl{55}           & \tablehl{55}           & \tablehl{55}           & \tablehl{55}           \\
\texttt{tof\_5}                        & 105                                                  & 87                                                 & 87                                          & \tablehl{75}           & \tablehl{75}           & \tablehl{75}           & \tablehl{75}           & \tablehl{75}           & \tablehl{75}            & \tablehl{75}           & \tablehl{75}           & \tablehl{75}           & \tablehl{75}           & \tablehl{75}           & \tablehl{75}            & \tablehl{75}           & \tablehl{75}           & \tablehl{75}           & \tablehl{75}           & \tablehl{75}           \\
\texttt{barenco\_tof\_4}               & 114                                                  & 86                                                 & 86                                          & 78                    & 78                    & 78                    & 78                    & 78                    & 78                     & 78                    & 72                    & 72                    & 72                    & \tablehl{68}           & \tablehl{68}            & 78                    & 72                    & 72                    & 72                    & \tablehl{68}           \\
\texttt{mod\_mult\_55}                 & 119                                                  & 105                                                & 105                                         & 97                    & 94                    & 94                    & 94                    & 94                    & 94                     & 97                    & \tablehl{92}           & 93                    & 93                    & 93                    & 94                     & 97                    & \tablehl{92}           & 93                    & 93                    & 94                    \\
\texttt{vbe\_adder\_3}                 & 150                                                  & 115                                                & 115                                         & 95                    & 95                    & 95                    & 95                    & 95                    & 95                     & 95                    & 91                    & 91                    & 89                    & \tablehl{85}           & \tablehl{85}            & 95                    & 91                    & 90                    & 91                    & 91                    \\
\texttt{barenco\_tof\_5}               & 170                                                  & 126                                                & 126                                         & 114                   & 114                   & 114                   & 114                   & 114                   & 114                    & 114                   & 104                   & 104                   & 104                   & \tablehl{98}           & 100                    & 114                   & 104                   & 104                   & 104                   & 102                   \\
\texttt{csla\_mux\_3}                  & 170                                                  & 164                                                & 164                                         & 160                   & 152                   & 152                   & 152                   & 152                   & 152                    & 160                   & \tablehl{146}          & 148                   & 149                   & 154                   & 156                    & 160                   & 150                   & 148                   & 153                   & 153                   \\
\texttt{rc\_adder\_6}                  & 200                                                  & 174                                                & 174                                         & 154                   & \tablehl{152}          & \tablehl{152}          & \tablehl{152}          & \tablehl{152}          & \tablehl{152}           & 154                   & \tablehl{152}          & \tablehl{152}          & \tablehl{152}          & 154                   & 154                    & 154                   & \tablehl{152}          & \tablehl{152}          & 154                   & 156                   \\
\texttt{gf2\textasciicircum{}4\_mult}  & 225                                                  & 186                                                & 186                                         & 186                   & 180                   & 180                   & 180                   & 180                   & 180                    & 186                   & 176                   & \tablehl{175}          & 177                   & 177                   & 178                    & 186                   & \tablehl{175}          & 177                   & 178                   & 180                   \\
\texttt{tof\_10}                       & 255                                                  & 207                                                & 207                                         & \tablehl{175}          & \tablehl{175}          & \tablehl{175}          & \tablehl{175}          & \tablehl{175}          & \tablehl{175}           & \tablehl{175}          & \tablehl{175}          & \tablehl{175}          & \tablehl{175}          & \tablehl{175}          & \tablehl{175}           & \tablehl{175}          & \tablehl{175}          & \tablehl{175}          & \tablehl{175}          & 191                   \\
\texttt{mod\_red\_21}                  & 278                                                  & 236                                                & 226                                         & \tablehl{202}          & \tablehl{202}          & \tablehl{202}          & \tablehl{202}          & \tablehl{202}          & \tablehl{202}           & \tablehl{202}          & \tablehl{202}          & \tablehl{202}          & \tablehl{202}          & \tablehl{202}          & 210                    & \tablehl{202}          & \tablehl{202}          & \tablehl{202}          & \tablehl{202}          & 228                   \\
\texttt{gf2\textasciicircum{}5\_mult}  & 347                                                  & 287                                                & 287                                         & 287                   & 279                   & 279                   & 279                   & 279                   & 279                    & 287                   & \tablehl{273}          & \tablehl{273}          & \tablehl{273}          & 277                   & 279                    & 287                   & 277                   & 277                   & 279                   & 283                   \\
\texttt{csum\_mux\_9}                  & 420                                                  & 308                                                & 308                                         & 308                   & 308                   & 308                   & 308                   & 308                   & 308                    & 308                   & 280                   & 280                   & 280                   & \tablehl{272}          & 302                    & 308                   & 280                   & 280                   & 305                   & 307                   \\
\texttt{qcla\_com\_7}                  & 443                                                  & 349                                                & 347                                         & 293                   & 293                   & 293                   & 293                   & 293                   & 293                    & 293                   & 289                   & \tablehl{288}          & \tablehl{288}          & 292                   & 339                    & 293                   & 289                   & \tablehl{288}          & 321                   & 347                   \\
\texttt{barenco\_tof\_10}              & 450                                                  & 326                                                & 326                                         & 294                   & 294                   & 294                   & 294                   & 294                   & 294                    & 294                   & 268                   & 271                   & 271                   & \tablehl{262}          & 316                    & 294                   & 268                   & 275                   & 298                   & 324                   \\
\texttt{gf2\textasciicircum{}6\_mult}  & 495                                                  & 401                                                & 401                                         & 401                   & 391                   & 391                   & 391                   & 391                   & 391                    & 401                   & \tablehl{381}          & 383                   & 386                   & 391                   & 393                    & 401                   & 386                   & 391                   & 391                   & 401                   \\
\texttt{qcla\_adder\_10}               & 521                                                  & 450                                                & 450                                         & 416                   & 414                   & 414                   & 414                   & 414                   & 414                    & 416                   & \tablehl{407}          & 408                   & 408                   & 422                   & 444                    & 416                   & 408                   & 414                   & 436                   & 450                   \\
\texttt{gf2\textasciicircum{}7\_mult}  & 669                                                  & 543                                                & 543                                         & 543                   & 531                   & 531                   & 531                   & 531                   & 531                    & 543                   & \tablehl{517}          & 519                   & 529                   & 531                   & 539                    & 543                   & 524                   & 530                   & 537                   & 543                   \\
\texttt{gf2\textasciicircum{}8\_mult}  & 883                                                  & 703                                                & 703                                         & 703                   & 703                   & 703                   & 703                   & 703                   & 703                    & 703                   & \tablehl{690}          & 703                   & 703                   & 703                   & 703                    & 703                   & 703                   & 703                   & 703                   & 703                   \\
\texttt{qcla\_mod\_7}                  & 884                                                  & 727                                                & 727                                         & 657                   & 657                   & 657                   & 657                   & 657                   & 657                    & 657                   & 654                   & \tablehl{651}          & 677                   & 719                   & 727                    & 657                   & 655                   & 697                   & 725                   & 727                   \\
\texttt{adder\_8}                      & 900                                                  & 732                                                & 732                                         & 644                   & 640                   & 640                   & 640                   & 640                   & 644                    & 644                   & 638                   & \tablehl{634}          & 688                   & 724                   & 732                    & 644                   & \tablehl{634}          & 706                   & 730                   & 732                   \\
\texttt{gf2\textasciicircum{}9\_mult}  & 1095                                                 & 879                                                & 879                                         & 879                   & 877                   & 869                   & 869                   & 877                   & 877                    & 879                   & 857                   & \tablehl{856}          & 870                   & 873                   & 879                    & 879                   & 857                   & 871                   & 879                   & 879                   \\
\texttt{gf2\textasciicircum{}10\_mult} & 1347                                                 & 1062                                               & 1062                                        & 1062                  & 1062                  & 1058                  & 1058                  & 1058                  & 1058                   & 1062                  & \tablehl{1030}         & 1049                  & 1052                  & 1060                  & 1062                   & 1062                  & 1061                  & 1058                  & 1062                  & 1062                  \\ \hline
\end{tabular}

\label{tab:nam:full}
\vspace{4em}
\end{table*}

\begin{table*}[hbt!]
\caption{Metrics for \sys's generator, when generating \nq-complete \ccss for $2 \le n \le 7$ and $1 \le q \le 4$ for the Nam gate set. $|\m{T}|$ is the resulting number of transformations. The characteristics (see definition in \Cref{subsec:complexity}) for $q=1,2,3,4$ are 7, 16, 27, 40, respectively.}
\begin{threeparttable}
\begin{tabular}{c|rrrr|rrrr|rrrr}
\hline
    & \multicolumn{4}{c|}{\bf $|\m{T}|$}                                                                                           & \multicolumn{4}{c|}{\bf Verification Time (s)}                                                                               & \multicolumn{4}{c}{\bf Total Time (s)}                                                                                      \\ \hline
$n$ & \multicolumn{1}{c}{$q=1$} & \multicolumn{1}{c}{$q=2$} & \multicolumn{1}{c}{$q=3$} & \multicolumn{1}{c|}{$q=4$}               & \multicolumn{1}{c}{$q=1$} & \multicolumn{1}{c}{$q=2$} & \multicolumn{1}{c}{$q=3$} & \multicolumn{1}{c|}{$q=4$}               & \multicolumn{1}{c}{$q=1$} & \multicolumn{1}{c}{$q=2$} & \multicolumn{1}{c}{$q=3$} & \multicolumn{1}{c}{$q=4$}               \\ \hline
2   & 14                        & 38                        & 62                        & 78                                       & 0.5                       & 0.7                       & 1.2                       & 2.5                                      & 0.5                       & 0.7                       & 1.3                       & 2.8                                     \\
3   & 14                        & 90                        & 196                       & 208                                      & 0.8                       & 1.3                       & 2.6                       & 7.8                                      & 0.8                       & 1.5                       & 3.7                       & 12.0                                    \\
4   & 44                        & 332                       & 1,304                     & 2,988                                    & 1.1                       & 2.3                       & 8.5                       & 48.0                                     & 1.2                       & 3.8                       & 21.4                      & 118.3                                   \\
5   & 78                        & 1,334                     & 8,002                     & 27,942                                   & 1.5                       & 8.8                       & 49.5                      & 917.0                                    & 1.6                       & 19.2                      & 174.7                     & 2,452.0                                 \\
6   & 120                       & 5,794                     & 56,152                    & 273,532                                  & 1.9                       & 52.5                      & 370.3                     & 5,802.6                                  & 2.2                       & 138.9                     & 1,400.4                   & 19,448.0                                \\
7   & 164                       & 21,824                    & 379,864                   & \multicolumn{1}{c|}{--\tnote{$\dagger$}} & 2.3                       & 71.8                      & 2,673.6                   & \multicolumn{1}{c|}{--\tnote{$\dagger$}} & 2.8                       & 222.9                     & 10,461.2                  & \multicolumn{1}{c}{--\tnote{$\dagger$}} \\ \hline
\end{tabular}
\begin{tablenotes}
\footnotesize
\item[$\dagger$] 
We were unable to generate a $(7, 4)$-complete \ccs using \SI{512}{GB} of RAM.
\end{tablenotes}
\end{threeparttable}
\label{tab:complexity:full}
\end{table*}

\Cref{tab:nam:full} shows the final gate count for each circuit for varying $n$ and $q$ after a 24-hour search timeout, with the best results of each circuit highlighted. Interestingly, $q=3$ with $3\le n \le 6$ covers the best optimization results for all circuits obtained among all configurations considered in the table. As circuits are sorted in the order of original size in the table, we can see that when $q=3$, small circuits tend to be better optimized with larger values of $n$, and larger circuits tend to be better optimized with smaller values of $n$.

\Cref{tab:complexity:full} shows the run times of the entire generation procedure,
and also the time out of that spent in verification, for each of the \ccs used in \Cref{tab:nam:full}.
The table also lists the number of resulting circuit transformations $|\m{T}|$ for each $n$ and $q$,
and the characteristic for each $q$ (see definition in \Cref{subsec:complexity}).

\Cref{fig:adder_8,fig:barenco_tof_3,fig:barenco_tof_4,fig:barenco_tof_5,fig:barenco_tof_10,fig:csla_mux_3,fig:csum_mux_9,fig:gf2^4_mult,fig:gf2^5_mult,fig:gf2^6_mult,fig:gf2^7_mult,fig:gf2^8_mult,fig:gf2^9_mult,fig:gf2^10_mult,fig:mod5_4,fig:mod_mult_55,fig:mod_red_21,fig:qcla_adder_10,fig:qcla_com_7,fig:qcla_mod_7,fig:rc_adder_6,fig:tof_3,fig:tof_4,fig:tof_5,fig:tof_10,fig:vbe_adder_3} 
show plots akin to \Cref{fig:scalability} and \Cref{fig:time} for each circuit. In each figure, the left plot shows the optimization result with \nq-complete \ccss for varying $n$ and $q$ after a 24-hour search timeout. The right plot shows the optimization result over time for $q=3$ and $2\le n \le 7$. Each 24-hour result on the right plot corresponds to a green round marker ($q=3$) on the left plot.

Among these figures, \Cref{fig:mod5_4} shows the median run for \texttt{mod5\_4}. The median run is defined to be the run with the final gate count being the median of the 7 runs.
We present the results of all 7 runs for \texttt{mod5\_4} for $q=3$ and $3 \le n \le 7$ in \Cref{fig:mod54:full}.

\clearpage

\begin{figure*}
\centering
\includegraphics[width=0.48\linewidth]{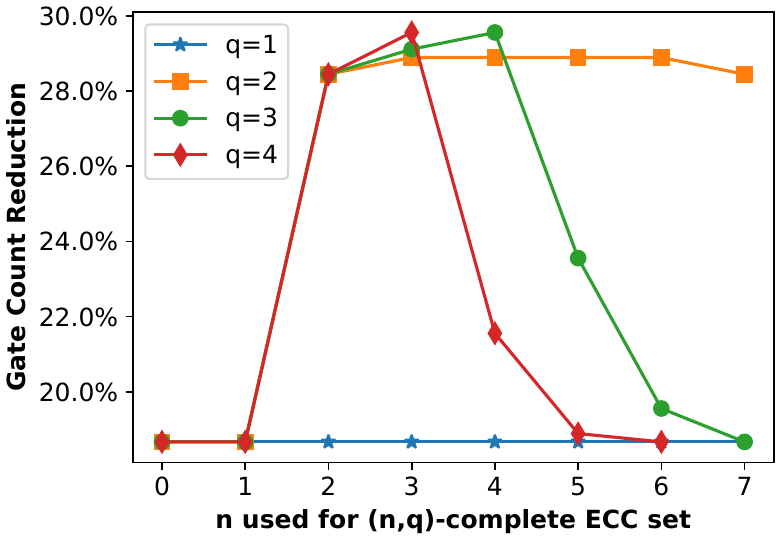}
\hfill
\includegraphics[width=0.48\linewidth]{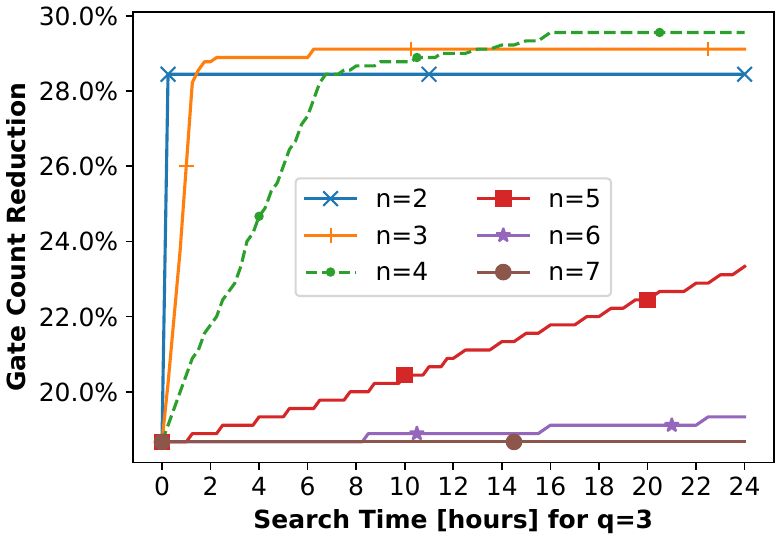}
\caption{\texttt{adder\_8} (900 gates).\vspace{1.5em}}
\label{fig:adder_8}
\end{figure*}

\begin{figure*}
\centering
\includegraphics[width=0.48\linewidth]{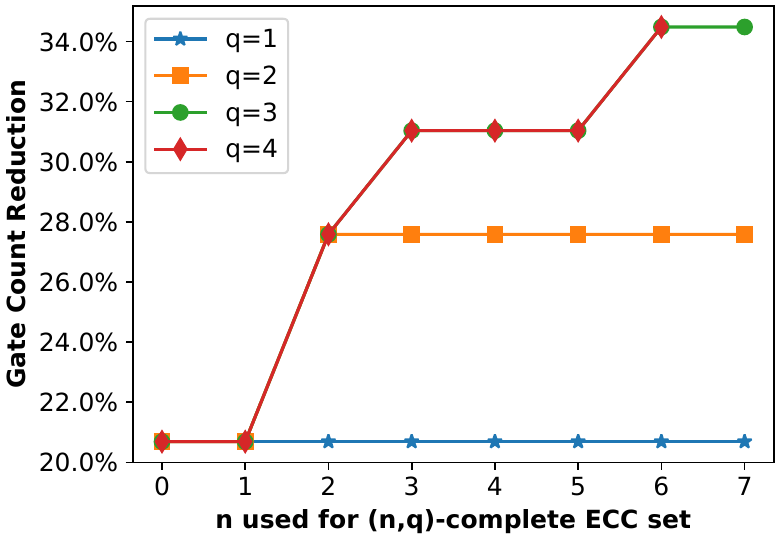}
\hfill
\includegraphics[width=0.48\linewidth]{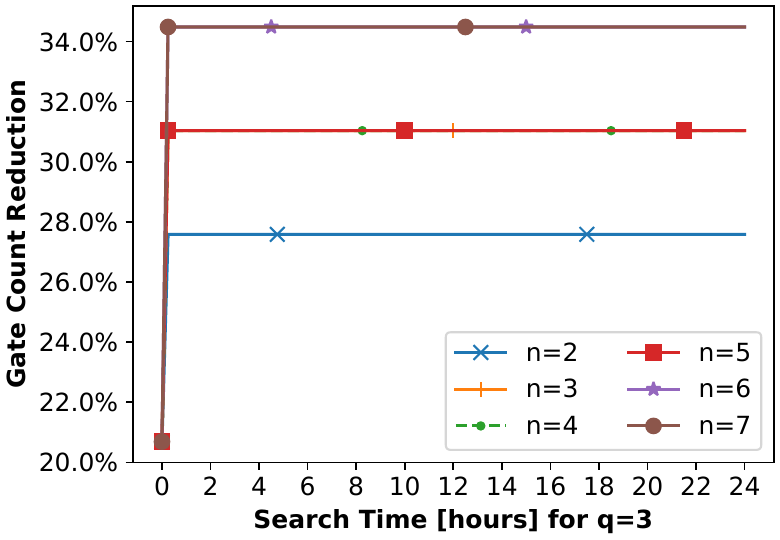}
\caption{\texttt{barenco\_tof\_3} (58 gates).\vspace{1.5em}}
\label{fig:barenco_tof_3}
\end{figure*}

\begin{figure*}
\centering
\includegraphics[width=0.48\linewidth]{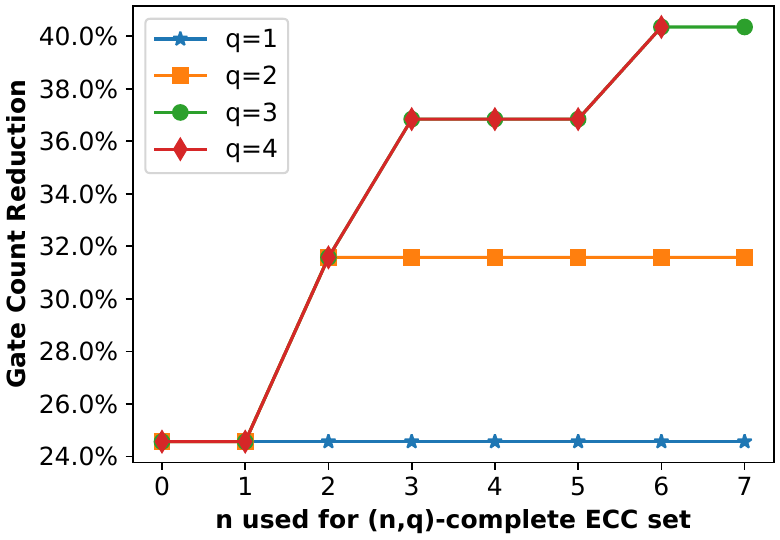}
\hfill
\includegraphics[width=0.48\linewidth]{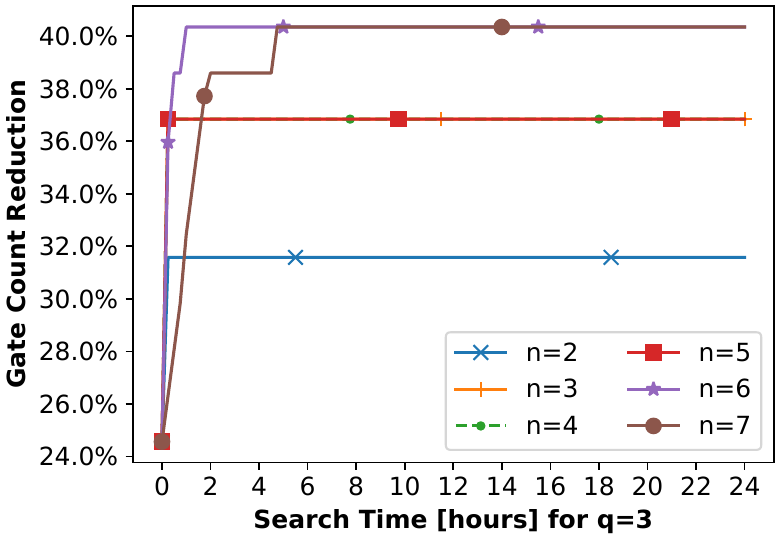}
\caption{\texttt{barenco\_tof\_4} (114 gates).\vspace{1.5em}}
\label{fig:barenco_tof_4}
\end{figure*}

\begin{figure*}
\centering
\includegraphics[width=0.48\linewidth]{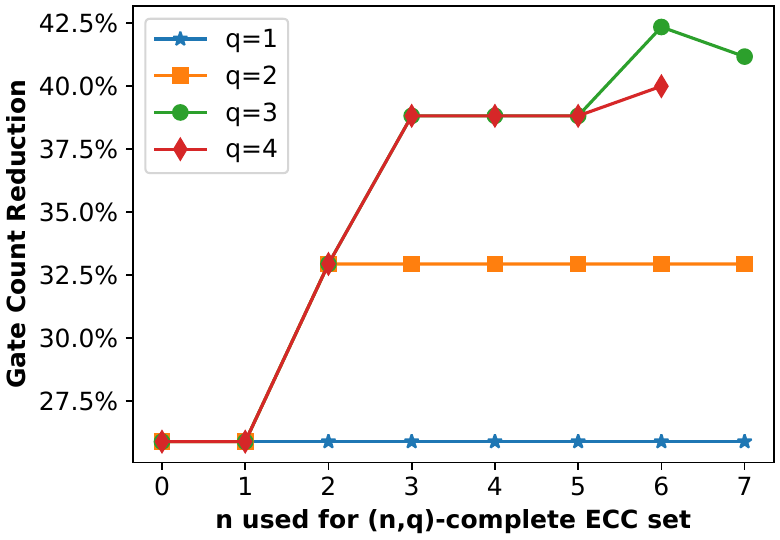}
\hfill
\includegraphics[width=0.48\linewidth]{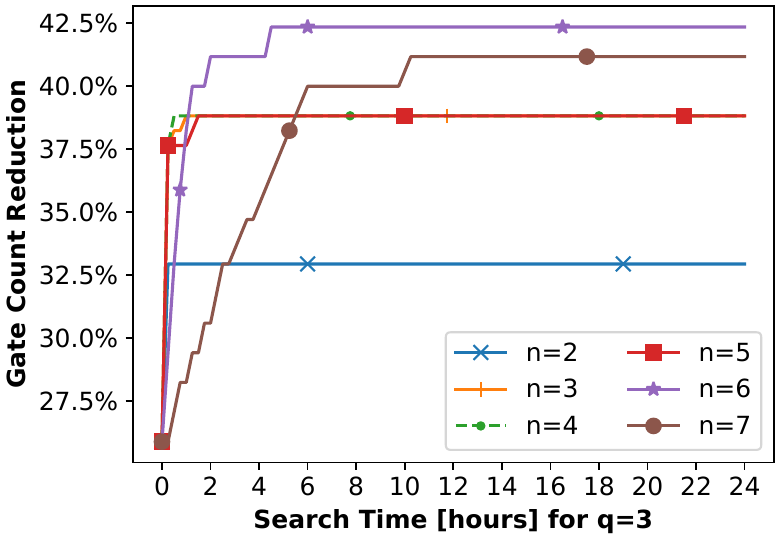}
\caption{\texttt{barenco\_tof\_5} (170 gates).\vspace{1.5em}}
\label{fig:barenco_tof_5}
\end{figure*}

\begin{figure*}
\centering
\includegraphics[width=0.48\linewidth]{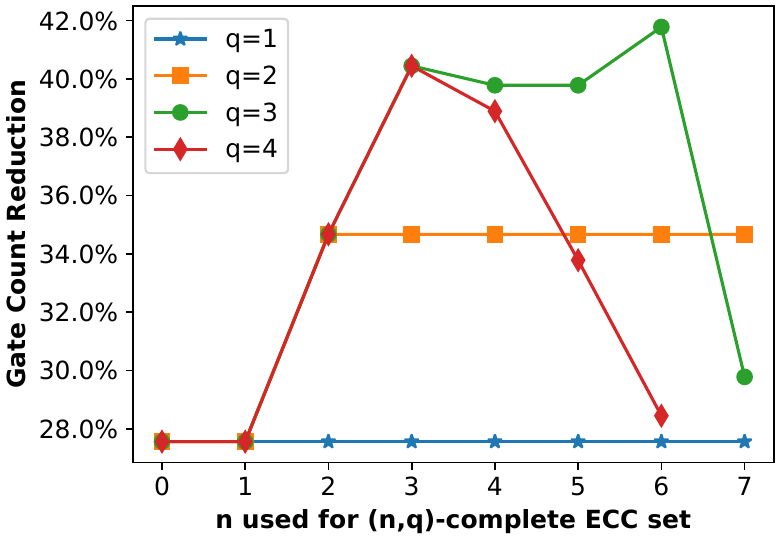}
\hfill
\includegraphics[width=0.48\linewidth]{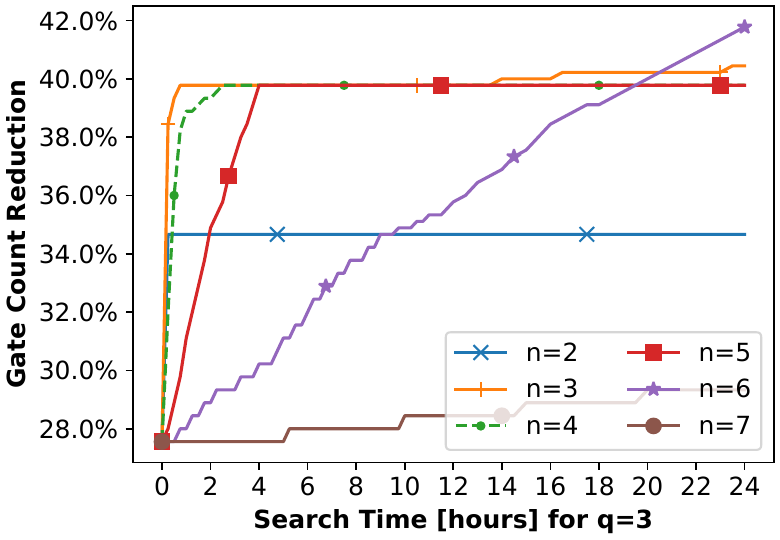}
\caption{\texttt{barenco\_tof\_10} (450 gates).\vspace{1.5em}}
\label{fig:barenco_tof_10}
\end{figure*}

\begin{figure*}
\centering
\includegraphics[width=0.48\linewidth]{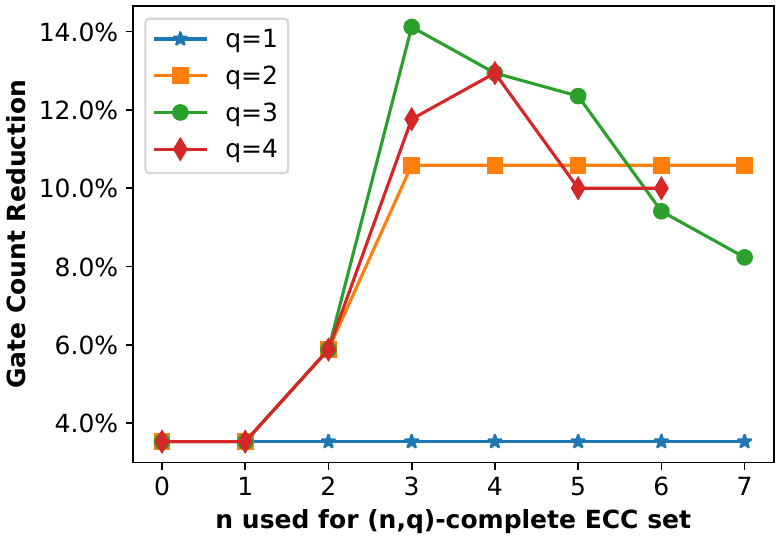}
\hfill
\includegraphics[width=0.48\linewidth]{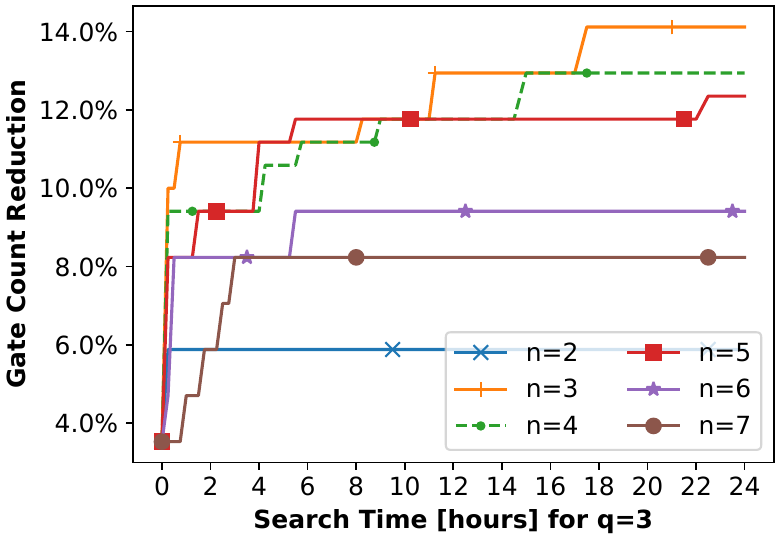}
\caption{\texttt{csla\_mux\_3} (170 gates).\vspace{1.5em}}
\label{fig:csla_mux_3}
\end{figure*}

\begin{figure*}
\centering
\includegraphics[width=0.48\linewidth]{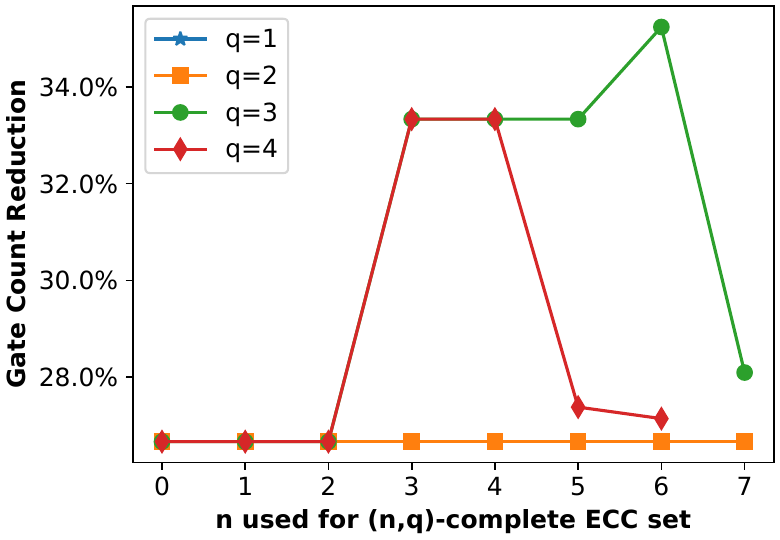}
\hfill
\includegraphics[width=0.48\linewidth]{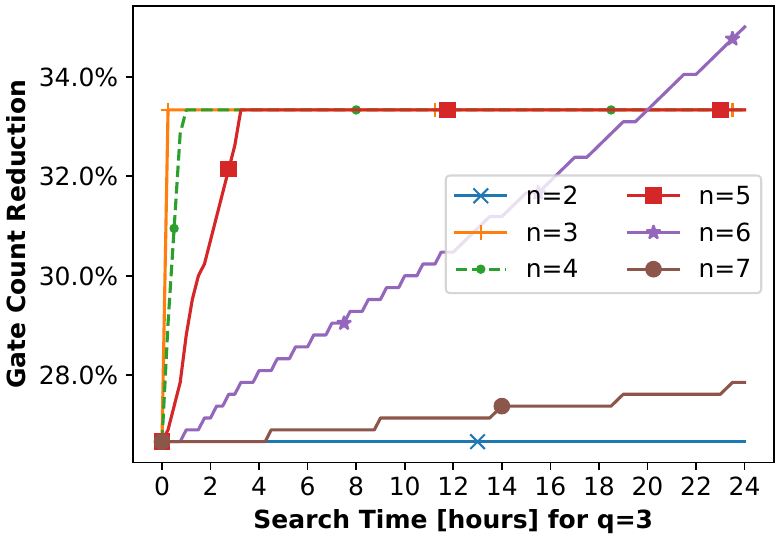}
\caption{\texttt{csum\_mux\_9} (420 gates).\vspace{1.5em}}
\label{fig:csum_mux_9}
\end{figure*}

\begin{figure*}
\centering
\includegraphics[width=0.48\linewidth]{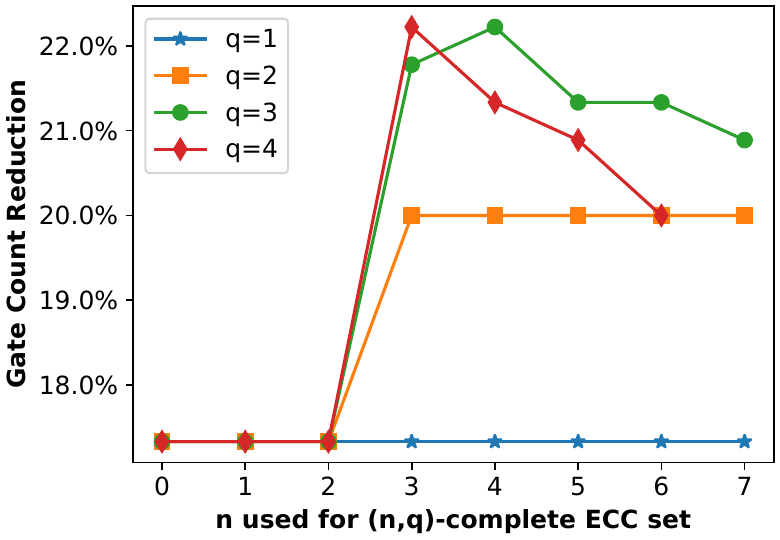}
\hfill
\includegraphics[width=0.48\linewidth]{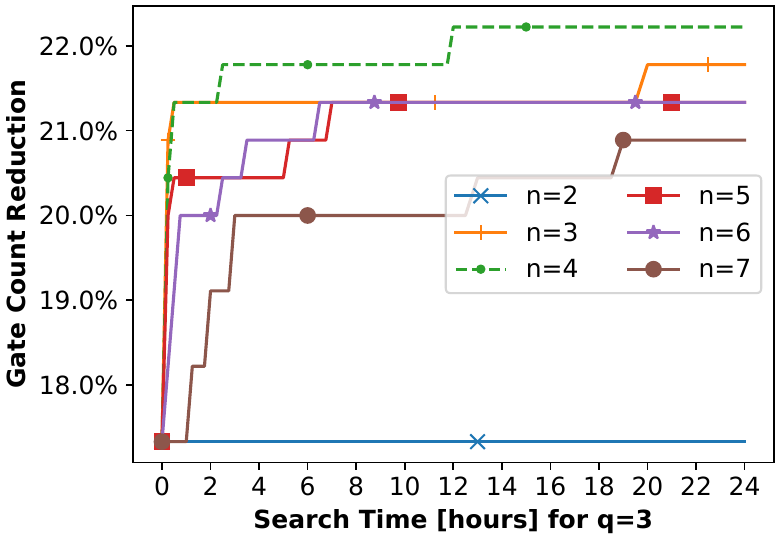}
\caption{\texttt{gf2\^{}4\_mult} (225 gates).\vspace{1.5em}}
\label{fig:gf2^4_mult}
\end{figure*}

\begin{figure*}
\centering
\includegraphics[width=0.48\linewidth]{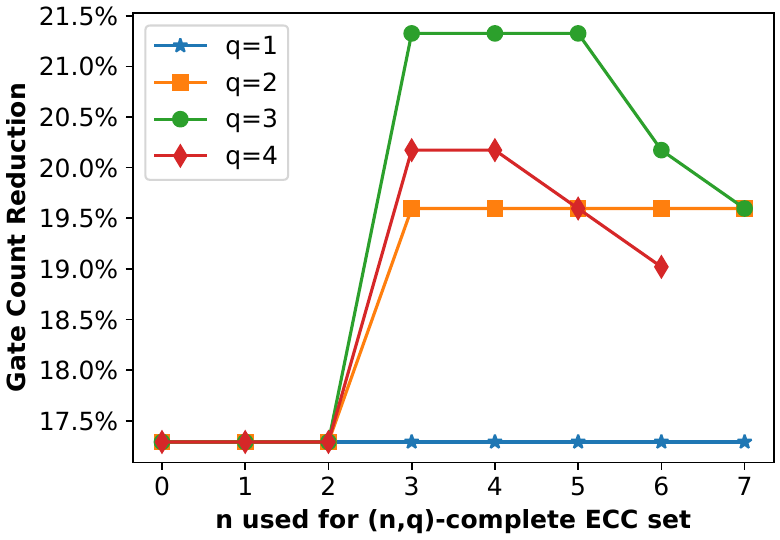}
\hfill
\includegraphics[width=0.48\linewidth]{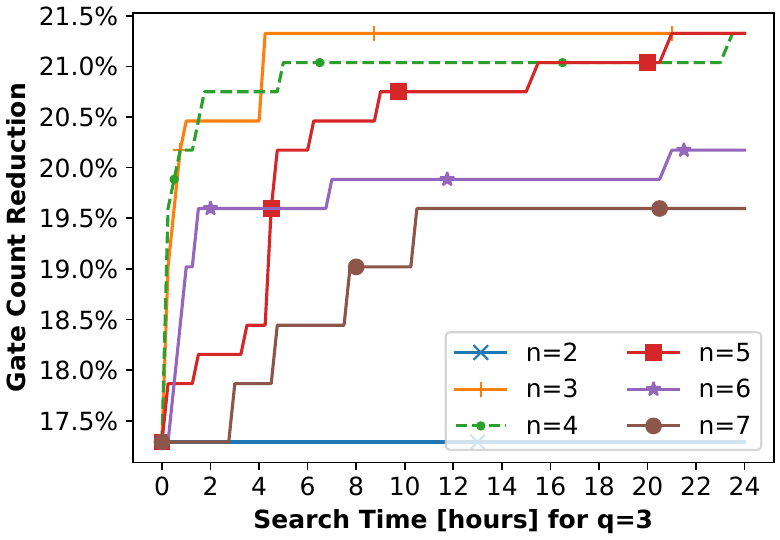}
\caption{\texttt{gf2\^{}5\_mult} (347 gates).\vspace{1.5em}}
\label{fig:gf2^5_mult}
\end{figure*}

\begin{figure*}
\centering
\includegraphics[width=0.48\linewidth]{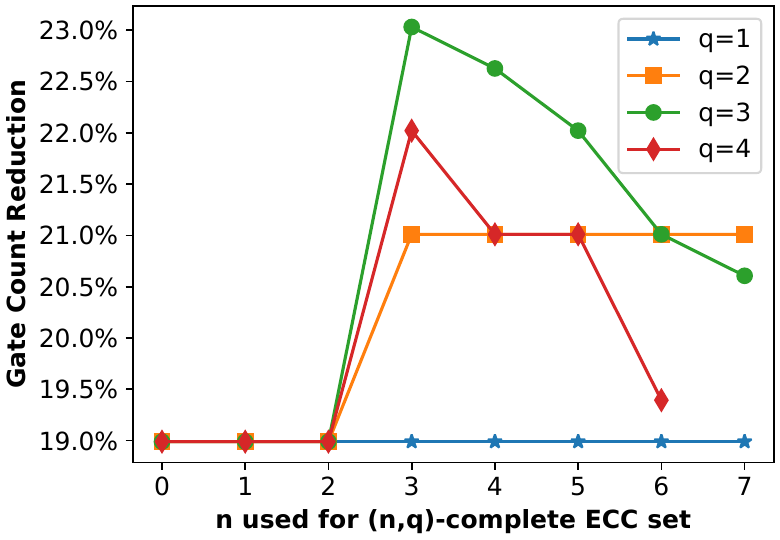}
\hfill
\includegraphics[width=0.48\linewidth]{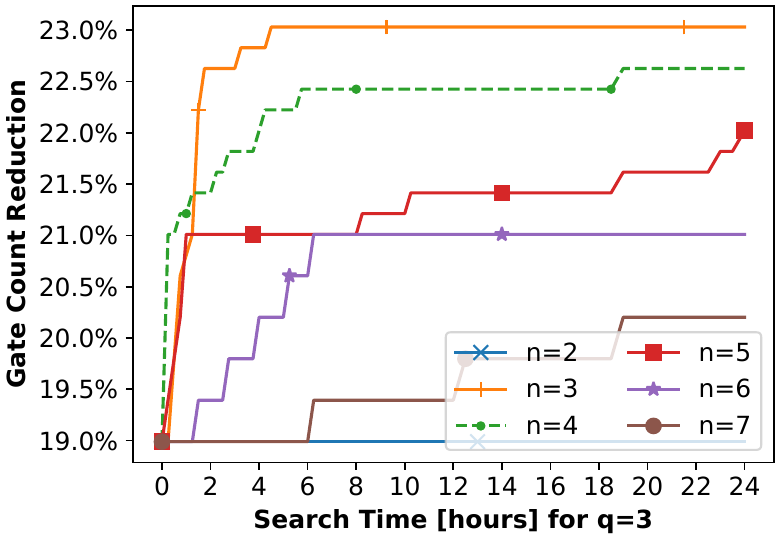}
\caption{\texttt{gf2\^{}6\_mult} (495 gates).\vspace{1.5em}}
\label{fig:gf2^6_mult}
\end{figure*}

\begin{figure*}
\centering
\includegraphics[width=0.48\linewidth]{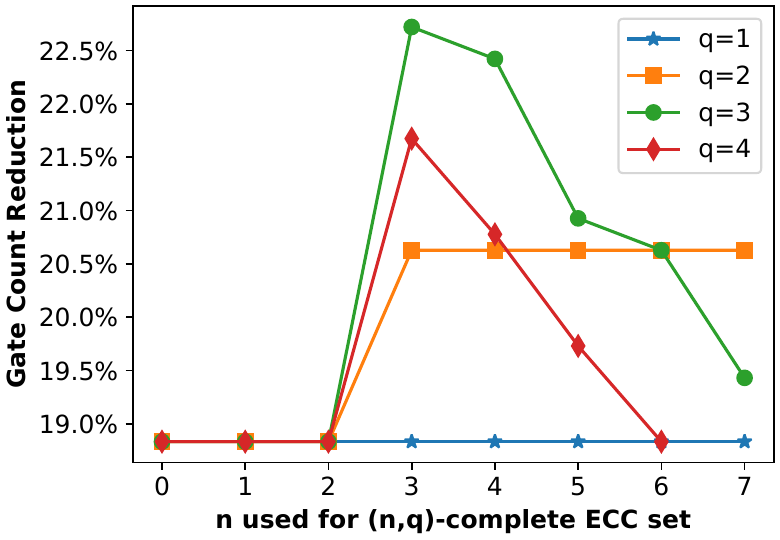}
\hfill
\includegraphics[width=0.48\linewidth]{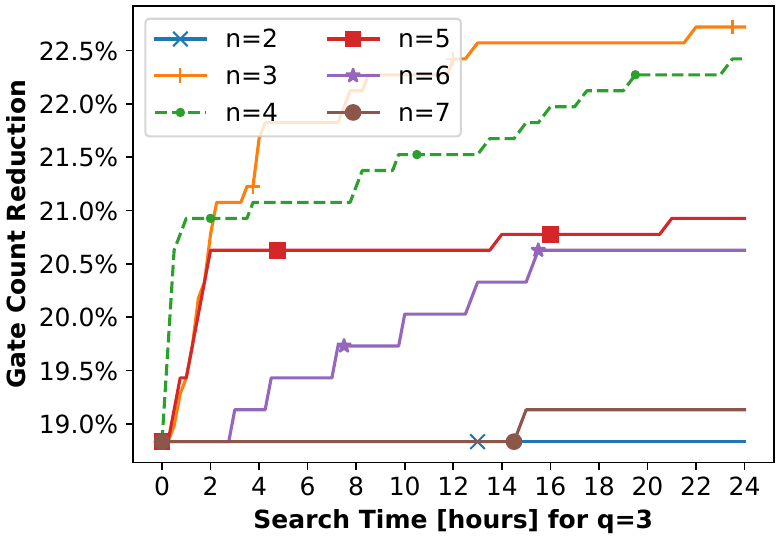}
\caption{\texttt{gf2\^{}7\_mult} (669 gates).\vspace{1.5em}}
\label{fig:gf2^7_mult}
\end{figure*}

\begin{figure*}
\centering
\includegraphics[width=0.48\linewidth]{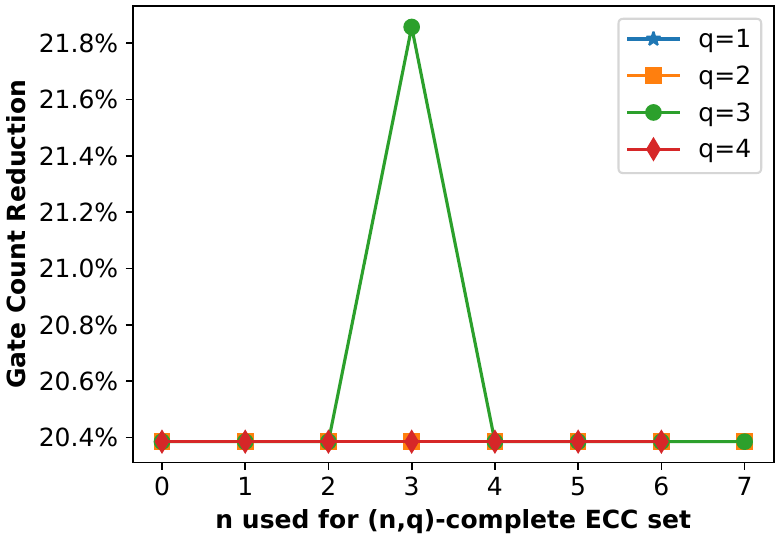}
\hfill
\includegraphics[width=0.48\linewidth]{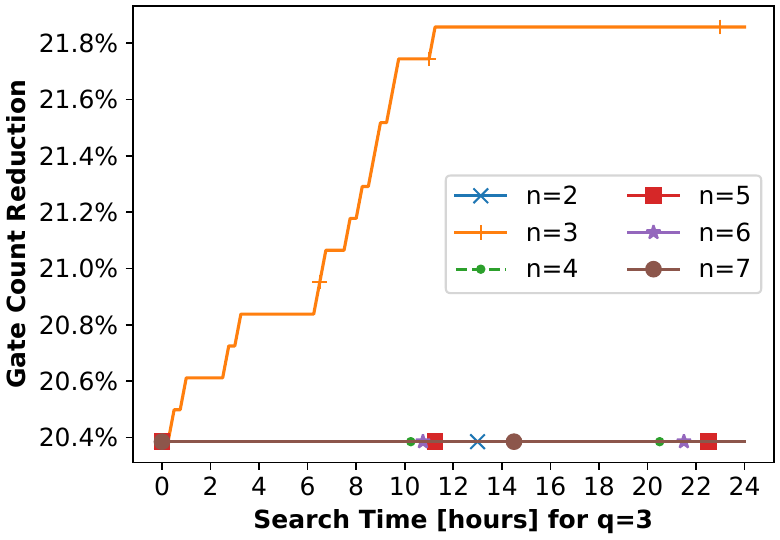}
\caption{\texttt{gf2\^{}8\_mult} (883 gates).\vspace{1.5em}}
\label{fig:gf2^8_mult}
\end{figure*}

\begin{figure*}
\centering
\includegraphics[width=0.48\linewidth]{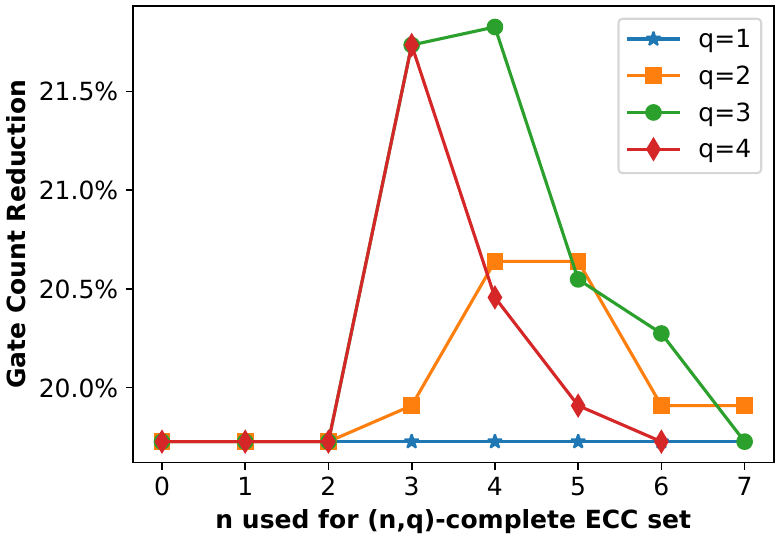}
\hfill
\includegraphics[width=0.48\linewidth]{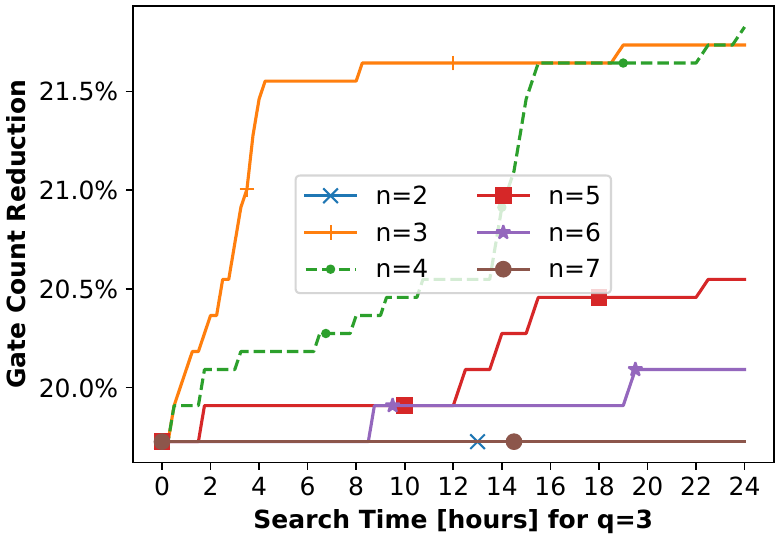}
\caption{\texttt{gf2\^{}9\_mult} (1095 gates).\vspace{1.5em}}
\label{fig:gf2^9_mult}
\end{figure*}

\begin{figure*}
\centering
\includegraphics[width=0.48\linewidth]{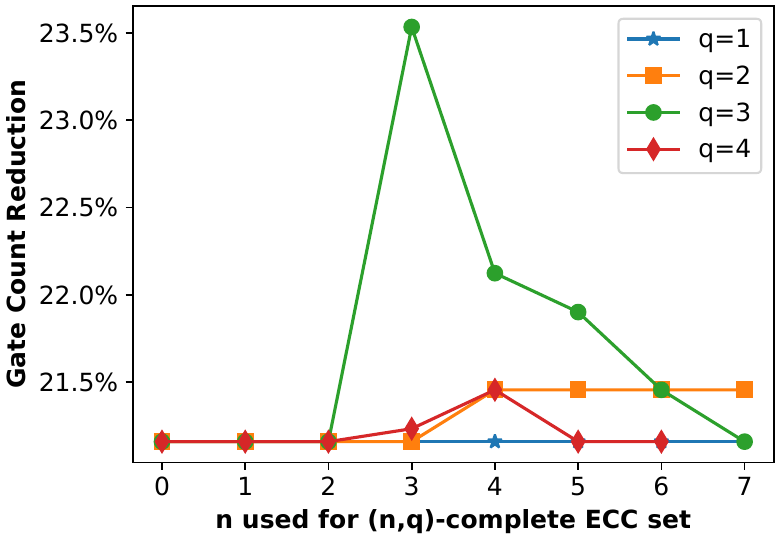}
\hfill
\includegraphics[width=0.48\linewidth]{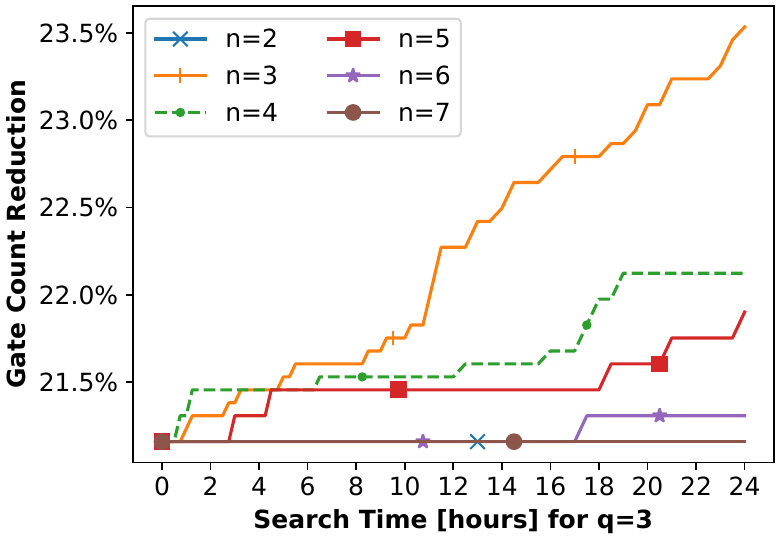}
\caption{\texttt{gf2\^{}10\_mult} (1347 gates).\vspace{1.5em}}
\label{fig:gf2^10_mult}
\end{figure*}

\begin{figure*}
\centering
\includegraphics[width=0.48\linewidth]{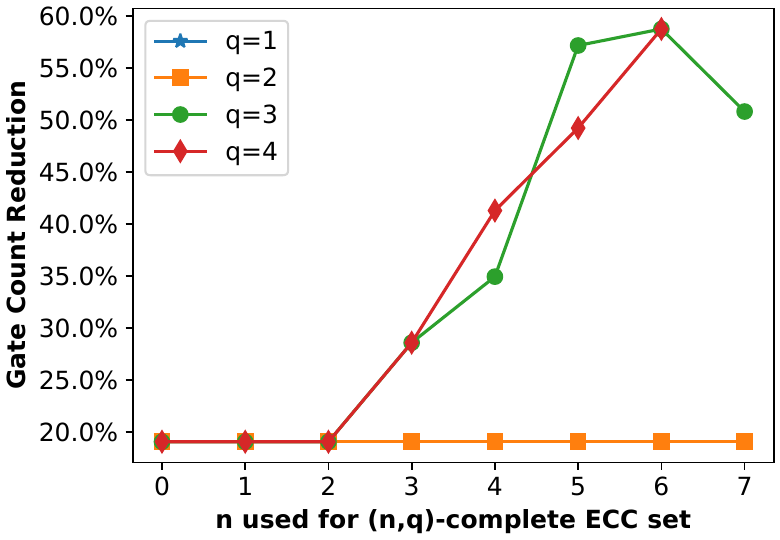}
\hfill
\includegraphics[width=0.48\linewidth]{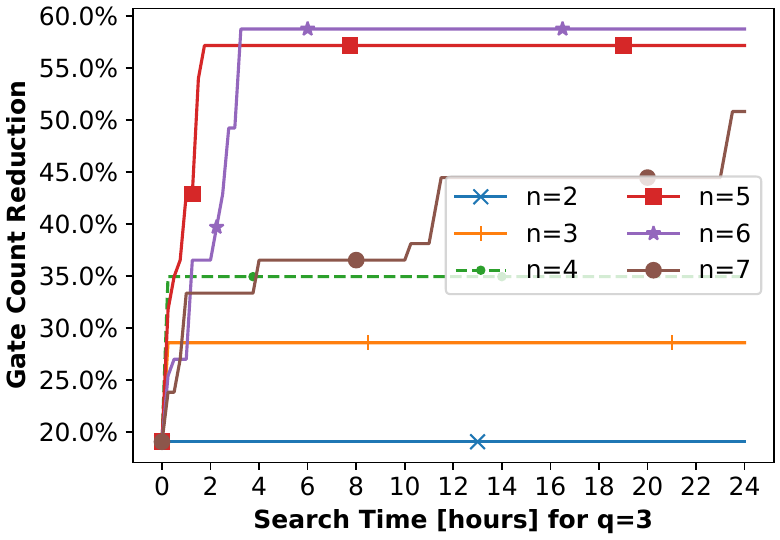}
\caption{\texttt{mod5\_4} (63 gates).\vspace{1.5em}}
\label{fig:mod5_4}
\end{figure*}

\begin{figure*}
\centering
\includegraphics[width=0.48\linewidth]{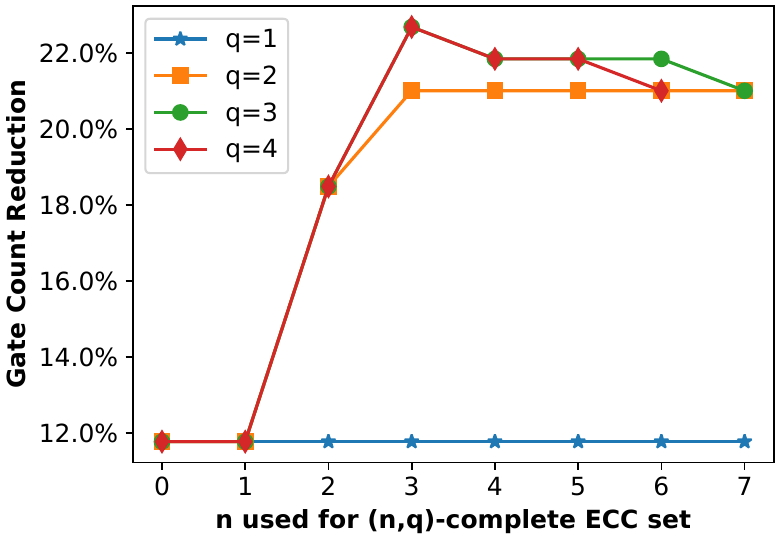}
\hfill
\includegraphics[width=0.48\linewidth]{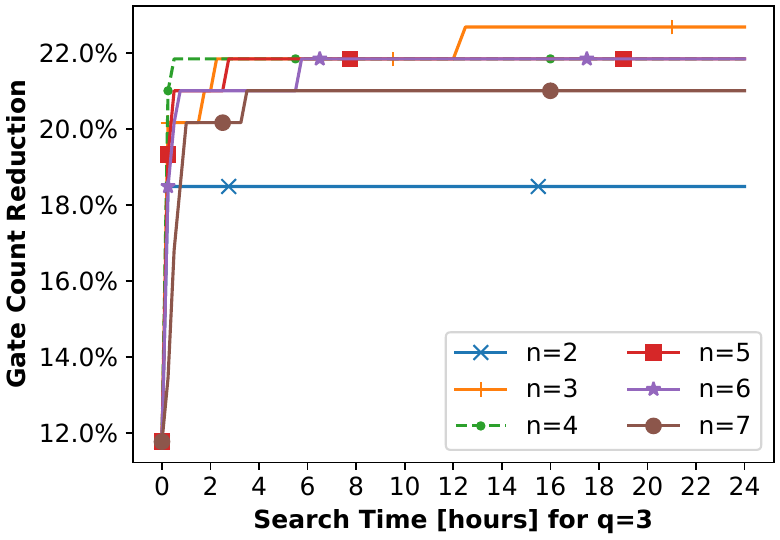}
\caption{\texttt{mod\_mult\_55} (119 gates).\vspace{1.5em}}
\label{fig:mod_mult_55}
\end{figure*}

\begin{figure*}
\centering
\includegraphics[width=0.48\linewidth]{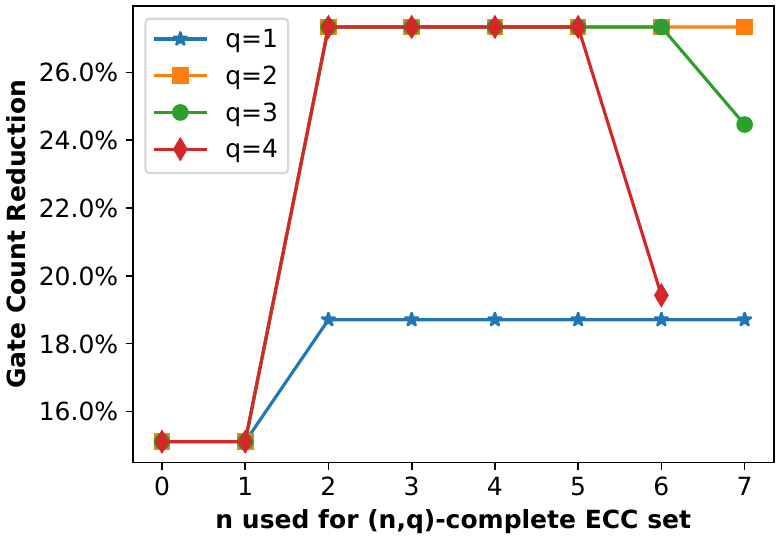}
\hfill
\includegraphics[width=0.48\linewidth]{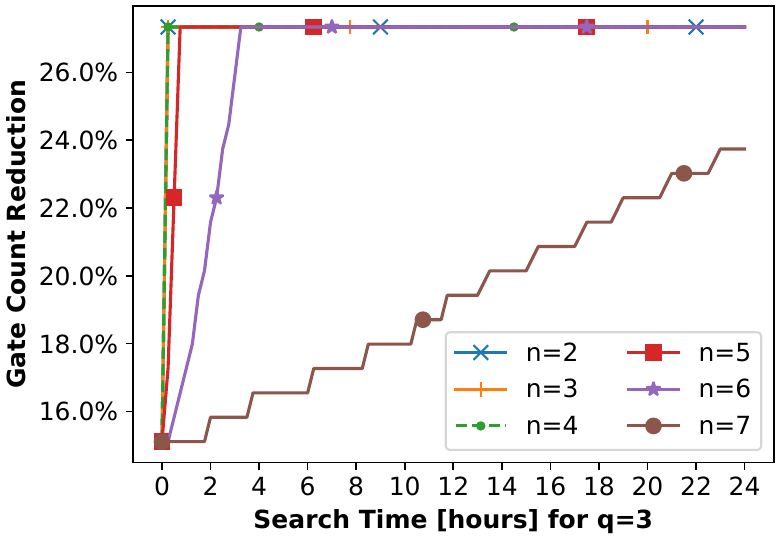}
\caption{\texttt{mod\_red\_21} (278 gates).\vspace{1.5em}}
\label{fig:mod_red_21}
\end{figure*}

\begin{figure*}
\centering
\includegraphics[width=0.48\linewidth]{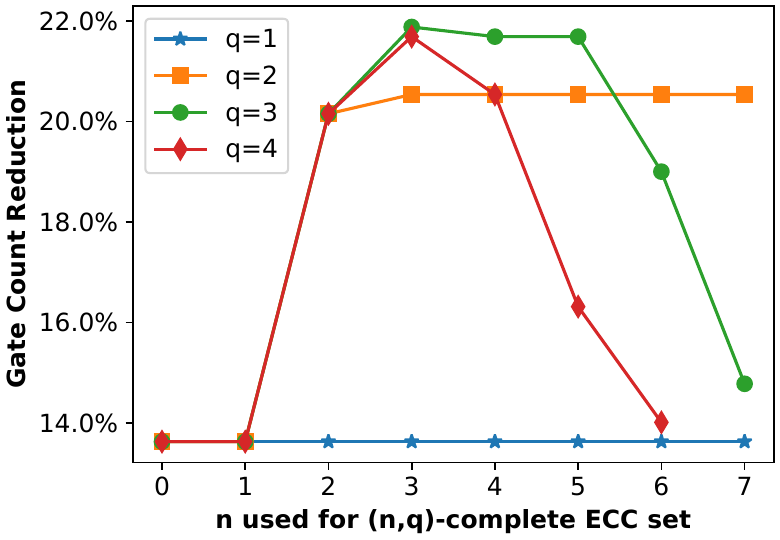}
\hfill
\includegraphics[width=0.48\linewidth]{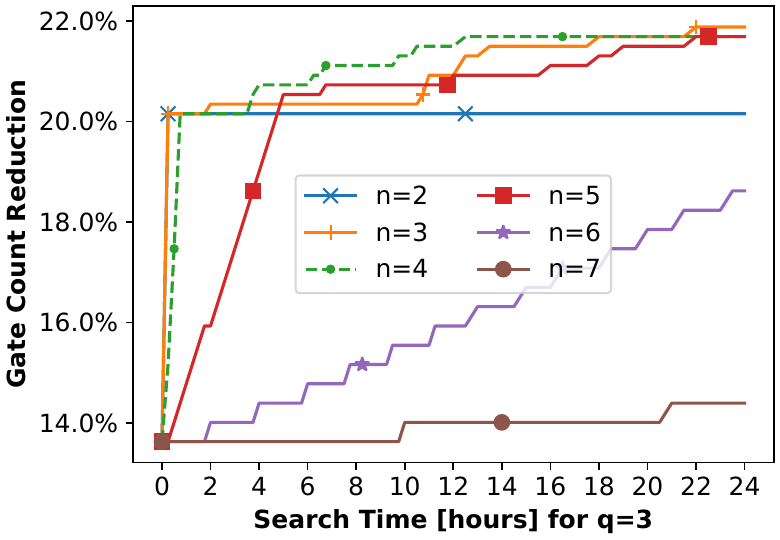}
\caption{\texttt{qcla\_adder\_10} (521 gates).\vspace{1.5em}}
\label{fig:qcla_adder_10}
\end{figure*}

\begin{figure*}
\centering
\includegraphics[width=0.48\linewidth]{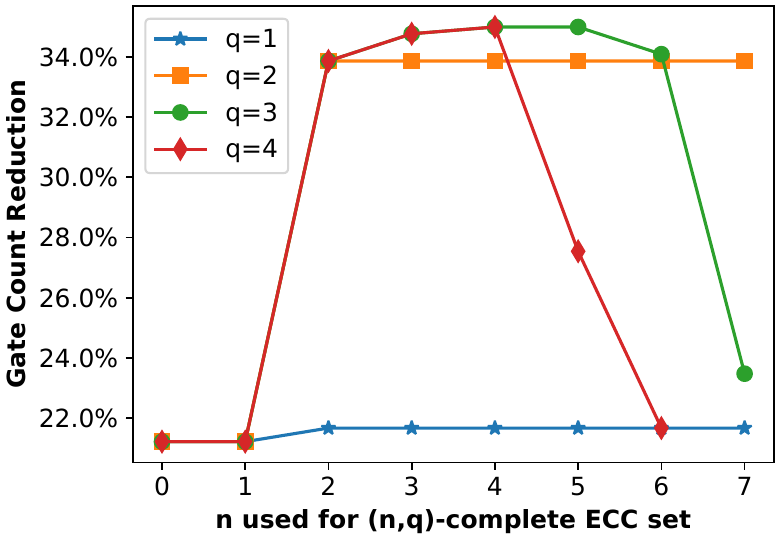}
\hfill
\includegraphics[width=0.48\linewidth]{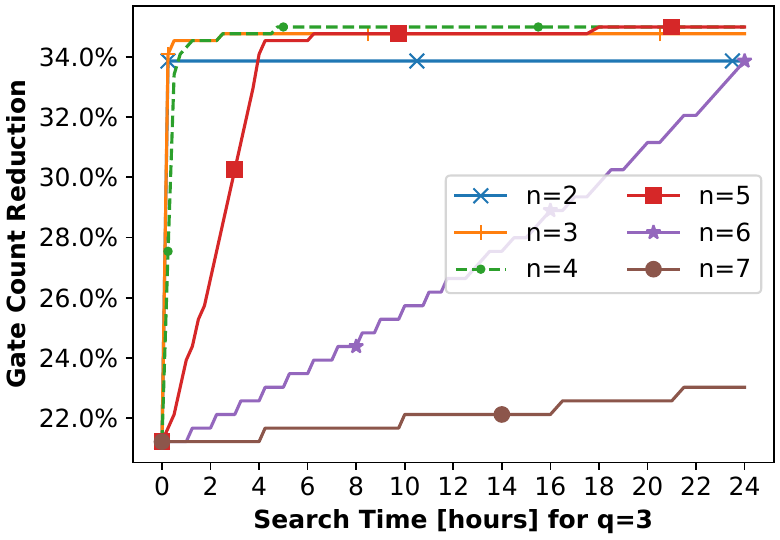}
\caption{\texttt{qcla\_com\_7} (443 gates).\vspace{1.5em}}
\label{fig:qcla_com_7}
\end{figure*}

\begin{figure*}
\centering
\includegraphics[width=0.48\linewidth]{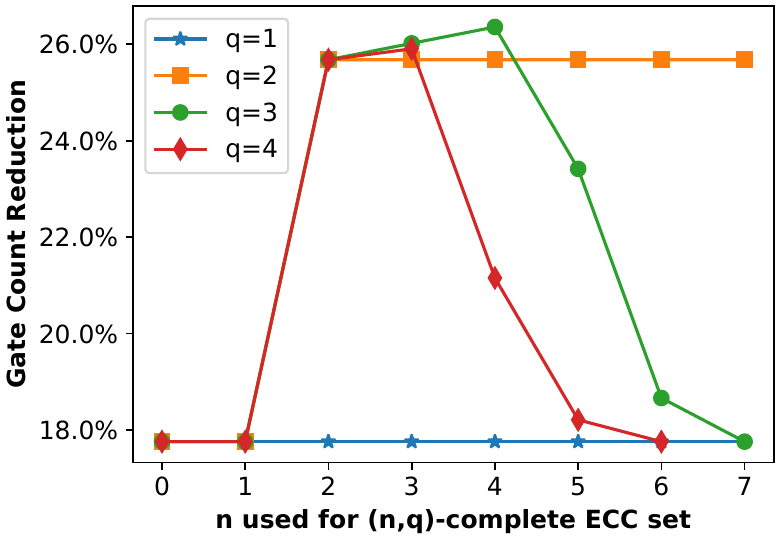}
\hfill
\includegraphics[width=0.48\linewidth]{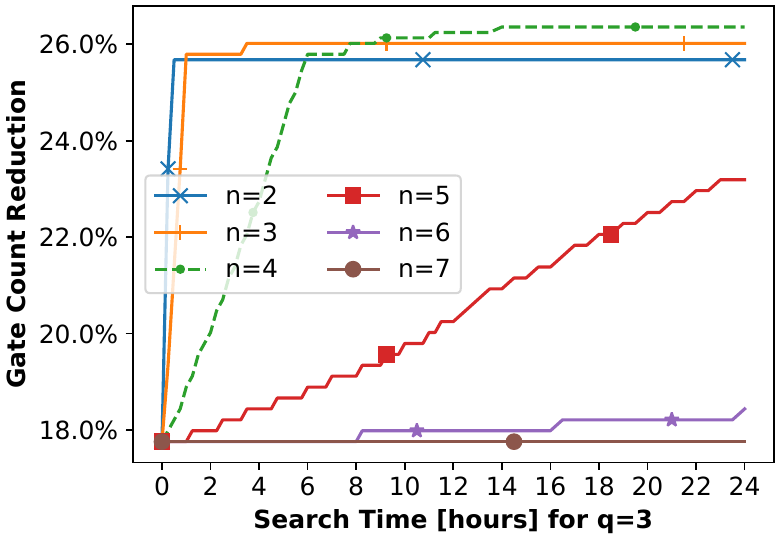}
\caption{\texttt{qcla\_mod\_7} (884 gates).\vspace{1.5em}}
\label{fig:qcla_mod_7}
\end{figure*}

\begin{figure*}
\centering
\includegraphics[width=0.48\linewidth]{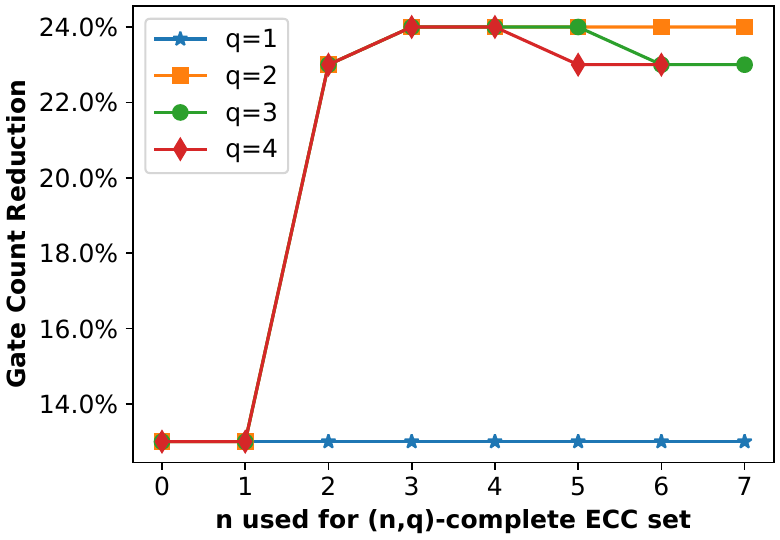}
\hfill
\includegraphics[width=0.48\linewidth]{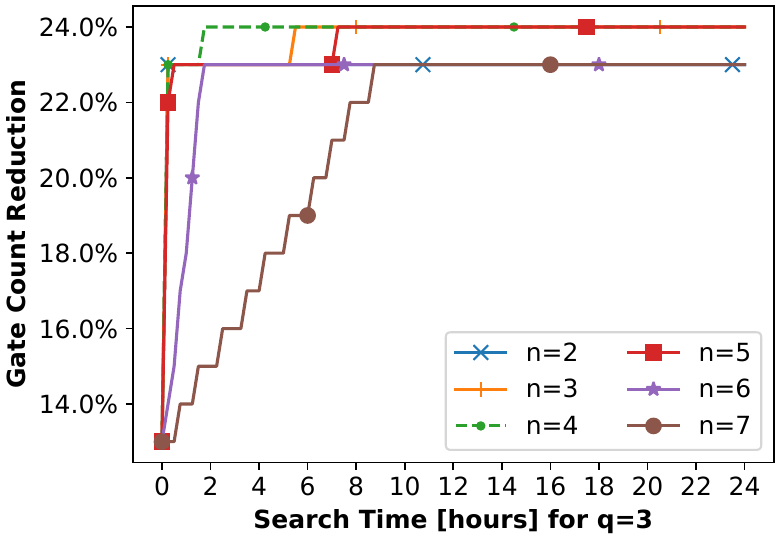}
\caption{\texttt{rc\_adder\_6} (200 gates).\vspace{1.5em}}
\label{fig:rc_adder_6}
\end{figure*}

\begin{figure*}
\centering
\includegraphics[width=0.48\linewidth]{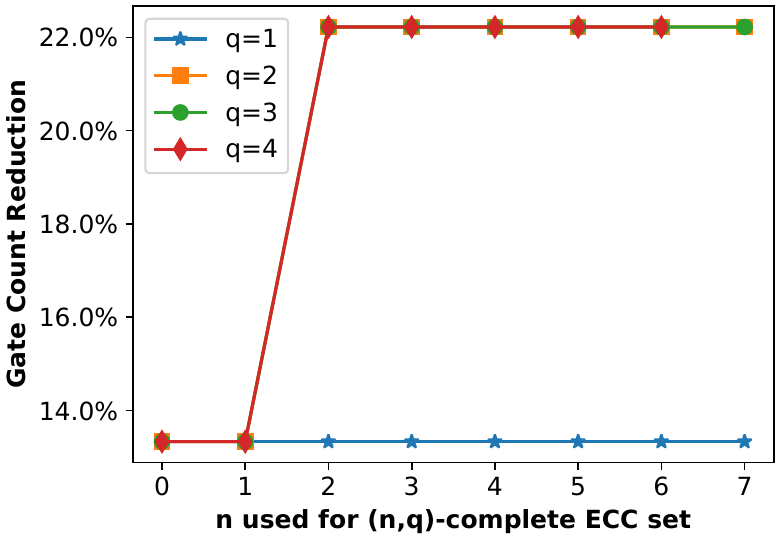}
\hfill
\includegraphics[width=0.48\linewidth]{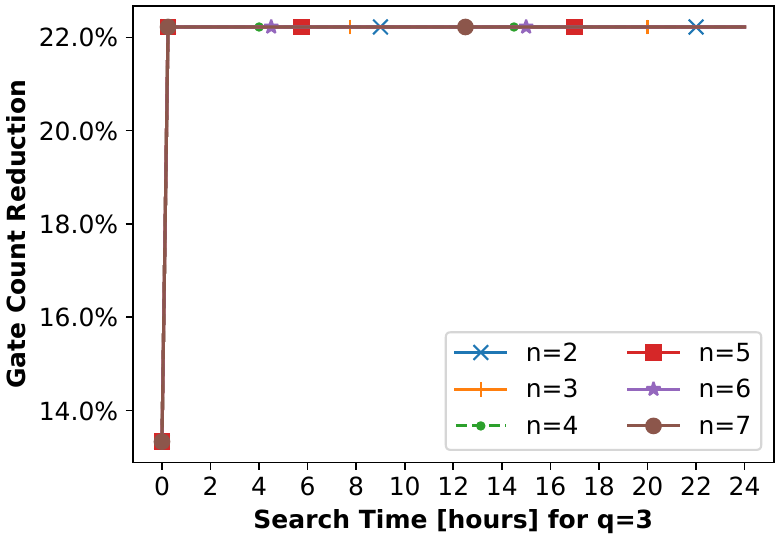}
\caption{\texttt{tof\_3} (45 gates).\vspace{1.5em}}
\label{fig:tof_3}
\end{figure*}

\begin{figure*}
\centering
\includegraphics[width=0.48\linewidth]{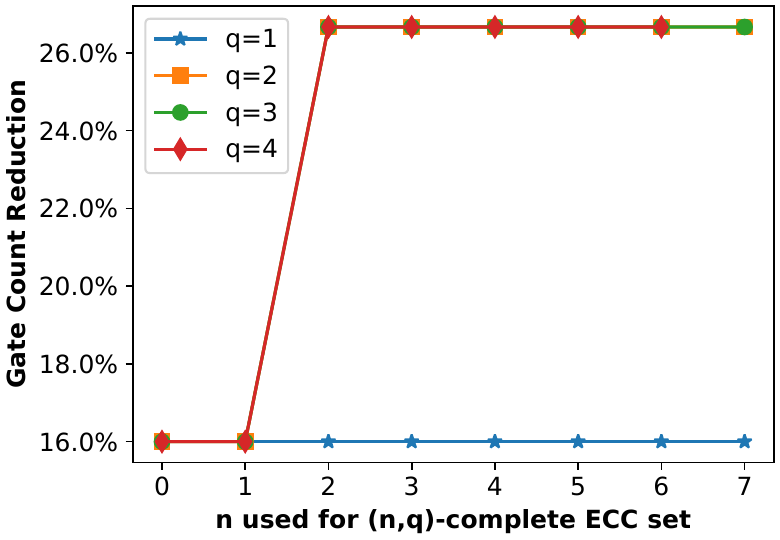}
\hfill
\includegraphics[width=0.48\linewidth]{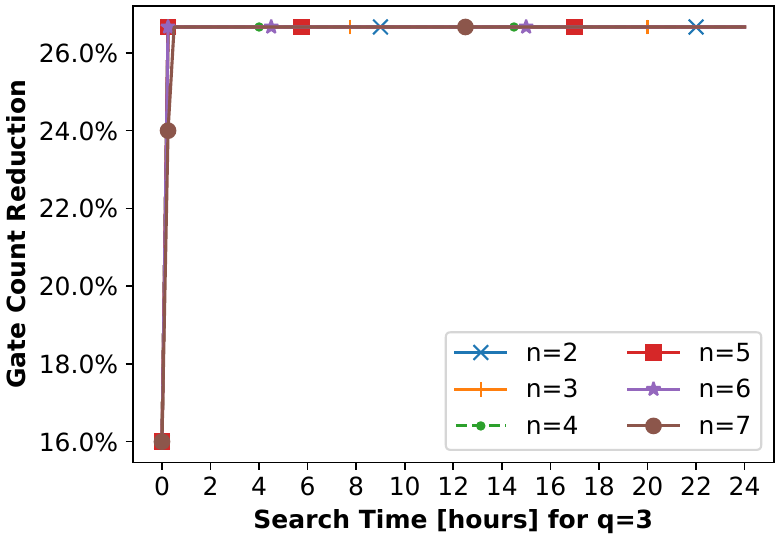}
\caption{\texttt{tof\_4} (75 gates).\vspace{1.5em}}
\label{fig:tof_4}
\end{figure*}

\begin{figure*}
\centering
\includegraphics[width=0.48\linewidth]{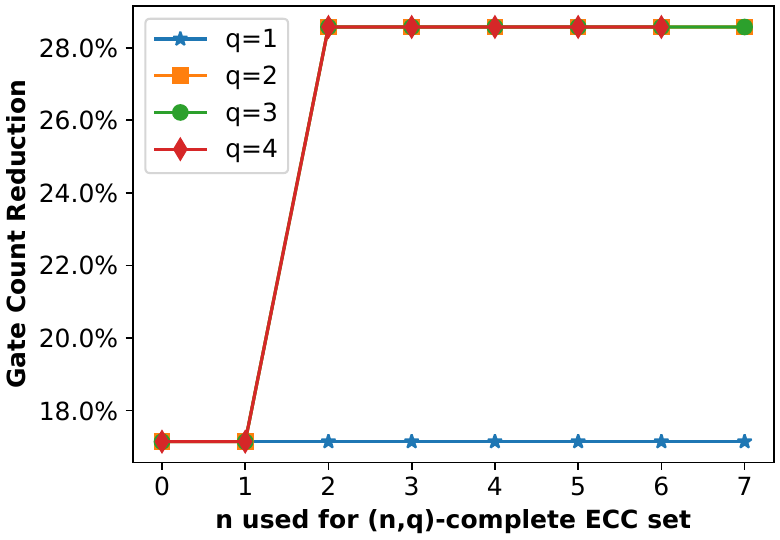}
\hfill
\includegraphics[width=0.48\linewidth]{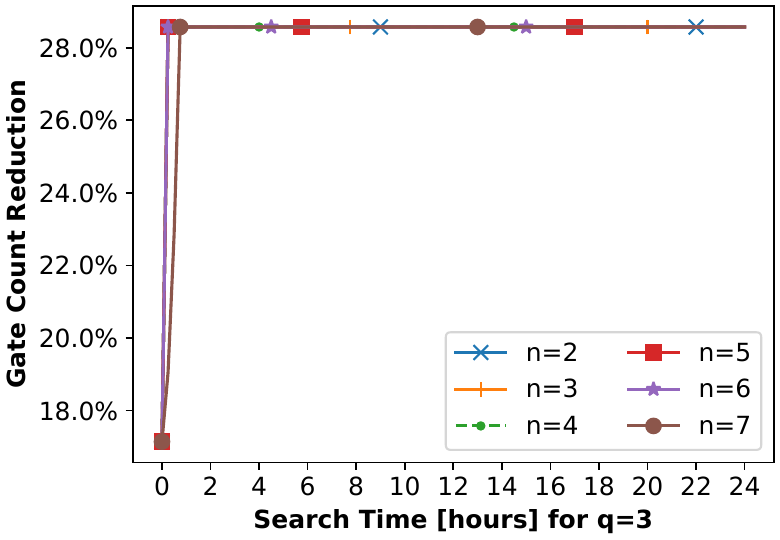}
\caption{\texttt{tof\_5} (105 gates).\vspace{1.5em}}
\label{fig:tof_5}
\end{figure*}

\begin{figure*}
\centering
\includegraphics[width=0.48\linewidth]{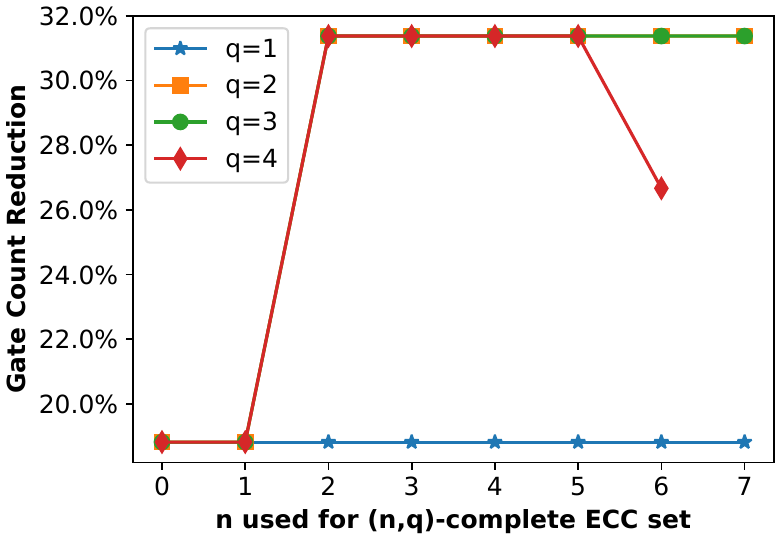}
\hfill
\includegraphics[width=0.48\linewidth]{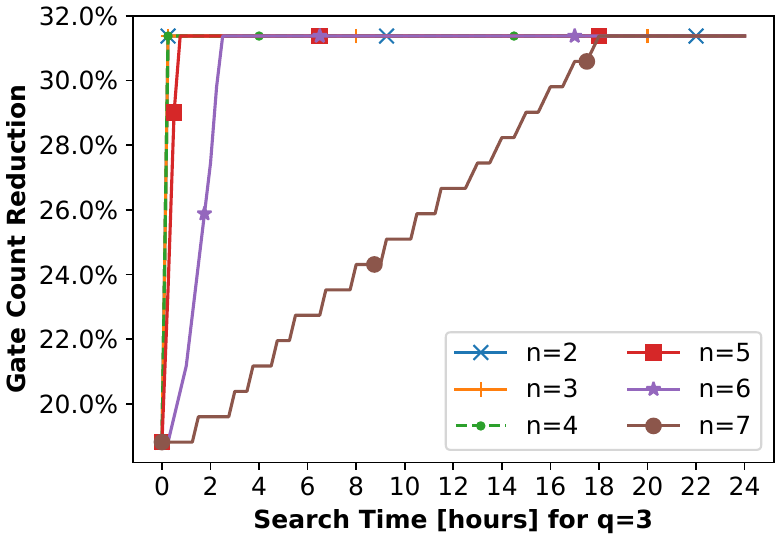}
\caption{\texttt{tof\_10} (255 gates).\vspace{1.5em}}
\label{fig:tof_10}
\end{figure*}

\begin{figure*}
\centering
\includegraphics[width=0.48\linewidth]{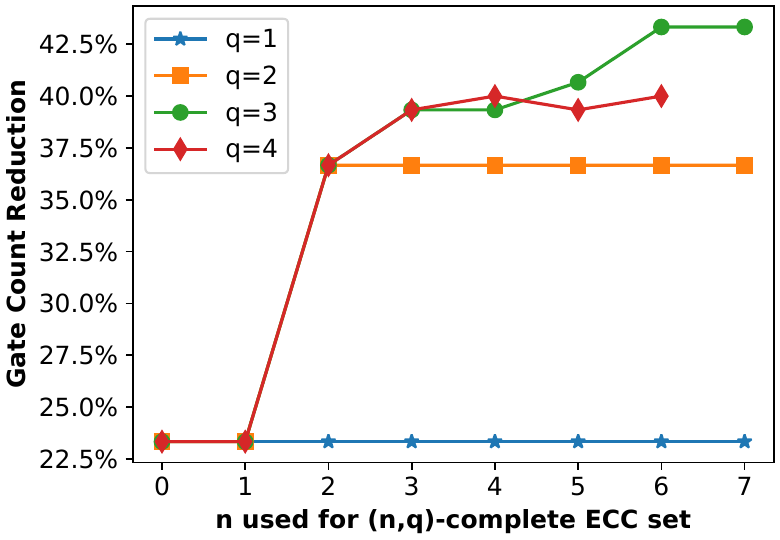}
\hfill
\includegraphics[width=0.48\linewidth]{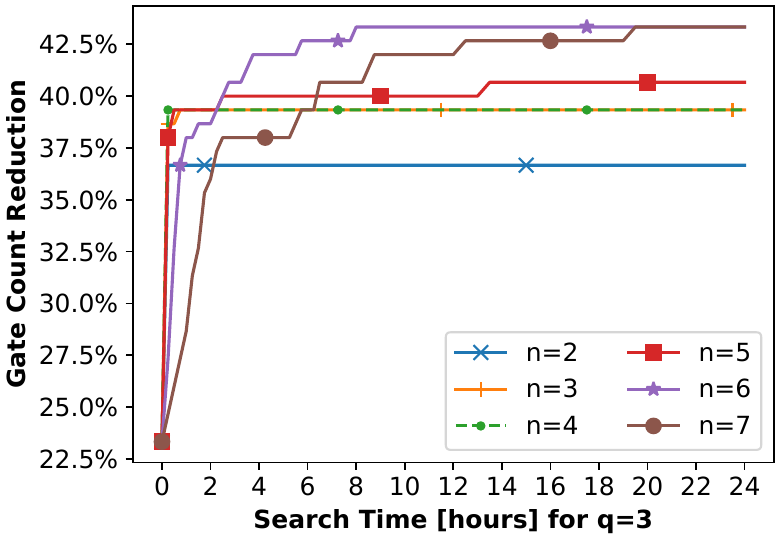}
\caption{\texttt{vbe\_adder\_3} (150 gates).\vspace{1.5em}}
\label{fig:vbe_adder_3}
\end{figure*}

\begin{figure*}
\centering
\includegraphics[width=\linewidth]{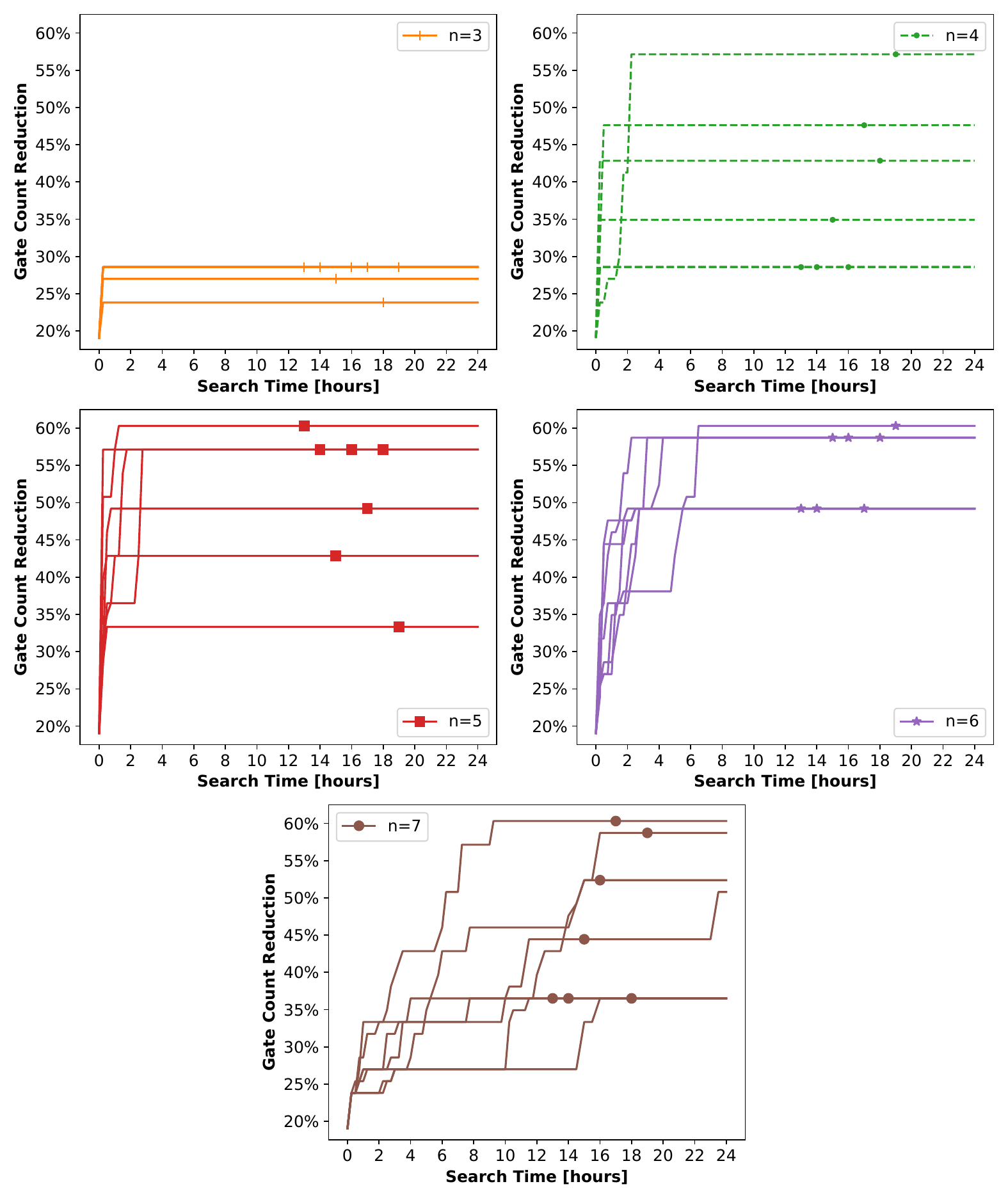}
\caption{7 runs of \texttt{mod5\_4} with a $(n, 3)$-complete ECC set for each $3 \le n \le 7$. Each marker denotes one run. For example, the three markers on the 57.1\%-reduction line of $n=5$ show three runs resulting in 27 gates after 24 hours.}
\label{fig:mod54:full}
\end{figure*}

\fi

\end{document}